\documentclass[a4paper,11pt]{amsart}
\usepackage{easybmat}
\usepackage{amsmath,amssymb,amsthm}
\usepackage[T1]{fontenc}
\usepackage{bbm}
\usepackage{color}

\newtheorem{theorem}{Theorem}[section]
\newtheorem{lemma}[theorem]{Lemma}
\newtheorem{definition}[theorem]{Definition}
\newtheorem{proposition}[theorem]{Proposition}
\newtheorem{remark}[theorem]{Remark}
\newtheorem{corollary}[theorem]{Corollary}

\renewcommand\H{\ensuremath{\mathcal{H}}}
\newcommand\Z{\ensuremath{\mathbbm{Z}}}
\newcommand\ltwo[1][\mathbbm{Z}]{\ensuremath{\ell_2(#1)}}
\newcommand\C{\ensuremath{\mathbbm C}}
\newcommand\K{\ensuremath{\mathbbm K}}
\newcommand\La{\ensuremath{\mathcal{L}}}
\newcommand\B{\ensuremath{\mathcal{B}}}
\newcommand\R{\ensuremath{\mathbbm R}}
\newcommand\UG[1][2]{\ensuremath{\mathcal{U}(#1)}}
\newcommand\N{\ensuremath{\mathbbm N}}
\newcommand\UC{\ensuremath{\mathbb{T}}}
\newcommand\T{\ensuremath{\mathbb{T}}}
\newcommand\f{\ensuremath{\mathcal{L}}}

\newcommand\Id{\ensuremath{\mathbbm{1}}}
\newcommand\W{\ensuremath{W_\omega}}
\newcommand\Wf{\ensuremath{W_\omega (N)}}
\newcommand\rhof[1][x,y]{\ensuremath{\rho^{#1}_{\omega,N}}}

\newcommand\U{\ensuremath{U_\omega}}
\renewcommand\S{\ensuremath{S}}
\def\norm #1{\Vert #1\Vert}

\newcommand\Dens{\ensuremath{\mathcal{N}}}

\DeclareMathOperator{\DummyRes}{G}
\newcommand\Res[1][z]{\ensuremath{\DummyRes_{#1}}}
\newcommand\Resf[1][z]{\ensuremath{\DummyRes^N_{#1}}}

\newcommand\spec[1]{\ensuremath{\sigma (#1)}}

\newcommand\Scp[2]{\ensuremath{\, \langle #1 \,, #2 \,\rangle}}

\newcommand\Loc[1][x]{\ensuremath{\delta_{#1}}}
\newcommand\LocSe[2][x]{\ensuremath{\Loc[#1]\otimes e_#2}}
\newcommand\LocS[2][x]{\ensuremath{\Loc[#1]\otimes #2}}

\newcommand\GenEVP[1][n]{\ensuremath{\phi_+ {\scriptstyle (#1)}}}
\newcommand\GenEVM[1][n]{\ensuremath{\phi_- {\scriptstyle (#1)}}}

\newcommand\ResColArg[2]{\ensuremath{\gamma^#1_{#2}}}
\newcommand\ResCol[1]{\ensuremath{\gamma^{#1}}}

\newcommand\Gmu[1][\mu]{\ensuremath{\langle #1 \rangle}}
\newcommand\Gmuz[1][\mu]{\Gmu[\mu_z]}
\newcommand\Gmut[1][\mu]{\Gmu[\mu_\theta]}
\newcommand\Det[1]{\ensuremath{\mathrm{det({\mathnormal {#1}})}}}

\newcommand\Norm[2][]{\ensuremath{\left| \left| #2\right| \right|_{#1}}}

\newcommand\Abs[1]{\ensuremath{| #1|}}

\newcommand{\ii}{{\rm i}}

\newcommand\Expect[1]{\ensuremath{{\mathbbm E}\left({\mathnormal {#1}} \right)}}

\DeclareMathOperator{\tr}{tr}

\newcommand\Prob[1]{\ensuremath{{\mathbbm P}\left({\mathnormal {#1}}\right)}}

\newcommand\ex[1]{\ensuremath{e^{#1}}}

\newcommand\Trans[2][z]{\ensuremath{T_{#2}{\scriptstyle (#1)}}}

\newcommand\NPNonInvTrans[2]{\ensuremath{\Norm{T_{#1}{\scriptstyle (#2)} \cdot\ldots\cdot T_{1}{\scriptstyle (#2)} v}}}

\newcommand\NPNonInvProd[2]{\ensuremath{T_{#1}{\scriptstyle (#2)} \cdot\ldots\cdot T_{1}{\scriptstyle (#2)} }}

\newcommand\VecD[2]{\ensuremath{\genfrac{(}{)}{0pt}{0}{#1}{#2}}}

\newcommand\betrag[1]{\ensuremath{|#1|}}

\newcommand\ie{{\it{i.e.}}}
\newcommand\eg{{\it{e.g.}}}

\newcommand\Plane{\ensuremath{P_2}}

\begin{document}

\title{Disordered Quantum Walks in one lattice dimension}
\author{Andre Ahlbrecht, Volkher B. Scholz, Albert H.~Werner}
\address{Institut f\"{u}r Theoretische Physik\\ Leibniz Universit\"{a}t Hannover\\ Appelstr. 2, 30167 Hannover, Germany}

\begin{abstract}
We study a spin-$\frac{1}{2}$-particle moving on a one dimensional lattice subject to disorder induced by a random, space-dependent quantum coin. The discrete time evolution is given by a family of random unitary quantum walk operators, where the shift operation is assumed to be deterministic. Each coin is an independent identically distributed random variable with values in the group of two dimensional unitary matrices. We derive sufficient conditions on the probability distribution of the coins such that the system exhibits dynamical localization. Put differently, the tunneling probability between two lattice sites decays rapidly for almost all choices of random coins and after arbitrary many time steps with increasing distance. Our findings imply that this effect takes place if the coin is chosen at random from the Haar measure, or some measure continuous with respect to it, but also for a class of discrete probability measures which support consists of two coins, one of them being the Hadamard coin.
\end{abstract}

\maketitle

\tableofcontents

\newpage

\section{Introduction}

Classical random walks are of importance for the field of randomized algorithms \cite{Motwani1995}, \eg\ for search algorithms, connectivity and satisfiability problems. The generalization of a classical random walk to the quantum world, called a quantum walk, is in its simplest form given by a spin-$\frac{1}{2}$-particle (qubit) moving on a line, and space as well as time are discrete parameters, see \cite{Kempe2003,Ambainis2001} for reviews.

Hence the Hilbert space of the particle is given by $\H=\ell_2(\mathbbm{Z}) \otimes \mathbbm{C}^2$, where $\ell_2(\mathbbm{Z})$ denotes the Hilbert space of square summable sequences over $\mathbbm{Z}$. Each step in the time evolution is described by the same unitary operator, called the \emph{walk operator}. It is defined as the product of two unitary operators, the first one being a \emph{shift} operation  which acts on the position degree of freedom of the particle. The \emph{shift operator}  $S$ moves the particle one position to the left or to the right, depending on its internal degree of freedom,
\begin{align*}
	S ( \LocSe{\pm} ) = \LocSe[x\pm1]{\pm}
\end{align*}
where $e_+$ and $e_-$ label an orthonormal basis of $\mathbbm{C}^2$, \eg\ spin up and down states, and $\Loc$ denotes the state of a particle which is localized at position $x$ on the lattice. Mathematically speaking, $\Loc$ is just the element $(\dots,0 ,0, 1, 0, 0, \dots)$ of the canonical basis of $\ltwo$ which has only one non-zero element corresponding to the point $x \in \Z$. The second operation is a unitary transformation acting locally on the internal degree of freedom, defined by the direct sum of elements of the unitary group in two dimensions,
\begin{align*}
	U = \bigoplus_{x \in \mathbbm{Z}} U_x \, , \quad U_x \in \UG \,.
\end{align*}
We speak of the $U_x$ as \emph{quantum coins} or just \emph{coins}, and use the term \emph{coin operator} for $U$. It is block-diagonal with respect to the identification of $\ell_2(\mathbbm{Z}) \otimes \mathbbm{C}^2$ with $\oplus_{i \in \mathbbm{Z}} \mathbbm{C}^2$. If the coin action does not depend on the position of the particle in the lattice, $U_x = V$ for all $x \in \mathbbm{Z}$, then $U$ has the simple form $U = \Id_{\ell_2(\Z)} \otimes V$.  The total walk operator is now defined to be the product
\begin{align*}
	W = U \cdot S \,.
\end{align*}
Of course, interesting effects can only take place if the coin creates a superposition in the basis elements $e_\pm$ of $\C^2$, so that the particle is shifted both ways. For translation invariant walks Fourier methods have been a fruitful method to determine the position distribution for finite times as well as asymptotically \cite{Vogts2010}.

It is well known that a quantum walk can exhibit a propagation speed which is quadratically faster compared to a classical random walk \cite{ambainis-2003-1}. Hence, it is natural to ask for their use in the theory of computational algorithms. And indeed, there are several proposals for quantum algorithms making use of the concept of quantum walks \cite{Kempe2005,ambainis-2003-1}, which improve on their classical counterparts. Another feature of quantum walks is that experimental realizations are challenging but feasible with current technology \cite{karski-2009-325,Schmitz2009}.

Such implementations of quantum walks, however, involve the control of many experimental parameters, \eg\ laser beams or microwaves. Since this can only be done with finite accuracy, noise will be introduced to the system. There are several extremal cases of noise, for instance we could assume that the control parameters vary homogeneously in space and on small time scales compared to the execution time of the quantum walk. This means, that the time evolution changes over time, \ie\ the coin operation becomes time dependent. This kind of noise has been studied in \cite{Vogts2010} and it was shown that this leads to diffusive behavior of the quantum walk, \ie\ the walk behaves like a classical random walk.

Here, we consider the scenario where the control parameters vary on large timescales compared to the execution time of the experiment, but spatial disorder breaks translation invariance. This situation is modeled by a random, space dependent coin. These fluctuations resemble a random potential, but are still described by a unitary operator. The question we address is whether the performance of quantum walks is seriously affected by such kinds of disorder. From the theory of disordered crystals we would expect to observe localization phenomena similar to those described by Anderson in \cite{Anderson}. There, imperfect crystals are also modeled by a random, space dependent potential. The strength of the varying potential is specified at each point by some probability distribution, which is usually assumed to be independently and identically distributed, \ie\ it is the same at each point in space. The Hamiltonian describing such a system is then a sum of the lattice Laplacian and a random, multiplicative potential. Anderson argued that due to interference effects, electrons in disordered crystals can exhibit a strong tendency to localize in finite regions up to exponential tails. In other words, the tunneling probability between two lattice sites decays for all times with increasing distance faster than any polynomial. This phenomenon is called \emph{dynamical localization} \cite{Kirsch:2007bf}.

We show that the statement still holds for the case of a spin-$\frac{1}{2}$ disordered quantum walk in one lattice dimension, where the coin is assumed to vary probabilistically over space, but remains fixed for all times, under mild conditions on the coin distribution. The class of measures for which these conditions are verified include distributions of coins continuous to the Haar measure, \eg\ a gaussian distribution with small variance and centered around some coin of interest, but also a class of probability measures supported on two unitaries, one of them being the Hadamard matrix.

Related models have been studied by Joye and Merkli \cite{Joye2010} and Hamza et. al. \cite{Hamza2009}, which can be seen as a study of disordered quantum walks with identically independent distributed random phases, chosen with respect to an absolutely continuous probability measure. Hamza and co-authors considered diagonal disorder, and find that localization effects occur only if the shift is assumed to be imperfect. Joye and Merkli generalized these results to the case where the coin is given by a product of a fixed unitary with a diagonal one with independent random phases and observe localization effects also for the deterministic shift. Linden and coworkers \cite{Linden2009} considered a case where the perturbation of the coin operator is periodic in space, and Konno and coworkers \cite{Konno2009a,Konno2009} as well as Shikano and Katsura \cite{shikano} derived limit theorems for a special kind of disorder. Compared to previous work, we consider all possible coin distributions and derive sufficient conditions such that they lead to dynamical localization.

Our proof incorporates ideas from the study of localization effects in random unitaries \cite{Hamza2009,Bourget2003,Joye2004}, and uses some general results about Borel measures on the unit circle (see the Book by Cima, Matheson and Ross \cite{bookcauchytransform} for a really nice treatment of this theory) . We rely heavily on the theory of products of real-valued random matrices developed mainly by F\"urstenberg \cite{Furstenberg1963}, which we use as a starting point to derive limit theorems of products of complex matrices having determinant one. As a note for specialists concerning localization questions, let us mention that the proof is based on a multiscale-analysis-scheme, and we first derive the necessary statements, \ie\ a Thouless-like formula and a Wegner estimate.

The paper is organized as follows. We first define in a mathematically rigorous manner the notion of a disordered quantum walk and state our main theorem. The proof is split into different lemma and is given in section \ref{sec_proof}. Some technicalities as well as the verifications of our conditions for our examples of coin distributions are shifted into the appendix.

\section{General Setting}

A disordered quantum walk is a quantum walk where the time evolution is defined by an element of a family of random walk operators $W_{\omega}$ with space dependent random coins. That is, the shift part $S$ of the walk operator is assumed to be undisturbed, but the action of the coin operator at each lattice site $x \in \Z$ is now given by a random variable taking values in the group $\UG$ of two dimensional unitary matrices,
\begin{align*}
U_{\omega_x}:\Omega_x \rightarrow \UG.
\end{align*}
We will call the distribution of this random variable the \emph{single site distribution} and will denote it by $\mu_x$. From now on, we call the set
\begin{align*}
\mathrm{supp}(\mu_x):=\{M\in \mathrm{GL}(\C,k)\,:\,\mu_x(\mathcal{B}_\varepsilon(M))>0\,\forall \varepsilon>0\},
\end{align*}
the support of $\mu_x$. Here, $\mathcal{B}_\varepsilon(M)$ denotes the open sphere of radius $\varepsilon$ around $M$ given by the operator norm. We require that the coin operations are independent and identically distributed random variables at each lattice site, that is $\mu_x = \mu$. Correspondingly, the joint distribution of coin operations for a finite collection of $L$ lattice sites is given by the product of $L$ copies of the single site distribution $\mu$. Since this is true for all finite collections of lattice sites, there exists a unique joint distribution $\mu_\infty$ on the \emph{infinite product probability space}
\begin{align*}
	\Omega = \times_{x \in \Z} \, \Omega_x \, .
\end{align*}
We will use the symbol $\Expect{X}$ to denote the expectation value of some random variable $X$, and $\Prob{E}$ to denote the probability of some event $E$, both taken with respect to $\mu_\infty$.  The coin operator $U_\omega$, $\omega \in \Omega$ on $\ell_2(\Z) \otimes \C^2$ is a random variable on $\Omega$ and defined by the direct sum
\begin{align*}
	U_\omega =  U_{(\dots, \omega_{x-1}, \omega_x, \omega_{x+1}, \dots)} = \bigoplus_{x\in\mathbbm{Z}} U_{\omega_x} \, .
\end{align*}
The one-dimensional, spin-$\frac{1}{2}$ \emph{disordered quantum walk} $W_\omega$ is a random variable on $\Omega$ with values in the unitary group of $\ell_2(\Z) \otimes \C^2$ given by the product
\begin{align*}
	W_\omega = U_\omega \cdot S \, .
\end{align*}
Note that the coins are drawn in advance and kept constant during all time steps.

The coins of a quantum walk play the role of a potential, \ie they change the probability that a particle is moving to the right or to the left, therefore determining the transmission probability. If all coins are the same, \ie the system is translation invariant, the generalized eigenfunctions are described by Bloch waves being infinitely extended over the whole lattice. Thus, they are of course not square integrable and hence no elements of the Hilbert space. Because of that, they are called \emph{generalized} eigenfunctions. The associated spectrum of the Walk operator does not consist of separated points, but is rather constituted of energy bands. It follows that the walk operator has absolutely continuous spectrum and the system exhibits ballistic scaling. See \cite{Vogts2010} for a thorough treatment of this case.

As the coins vary in space, the system is no longer translational invariant, and the particle ``sees'' different ``potential values''. Due to that, we would expect that the spreading of initially localized wave packets will be slowed down. For the case of hamiltonian systems, we even have that the tunneling probability between two lattice sites is suppressed in their distance. This phenomenon is called \emph{dynamical localization}. The corresponding definition for quantum walks is as follows.

\begin{definition}
	\label{defndynloc}
	Let $W_\omega$ be a one-dimensional, spin-$\frac{1}{2}$ disordered quantum walk. $\W$ is said to exhibit \emph{dynamical localization}, if there exists a function $\f:\N \rightarrow \R_+$ such that $\f(n)$ goes to zero for $n \rightarrow \infty$ faster than any polynomial in $n$ and we have that
	\begin{align*}
		\Expect{\sup_{t \in \N} |\Scp{\LocS[y]{\phi}}{\, W_\omega^t \LocS{\psi}}|} \; \leq \;  \f(|x - y|) \, .
\end{align*}
Here, $\Loc[x]$ and $\Loc[y]$ denote again localized (position) states, whereas $\phi$ and $\psi$ refer to some arbitrary internal spin states. The function $\f$ is called the localization length.
\end{definition}

Dynamical localization is closely related to the spectral properties of $W_\omega$. As a quantum walk is described by a unitary operator, its spectrum $\sigma(\W)$ is a subset of the unit circle $\T = \{z \in \C, \, |z| = 1\}$. We will call its elements frequencies or quasi-energies. The disordered quantum walk is said to exhibit \emph{spectral localization} if it has only pure point spectrum, meaning that each of its spectral measures is a linear combination of Dirac point measures. Due to a version of the well-known RAGE theorem adapted for unitary operators (theorem \ref{ragethmuni}) dynamical localization also implies spectral localization. It follows that $\W$ has a complete set of square integrable eigenfunctions. In the case of disorder, we would expect even a stronger decay behavior of eigenfunctions.

\begin{definition}
	A disordered quantum walk $\W$ is said to exhibit strong spectral localization, if its eigenfunctions decay faster than any polynomial, implying that the spectrum is of pure point nature.
\end{definition}

\section{Results}

To state our results concerning disordered quantum walks, we first introduce the following family of mappings from the general linear group in two dimensions $\mathrm{GL}(\C,2)$ into itself
\begin{align*}
z \in \C \,:\quad \tau_z\, :\,\left( \begin{array}{cc}
a & b \\
c & d
\end{array}\right)
\quad\longrightarrow\quad
\frac{1}{a}\left( \begin{array}{cc}
\frac{ad-bc}{z} & c \\
-b & z
\end{array}\right)
\end{align*}
where of course $a \neq 0$ is assumed. In the special case when the mappings $\tau_z$ are applied to unitary matrices, it can easily be seen that their image consists only of elements of the group $\mathrm{SL}_\mathbbm{T}$ of complex matrices with determinant of modulus one,
\begin{align*}
\mathrm{SL}_\mathbbm{T} \,=\, \{M\in \mathrm{Mat}(\C,2)\,:\,|\Det{M}|=1\}\, .
\end{align*}
Here, $\mathrm{Mat}(\C,2)$ denotes the space of two by two complex matrices. The matrices $\tau_z(U) \in \mathrm{SL}_\mathbbm{T}$, where $U \in \UG$ is in the support of $\mu$, are called \emph{transfer matrices}. Their importance stems from the fact that the spectral and dynamical behavior of disordered quantum walks crucially depends on the group of matrices generated by the set of transfer matrices.

\begin{theorem}
	\label{thmtransfermatrices}
	Let $W_{\omega}$ be a disordered quantum walk characterized by its single site distribution $\mu$ with support $supp(\mu) \subset \UG$. Suppose that for almost all elements of the unit circle $\theta \in \T$ the group $\Gmut \subset \mathrm{SL}_\mathbbm{T}$ generated by its associated set of transfer matrices $\tau_{\theta}(U)$, $U \in supp(\mu)$ is
	\begin{enumerate}
		\item non-compact,
		\item contains no reducible subgroup of finite index and
		\item the expectation $\Expect{\Norm{\tau_\theta(U)}^\zeta}$ is finite for some $\zeta > 0$.
	\end{enumerate}
Then $W_\omega$ exhibits dynamical as well as strong spectral localization. Here, the terminus ``almost all'' refers to the Lebesgue measure on the unit circle.
\end{theorem}

\begin{remark}
If the group $\Gmut \subset \mathrm{SL}_\mathbbm{T}$ happens to be compact for some element $\theta$ of the unit circle, then $\theta$ would be an element of the absolutely continuous spectrum of $W_\omega$. This is exactly what happens in the case of translation invariant quantum walks.
\end{remark}

This result implies dynamical localization for many physical situations, \ie for a setup where the application of a fixed coin is desired but it cannot be circumvented that with some small probability quantum coins close to the target one are also applied.

\begin{corollary}
	Let $W_\omega$ be a disordered quantum walk such that the single site distribution possesses a positive density with respect to the Haar measure on $\UG$. Then $W_\omega$ exhibits dynamical as well as strong spectral localization.
\end{corollary}

In fact, it is enough if the single site distribution consists partly of some absolutely continuous measure. The verification of the assumptions of theorem \ref{thmtransfermatrices} for this case is carried out in Appendix \ref{AppFurst}. Although continuous distributions can be argued to model many experimental imperfections, the situation where the support of the single site distribution includes only a finite number of unitary matrices is of independent interest. Indeed, we may ask what happens if we disturb the usual Hadamard walk, as realized by \cite{karski-2009-325,Schmitz2009}, by just one single coin chosen with some arbitrary, but non-zero probability. Then, we have the following result, as shown in Appendix \ref{Appdiscretemeasure}.

\begin{corollary}
\label{finitecoinset}
Let $W_{\omega}$ be a disordered quantum walk such that the support of its single site distribution equals the set
\begin{align*}
&\left\{ \,
\frac{1}{\sqrt{2}} \, \left( \begin{array}{cc}
1 & -1 \\
1 & 1
\end{array}\right)
\,,\,
\left( \begin{array}{cc}
a & b \\
-\bar{b} & \bar{a}
\end{array}\right)
\, \right\} \, ,
\end{align*}
with complex numbers $a$ and $b$ fulfilling $a \neq 0$, $|a|^2 + |b|^2 = 1$, and $|a| < |\mathfrak{I} \, b|$, where $\mathfrak{I} b$ denotes the imaginary part of $b$. Then $W_\omega$ exhibits dynamical as well as strong spectral localization.
\end{corollary}

In the case that $a = 0$, there is a non-zero probability that ``flips'' occur. These are coins which cause reflections of the walking particle at the corresponding sites, see section \ref{InducedMeasure}. Dynamical localization follows easily in this case from lemma \ref{ZeroSet}. We continue with the proof of theorem \ref{thmtransfermatrices}.

\section{Proof}
\label{sec_proof}
\subsection{Outline of the proof}\label{BoundPowers}

At first, we reduce the problem to the case of finite lattices and show how the expectation value of interest can be bounded by an expression involving the resolvent of the quantum walk operator, see equations \eqref{eqn_fini_rest}, \eqref{eqspecmeascauchy} in section \ref{secfinrestrictions}. The next sections, \ref{Resolvent} and \ref{InducedMeasure}, study the decomposition of the resolvent into expressions mainly depending on transfer matrices, see in particular lemma \ref{Decomposition}. In section \ref{InducedMeasure} the properties of transfer matrices are examined, and we use the theory of F\"{u}rstenberg, (appendix \ref{F\"{u}rstenberg}), about products of random matrices to get a first estimate of the decay properties of the tunneling probability for a fixed length scale (proposition \ref{initiallengthscale}). We then proceed by examining the properties of the density of states, both for the finite and infinite lattice case (section \ref{sec_thouless}). In particular, we prove an analog of the Thouless formula for quantum walks (see equation \eqref{eqthouless}). The Thouless formula is then used to prove the H\"{o}lder continuity of the integrated density of states (see proposition \ref{lemcontiIDS}) based again on general properties of products of random matrices (appendix \ref{F\"{u}rstenberg}). This in turn allows us to prove an upper bound on the probability that the tunneling probability for some fixed quasi-energy is high for two independent regions at the same time, \ie\ we prove a Wegner-type bound (proposition \ref{lemwegner} and equation \eqref{equniformwegner}). The initial scale estimate and the Wegner bound are then combined using the multiscale-analysis technique, similar to the ones in \cite{Germinet:2001p4369,pre05533269}, in order to get an upper bound on the tunneling probability valid for a sequence of increasing distances, see lemma \ref{msalemma}. The final section, \ref{ProffOfExponentialDecay}, combines these results and finally provides an upper bound on the decay properties of the tunneling probability over arbitrary distances, as well as bounds on the decay of eigenfunctions.

\subsection{Restrictions to finite volumes and Cauchy transform}
\label{secfinrestrictions}

In this section we define a finite unitary restriction $\W(N)$ of the Walk operator $\W$ on the lattice sites $-N$ to $N$. The idea is to impose reflective boundary conditions on the lattice sites $-(N+1)$ and $N+1$ which is equivalent to changing the coins $U_i$ at these two lattice sites to the flip operation, \eg\ the Pauli $X$ matrix
\[ X = \left( \begin{array}{cc}
0 & 1 \\
1 & 0
\end{array}\right)  \]
multiplied by a phase factor $e^{i\eta^{\bf L,R}}$. Then we restrict the operator to the lattice sites $-N$ to $N$, but in order to obtain an unitary operator, we, roughly speaking, have to include half of the operator on the neighboring sites $-(N+1)$ and $N+1$ where we changed the coin to be the flip operation, see also \eqref{eqn_fini_rest}.
\begin{align}\label{eqn_fini_rest}
\left(
\begin{BMAT}[8pt]{cccccccccc}{ccccccccc}
0 & 0   & 0 & 0 & 0   & \dots & & & & \\
0 & 0   & e^{i\eta^{\bf L}} & 0 & 0   &  & & & &\\
0 & b_{-N} & 0 & 0 & a_{-N} & & & & & \\
0 & d_{-N} & 0 & 0 & c_{-N} &  & & & &\\
\vdots  &     &   &   &  & \ddots & & & & \vdots\\
& &  &   &   & b_N & 0 & 0&a_N&0 \\
& &  &   &   & d_N & 0 & 0&c_N&0\\
& & & & & 0&  0 &  e^{i\eta^{\bf R}} & 0 & 0 \\
& & & &\dots & 0&  0 &  0 & 0 & 0
\addpath{(0,9,.)rrddlluu}
\addpath{(2,7,.)rrddlluu}
\addpath{(6,4,.)rrddlluu}
\addpath{(8,2,.)rrddlluu}
\addpath{(8,2,.)rrddlluu}
\addpath{(1,8,1/2)rrrrrrrrdddddddlllllllluuuuuuu}
\end{BMAT}
\right)
\end{align}
To be more precise we define the restriction $\W(N):\C^{4(N+1)}\rightarrow \C^{4(N+1)}$ of \W\ explicitly via its matrix elements
\begin{align}\label{eq_uniRest}
(\W(N))_{i,j}=(\hat\W)_{i,j} , \quad -2N-1\leq i,j\leq 2N+2\;,
\end{align}
where $\hat\W$ is the walk operator with the coins $U_{\pm(N+1)}$ changed to the Pauli $X$ matrix multiplied by a phase factor $e^{i\eta^{\bf L,R}}$. Note that this choice already ensures unitarity of $\W(N)$ since all rows and columns are normalized and mutually orthogonal.

Since the only transport between neighboring lattice sites is caused by the shift $S$ the walk operator \W\ as well as $W_\omega^t$ for finite $t \in \N$ are by construction band matrices of finite width (of course depending on $t$).
This implies a finite velocity for any initially localized particle, meaning that we can provide a $t$-dependent upper bound on the region where the particle could be detected after $t$ time steps with non-zero probability. Hence, we can express the probability to detect an initially at position $x$ localized particle at position $y$ after applying a fixed number $t$ of time steps by finite restrictions of quantum walk operators,
\begin{align}\label{eqfinitelimit}
	|\Scp{\LocS[y]{\phi}}{W_\omega^t \LocS{\psi}}| \,=\, \lim_{N \rightarrow \infty} |\Scp{\LocS[y]{\phi}}{\Wf^t \LocS{\psi}}| \, .
\end{align}
Here $\phi$ and $\psi$ denote again some arbitrary internal spin states. But clearly, it is enough to consider the case $\phi=e_i$, $\psi=e_j$, $i,j \in \{1,2\}$ where $e_1, e_2$ is the standard basis set of $\C^2$. We make another simplification and consider the time behavior of $\Wf$ first just for open subsets of quasi-energies, \ie\ open arcs of $\T$. That is, we prove the decay of quantities of the form
\begin{align*}
	|\Scp{\LocS[y]{e_i}}{\Wf^t \chi(I_\delta) \LocS{e_j}}|
\end{align*}
where $I_\delta$ denotes an open arc of the unit circle centered around some $\theta \in \T$ with elements whose phases differ from the phase of $\theta$ only up to $\delta>0$, and $\chi(I_\delta)$ is the associated (eigen-)projector. Once we proved the decay for such an open arc, the result for the whole unit circle follows by compactness.

Let us denote the spectral measure of $\Wf$ associated to the vectors $\LocS{e_i}$ and $\LocS[y]{e_j}$ by $\rhof$. We omit here the dependence on $e_i,e_j$, since all our arguments are independent from them. Note that since $\Wf$ is a unitary operator acting on a finite dimensional Hilbert space, its spectral measures are singular measures, \ie\ only supported on points.  By the spectral theorem and since $\Wf$ is unitary, $\rhof$ is a signed measure on the unit circle $\mathbb{T}$ and
\begin{align*}
	\sup_{t \in \N} ||\Scp{\LocS[y]{e_i}}{\Wf^t \chi(I_\delta) \LocS{e_j}}|| \,=\, \sup_{t \in \N} \left| \int_{\mathbb{T}} \, \rhof(d\theta) \, \chi(I_\delta) \, \theta^t \right| \,,
\end{align*}
where we denoted in a slight abuse of notation also the characteristic function of the set $I_\delta$ by $\chi(I_\delta)$. Since $|\theta^t| = 1$ for $\theta \in \mathbb{T}$ we can bound the expression on the left hand site using the triangle inequality by the measure of the open arc $I_\delta$,
\begin{align*}
	\sup_{t \in \N} |\Scp{\LocS[y]{e_i}}{\Wf^t \chi(I_\delta) \LocS{e_j}}| \,\leq\, \int_{\mathbb{T}} \, |\rhof(d\theta)| \chi(I_\delta) \,=\, |\rhof|(I_\delta) \,.
\end{align*}
In order to estimate the value associated by $\rhof$ to the open arc $I_\delta$ we introduce the \emph{Cauchy transform} $K\rhof$, defined as a function on the open unit disk $\mathbb{D} = \{z \in \C : |z| < 1 \}$ by the formula
\begin{align*}
	K\rhof (z) \, = \, \int_{\mathbb{T}} \, \rhof (d\theta) \frac{1}{1 - \overline{\theta} z} \, .
\end{align*}
Note that $K\rhof (z)$ is an analytic function on the open unit disk. It may become infinite for some values of $z$ as $z$ approaches the unit circle, and these exceptional values are precisely the points where the spectral measure is supported and hence are equal to the eigenvalues of the walk operator\footnote{Note that we are dealing here with unitary operators on finite dimensional Hilbert spaces and hence the  absolutely continuous part  of every spectral measure is identically zero.}. A theorem of Smirnov assures that except from these points (which have Lebesgue measure zero), the limits $\lim_{r \rightarrow 1^-} K\rhof (r \theta)$, $\theta \in \T$ exist. See \cite{bookcauchytransform} for a formal statement of this theorem as well as for a proof. The reference is also well suited for a general introduction into the topic of Cauchy integral transforms and their properties. Shortening the notation, we denote the corresponding, almost everywhere well defined function on the unit circle again by $K\rhof$.

A theorem by Poltoratski \cite{Poltoratski:1996p4363}, which can be thought of as a refinement of Boole's well known theorem, states that we can express the action of a singular measure on the unit circle on some continuous function $f$ by studying its Cauchy transform,
\begin{align}\label{eqcauchydistrweakstar}
	\int_{\mathbb{T}} \, |\rhof| (d \theta) f(\theta) \,=\, \lim_{\kappa \rightarrow \infty} \pi \kappa \cdot \int_{\mathbb{T}} d\theta \, \chi_{\left\{ \, \theta \in \mathbb{T} \,:\, |K\rhof (\theta)| \,>\,\kappa \, \right\}} (\theta) \, f(\theta) \, .
\end{align}
where $d\theta$ is the integration with respect to the normalized Lebesgue measure on the unit circle and $\chi_A$ again denotes the indicator function of the set $A$. Put differently, the measures
\[ \pi\, \kappa\, d\theta\, \chi_{\left\{ \, \theta \in \mathbb{T} \,:\, |K\rhof (\theta)| \,>\,\kappa \, \right\}} \]
converge to the absolute value of the spectral measure $\rhof (d \theta)$ in the weak-$\ast$-topology as $\kappa$ goes to infinity. Remarkably, this theorem also holds if we restrict ourselves onto some open arc of the unit circle, which gives us
\begin{align}\label{eqcauchydistr}
	|\rhof|(I_\delta) \,=\, \lim_{\kappa \rightarrow \infty} \pi \kappa \cdot \lambda \left| \left\{ \, \theta \in I_\delta \,:\, |K\rhof (\theta)| \,>\,\kappa \, \right\} \right| \, .
\end{align}
Here, $\lambda |A| = \int_{A} d\theta$ denotes the Lebesgue measure of the set $A$. Hence we have to control the growth of the Cauchy transform of the spectral measure $\rhof$.

In the special case when the measure on the unit circle is induced by a unitary operator, the Cauchy transform is directly linked to the corresponding matrix elements of the resolvent,
\begin{align*}
	K\rhof (z) \, &= \, \int_{\mathbb{T}} \, \rhof (d\theta) \, \frac{\theta}{\theta - z} \\
	&=\, \Scp{\LocS[y]{e_i}}{ \Wf (\Wf - z )^{-1} \LocS{e_j}} \\
	&=\, \Scp{\LocS[y]{e_i}}{ \Id + z \cdot (\Wf - z )^{-1} \LocS{e_j}} \\
	&=\, z \cdot \Scp{\LocS[y]{e_i}}{(\Wf - z )^{-1} \LocS{e_j}} \, ,
\end{align*}
for $x\neq y$. Hence, we have that for $x\neq y$
\begin{align}
	\label{eqspecmeascauchy}
	|\rhof|(I_\delta) \,=\, \lim_{\kappa \rightarrow \infty} \pi \kappa \cdot \lambda \left| \left\{ \, \theta \in I_\delta \,:\, |\Scp{\LocS[y]{e_i}}{ (\Wf - \theta )^{-1} \LocS{e_j}}| \,>\,\kappa \, \right\} \right|\, .
\end{align}
We continue by first deriving an explicit decomposition of the resolvent in terms of the already mentioned transfer matrices, see the next sections. In section \ref{InducedMeasure}, we study their properties and we finish by proving the decay for some open arc $I_\delta$ in section \ref{ProffOfExponentialDecay}.

\subsection{Resolvent formula}\label{Resolvent}

In the following we derive an explicit expression for the resolvent in terms of transfer matrices. Since the results are independent of the size $N$ of the quantum walk $\Wf$, and for shortening the notation, we mostly omit the dependence on $N$. We can identify the Hilbert space $\ell_2(\Z)\otimes \C^2$ with $\ell_2(\Z)$ via the isomorphism
\[
\delta_x\otimes e_-\, \rightarrow \, f_{2x}\, ,\quad \,\delta_x\otimes e_+\, \rightarrow \, f_{2x+1} \,,\quad \forall x\in\Z\,,
\]
where the $e_\pm$ form an orthonormal basis of $\C^2$ and the $f_j$ denote the standard basis of $\ell_2(\Z)$. We want to study the matrix elements of the resolvent with respect to the standard basis vectors $\delta_x\otimes e_\pm$. To shorten the notation let us abbreviate the resolvent by
\[
\Res =(\W -z)^{-1}\, ,\quad z\notin \spec{\W}\,.
\]
We denote the $l$'th column vector of the resolvent by $\ResCol{l}\in \ltwo$. In the following we determine $\ResCol{l}$ locally in the sense that knowledge of two adjacent entries $\ResColArg{l}{2x}$ and $\ResColArg{l}{2x-1}$ allows us to construct $\ResCol{l}$ iteratively. Assume now we know the entries $\ResColArg{l}{2x}$ and $\ResColArg{l}{2x-1}$ for $l\neq 2x,2x+1$. The shift and coin
operators acting on $\ltwo$ are given by
\[
Sf_{2x+i}=\left\{\begin{array}{ll}f_{2(x-1)}&,\, i=0\\f_{2(x+1)+1}&,\,i=1\end{array}\right.\, ,
\]
\[
U_y f_{2x+i}=\left\{\begin{array}{ll}0&,\,x\neq y \\a_xf_{2x}+c_xf_{2x+1}&,\, x=y\,,\,i=0\\b_xf_{2x}+d_xf_{2x+1}&,\,x=y\,,\,i=1\end{array}\right.\,.
\]
Hence, we can determine the operator\footnote{The columns resp. rows of the matrix are indexed by integers increasing from left to right resp. top to bottom. The block in the middle corresponds to basis vectors $f_{2x}$ and $f_{2x+1}$.}
\[
\W-z=\left(
\begin{array}{ccc|cc|ccc}
\ddots & & & & & & &\\
& -z & 0 & a_{x-1} & 0 & 0 & 0 &\\
& 0 & -z & c_{x-1} & 0 & 0 & 0 &\\
\hline
& 0 & b_x & -z & 0 & a_x & 0 &\\
& 0 & d_x & 0 & -z & c_x & 0 &\\
\hline
& 0 & 0 & 0 & b_{x+1} & -z & 0 &\\
& 0 & 0 & 0 & d_{x+1} & 0 & -z &\\
& & & & & & &\ddots
\end{array}
\right)
\]
and it follows from the equation
\[
(\W-z)\cdot \ResCol{l} = f_l
\]
that for $l\neq 2x,2x+1$ the column vector $\gamma^l$ satisfies the equations
\begin{eqnarray*}
\ResColArg{l}{2x-1}b_x -z\ResColArg{l}{2x} + \ResColArg{l}{2x+2} a_x & = & 0\\
\ResColArg{l}{2x-1}d_x -z\ResColArg{l}{2x+1} +\ResColArg{l}{2x+2} c_x & = & 0\, .
\end{eqnarray*}
It is easy to see that this implies the following relation between the entries \ResColArg{l}{2x-1}, \ResColArg{l}{2x} and \ResColArg{l}{2x+1}, \ResColArg{l}{2x+2}
\[
\genfrac(){0pt}{}{\ResColArg{l}{2x+1} }{ \ResColArg{l}{2x+2}}
=
\frac{1}{a_x}\left( \begin{array}{cc}
\frac{\Det{U_x}}{z} & c_x \\
-b_x & z
\end{array}\right)
\genfrac(){0pt}{}{\ResColArg{l}{2x-1} }{ \ResColArg{l}{2x} }
\]
which motivates the following definition.
\begin{definition}
\label{TransferMatrices}
Let $\W=\U\cdot \S$ be a disordered quantum walk. The transfer matrices $\Trans{x}$ of $\W$ for all sites $x$ with $a_x\neq 0$ are defined by
\[
\Trans{x} :=
\frac{1}{a_x}\left( \begin{array}{cc}
\frac{\Det{U_x}}{z} & c_x \\
-b_x & z
\end{array}\right)\,,\, x\,\in\,\Z.
\]
Hence, with the definition
\[
\Gamma^l_x:=\genfrac(){0pt}{}{\ResColArg{l}{2x-1} }{\ResColArg{l}{2x} }
\]
and under the condition $l\neq 2x, 2x+1$ we have the relation
\[
\Gamma^l_{x+1}=\Trans{x}\Gamma^l_x\,.
\]
\end{definition}
\begin{remark}
The condition $a_x\neq0$ is crucial for the transfer matrices to be well defined. In fact, the case $a_x=0$ corresponds to a reflection of the walking particle at site $x$. We analyze this case in the next section and show that we can assume without loss of generality $a_x\neq 0$ for all sites $x$.
\end{remark}
If the requirement $l\neq 2x, 2x+1$ is not satisfied, the application of $\Trans{x}$ to the vector $\Gamma^l_{x}$ results in a vector different from $\Gamma^l_{x+1}$. The following lemma is concerned with an analysis of the sequence $\phi^{l,x}_n$ defined as
\begin{equation}
\label{TSequence}
\genfrac(){0pt}{}{ \phi^{l,x}_{2y-1} }{ \phi^{l,x}_{2y} } := \left\{
\begin{array}{rl}
\Trans{y}^{-1}\ldots\Trans{x-1}^{-1}\Gamma^l_x & , \, \mathrm{if} \,y<x \\
\Gamma^l_x & , \, \mathrm{if} \,y=x \\
\Trans{y-1}\ldots\Trans{x}\Gamma^l_x & , \, \mathrm{if} \,y>x
\end{array}
\right.\,.
\end{equation}
\begin{lemma}
\label{TransferSequence}
Let $y,x,l\in \Z$ and $\phi^{l,x}_n$ be as defined in (\ref{TSequence}). The resolvent and the sequence $\phi^{l,x}_n$ coincide, \ie\
\[
\genfrac(){0pt}{}{ \phi^{l,x}_{2y-1} }{ \phi^{l,x}_{2y} } =\Gamma^l_{y}\,,
\]
if $2x-1<l$ and $2y-1<l$ respectively $2x>l$ and $2y>l$.
\end{lemma}
\begin{proof}
The lemma follows directly from definition \ref{TransferMatrices}.
\end{proof}
\begin{remark}
In the case of finite restrictions $\Wf$ the sequences $\phi_n^{l,x}$ have to fulfill certain boundary conditions. In the limit $N\rightarrow \infty$ it is clear that the resolvent is a bounded operator and therefore it follows from lemma \ref{TransferSequence} that for $2x-1<l$ the sequence $\phi^{l,x}_n$ is left square summable, \ie\
\[
\sum_{n<r}|\phi^{l,x}_n|^2\,<\,\infty\, ,\quad \forall \,r\in \Z
\]
and for $2x>l$ it is right square summable, \ie\
\[
\sum_{n>r}|\phi^{l,x}_n|^2\,<\,\infty\, ,\quad \forall\, r\in \Z\,.
\]
In the following we call sequences $\phi_\pm$ satisfying these requirements left respectively right compatible with $\Res$.
\end{remark}
In fact, since the series (\ref{TSequence}) is a solution of the equation
\begin{equation}
\label{GeneralizedEigenvector}
(W_\omega-z)\phi =0
\end{equation}
for $z\notin \spec{\W}$ it can only be left or right square summable but not both. Now, let \GenEVM\ and \GenEVP\ be left (respectively, right) compatible solutions of (\ref{GeneralizedEigenvector}), for instance the ones given by lemma \ref{TransferSequence}. If we choose exactly the solutions of lemma \ref{TransferSequence} it is clear that the matrix element $\Res{\scriptstyle(n,m)}$ is given by \GenEVM\ respectively \GenEVP . For general left respectively right compatible solutions $\phi_\pm$ we write for the resolvent at $m=2y,2y+1$
\[
\Res{\scriptstyle (n,m)} =
\left\{ \begin{array}{lcl}
\alpha_m \, \GenEVM & , &  \textrm{if $n\leq 2y$}\\
\beta_m \, \GenEVP & , &  \textrm{if $n> 2y$} \, .
\end{array}\right .
\]
When we consider the cases $m=2y$ or $2y+1$ the constants $\alpha_m$ and $\beta_m$ have to satisfy the equations
\begin{eqnarray*}
\alpha_m (b_y \GenEVM[2y-1] -z \GenEVM[2y]) +\beta_m a_y \GenEVP[2y+2] & = &\delta_{m,2y}\\
\alpha_m d_y \GenEVM[2y-1] +\beta_m (c_y \GenEVP[2y+2] -z \GenEVP[2y+1]) &= &1-\delta_{m,2y}\,.
\end{eqnarray*}
Since \GenEVP\ and \GenEVM\ are solutions of (\ref{GeneralizedEigenvector}) we can rewrite these equations in the following form
\[
\left( \begin{array}{cc}
-a_y \GenEVM[2y+2] & a_y \GenEVP[2y+2] \\
d_y \GenEVM[2y-1] & -d_y \GenEVP[2y-1]
\end{array}\right)
\left( \begin{array}{c}
\alpha_m \\
\beta_m
\end{array}
\right)
=
\left( \begin{array}{c}
\delta_{m,2y} \\
1-\delta_{m,2y}
\end{array}
\right)\,.
\]
Using Cramer's rule and the definition
\[
A_y:=
\left( \begin{array}{cc}
\GenEVM[2y-1] & \GenEVP[2y-1] \\
\GenEVM[2y+2] & \GenEVP[2y+2]
\end{array}\right)
\]
we see that the coefficients $\alpha_m$ and $\beta_m$ are given by
\[
\begin{array}{rcr}
\alpha_{2y}  =  \frac{\GenEVP[2y-1]}{a_y \Det{A_y}} & , &
\beta_{2y}  =  \frac{\GenEVM[2y-1]}{a_y \Det{A_y}} \\
\alpha_{2y+1}  =  \frac{\GenEVP[2y+2]}{d_y \Det{A_y}} & , &
\beta_{2y+1}  =  \frac{\GenEVM[2y+2]}{d_y \Det{A_y}}
\end{array}\,.
\]
We would like to replace $A_y$ by a matrix $B_y$ where the components of $\phi_\pm$ appearing in $B_y$ correspond to adjacent sites. That would allow us to calculate $B_y$ from $B_{y+1}$ via transfer matrices. Arbitrary solutions $\phi$ of (\ref{GeneralizedEigenvector}) satisfy the relation
\[
\left( \begin{array}{c}
\phi {\scriptstyle (2y-1)} \\
\phi {\scriptstyle (2y)}
\end{array}\right)
=
\left( \begin{array}{cc}
1 & 0 \\
\frac{b_y}{z} &\frac{a_y}{z}
\end{array}\right)
\left( \begin{array}{c}
\phi {\scriptstyle (2y-1)}  \\
\phi {\scriptstyle (2y+2)}
\end{array}\right)\, ,
\]
hence, we have for
\begin{equation}
\label{BMatrix}
B_y:=
\left( \begin{array}{cc}
\GenEVM[2y-1] & \GenEVP[2y-1] \\
\GenEVM[2y] & \GenEVP[2y]
\end{array}\right)
\end{equation}
the following correspondence
\[
\Det{B_y}=\frac{a_y}{z}\Det{A_y}\,.
\]
Moreover, we have $B_{y+1}=\Trans{y}B_y$ and since $\Det{\Trans{y}}=d_y a_y^{-1}$ the matrices $B_y$ satisfy
\[
|\Det{B_y}|=|\Det{B_{y'}}|\, ,\quad \forall\, y,y'\in \Z \,.
\]
Finally, we get the following expressions for the coefficients $\alpha_m$ and $\beta_m$
\[
\begin{array}{rcl}
\alpha_{2y}=\frac{\GenEVP[2y-1]}{z \Det{B_y}} & , & \beta_{2y}=\frac{\GenEVM[2y-1]}{z \Det{B_y}}\\
\alpha_{2y+1}=\frac{a_y\GenEVP[2y+2]}{z d_y\Det{B_y}} & , & \beta_{2y+1}=\frac{a_y\GenEVM[2y+2]}{z d_y\Det{B_y}}
\end{array}\,.
\]
The following lemma summarizes the results of this section.
\begin{lemma}
\label{FractionalMoments}
Let $\W$ be a disordered quantum walk and $\GenEVM$ and $\GenEVP$ left respectively right compatible solutions of (\ref{GeneralizedEigenvector}). Then the modulus of the entries of the resolvent $\Res{\scriptstyle (n,m)}$ for $m\in\{2y,2y+1\}$ is given by the formula
\begin{equation}
\label{ResolventFormula}
|\Res {\scriptstyle (n,m)}| = \frac{1}{|z\Det{B_k}|}\cdot \left\{\begin{array}{ccl}
|\phi_- {\scriptstyle ( n)}\phi_+ {\scriptstyle ( \hat m)}|&,&n\leq 2y\\
|\phi_+ {\scriptstyle ( n)}\phi_- {\scriptstyle ( \hat m)}|&,&n>2y\end{array}\right.
\end{equation}
with $B_k$ defined by (\ref{BMatrix}) and arbitrary $k\in \Z$. The argument $\hat m$ is given by $\hat m=m-1$ for even $m$ and $\hat m=m+1$ for odd $m$.
\end{lemma}


\subsection{Transfer matrices and properties of the resolvent}
\label{InducedMeasure}
In this section we analyze the resolvent formula \eqref{ResolventFormula} in further detail and prove some useful facts originating from the structure of the transfer matrices $T{\scriptstyle (z)}$. Those matrices induce a family of maps $\tau_z$ labeled by $z\in \C\backslash\{0\}$ from a subset of the unitary group $\UG  $ to the group of invertible matrices $\mathrm{GL}(\C,2)$ via
\begin{align}
\label{Tau}
\tau_z\, :\,\left( \begin{array}{cc}
a & b \\
c & d
\end{array}\right)\in \UG  \, ,\,a\neq 0
\quad\longrightarrow\quad
\frac{1}{a}\left( \begin{array}{cc}
\frac{ad-bc}{z} & c \\
-b & z
\end{array}\right)\in \mathrm{GL}(\C,2)\,.
\end{align}
In order to make the map $\tau_z$ well defined, we define the set of two by two unitary matrices with non vanishing diagonal entries
\begin{align}
\label{NonOffDiagonal}
\mathcal{U}_{ND}:=\{U\in \UG \,:\,U_{11}\neq 0 \neq U_{22}\}\,,
\end{align}
the complement of $\mathcal{U}_{ND}$ in $\UG$ is denoted by $\overline{\mathcal{U}_{ND}}=\UG\backslash \mathcal{U}_{ND}$. For each $U\in \mathcal{U}_{ND}$ we have that $\Det{\tau_z (U)}=U_{22}/U_{11}$, hence it is clear that $\tau_z(\mathcal{U}_{ND})\subset \mathrm{GL}(\C,2)$. The following lemma proves that this inclusion is strict and gives an explicit parametrization of $\tau_z (\mathcal{U}_{ND})$.
\begin{lemma}
\label{ImageTau}
Let $\tau_z$ and $\mathcal{U}_{ND}$ be defined as in (\ref{Tau}) and (\ref{NonOffDiagonal}). Then, for $z\in \C\backslash \{0\}$ $\tau_z$ is injective on $\mathcal{U}_{ND}$ and its image $\tau_z (\mathcal{U}_{ND})$ is given by the set
\begin{equation}
\label{TransMatParam}
\left\{
\left( \begin{array}{cc}
\sqrt{1+r^2}\ex{i\alpha}|z|^{-1} & r\ex{i\beta}\\
r\ex{i\gamma} & \sqrt{1+r^2}\ex{i(\beta +\gamma -\alpha)}|z|
\end{array}
\right)
\,:\,r\in \R_+\,,\, \alpha ,\beta ,\gamma \in [0,2\pi)
\right\}\,.
\end{equation}
\end{lemma}
\begin{proof}
The inverse of $\tau_z$ is given by
\begin{equation}
\label{TauInv}
\tau_z ^{-1}\,:\, \left( \begin{array}{cc}
v & w \\
x & y
\end{array}\right)\,
\longrightarrow \,
\frac{z}{y}\left( \begin{array}{cc}
1 & -x \\
w & v y- x w
\end{array}
\right)\, ,
\end{equation}
which only exists for $y\neq 0$. Since we assumed $z\neq 0$ this already proves that $\tau_z^{-1}\circ \tau_z (U)=U $ for all $U\in \mathcal{U}_{ND}$, hence $\tau_z$ is injective.

The image of $\tau_z$ can be inferred from its action on a general unitary
\[
\tau_z \left (\left( \begin{array}{cc}
a & b \\
-\bar b \ex{i\phi} & \bar a \ex{i\phi}
\end{array}\right)\right) =
\frac{1}{a}\left( \begin{array}{cc}
\frac{\ex{i\phi}}{z} & -\bar b \ex{i\phi} \\
-b & z
\end{array}\right)\, ,\quad a\neq 0\, ,\, |a|^2+|b|^2=1\,.
\]
Clearly, if we define $r=|b|/|a|$, then $r\in \R_+$ and $|a|^{-1}=\sqrt{1+r^2}$. The phases $\alpha ,\beta$ and $\gamma$ are arbitrary, because $\phi$ and the phases of $a$ and $b$ are arbitrary.
\end{proof}
The parametrization \eqref{TransMatParam} immediately gives the following invariance property of the transfer matrices.
\begin{corollary}
\label{CorPlane}
Denote by $\Plane$ the set of vectors $ x \in \C^2$ such that
\[
x=\genfrac{(}{)}{0pt}{0}{x_1}{x_2}\in \Plane \quad \Longleftrightarrow \quad |x_1|=|x_2|\, ,
\]
then for $z\in \T$ and $T{\scriptstyle (z)}\in \tau_z(\mathcal{U}_{ND})$ we have $T_z \Plane \subset \Plane$.
\end{corollary}
\begin{proof}
Without loss of generality we may assume
\[
x=\genfrac{(}{)}{0pt}{0}{1}{c} \quad \text{with}\quad |c|=1\,.
\]
Now, the statement follows from the fact that
\[
|\sqrt{1+r^2}e^{i\alpha}+re^{i\beta}c|=|re^{i\gamma}+\sqrt{1+r^2}e^{i(\beta +\gamma -\alpha)}c|\,.
\]
\end{proof}

Since $\tau_z$ is injective on $\mathcal{U}_{ND}$, a measure $\mu$ on $\UG $ induces a family of measures $\mu_z$ on $\tau_z(\mathcal{U}_{ND})$ via
\[
\mu_z(X)= \mu \circ\tau_z^{-1}(X)\quad ,\, X\subset \tau_z(\mathcal{U}_{ND}) \,.
\]
If $\mu$ is a probability measure on $\UG$ the $\mu_z$ are probability measures on $\tau_z(\mathcal{U}_{ND})$ if and only if $\mu(\mathcal{U}_{ND})=1$. In this case, with probability one, $U_{\omega}(N)$ is a direct sum of elements in $\mathcal{U}_{ND}$ such that the transfer matrices $T{\scriptstyle (z)}$ for $W_{\omega}(N)$ are well defined. Consequently, it is valid to assume that the resolvent $(W_{\omega}(N)-z)^{-1}$ can be expressed through transfer matrices $T{\scriptstyle (z)}$.

The case $\mu(\mathcal{U}_{ND})<1$ also leads to dynamical localization of $W_{\omega}(N)$. The reason is that coin operators $U_x,\,U_y\in \overline{\mathcal{U}_{ND}}$ cause reflections of the walking particle at sites $x$ and $y$, see \cite{Linden2009}, hence a particle starting in between $x$ and $y$ is strictly localized in the finite region between $x$ and $y$ for all times. The following lemma proves dynamical localization for $\mu(\mathcal{U}_{ND})<1$.
\begin{lemma}
\label{ZeroSet}
If $\mu(\mathcal{U}_{ND})=p<1$ the disordered quantum walk $W_{\omega}$ satisfies
\[
\Expect{|\Scp{\LocS[y]{\phi}}{W_\omega^t \LocS{\psi}}|} \,\leq\, (1-p)^{|x-y|+1}
\]
\end{lemma}
\begin{proof}
We assume without loss of generality $x<y$. Let $x<r\in\Z$ denote the site such that $U_r \in \overline{\mathcal{U}_{ND}}$ and $U_s\in \mathcal{U}_{ND}$ for all $x<s<r$. If $U_s\in \mathcal{U}_{ND}$ for all $s>x$ we set $r=\infty$. The scalar product satisfies
\[
|\Scp{\LocS[y]{\phi}}{W_\omega^t \LocS{\psi}}|=0\, ,\quad \forall y>r\,,\; \forall t>0\,
\]
and the probability that $U_r \in \overline{\mathcal{U}_{ND}}$ and $U_s\in \mathcal{U}_{ND}$ for all $x<s<r$ is exactly $(1-p)^{r-x-1}p$. Hence, the bound of the lemma follows by assuming the worst case, \ie\
\[
|\Scp{\LocS[y]{\phi}}{W_\omega^t \LocS{\psi}}|=1\,
\]
for all $y$ with $r\geq y$.
\end{proof}
The preceding lemma tells us that we may assume that all coins $U_x$ of the walk operator $\W$ satisfy $U_x\in \mathcal{U}_{ND}$ when proving dynamical localization for $\W$. This leads us to a simplification of the resolvent formula \eqref{ResolventFormula} for finite walk operators $\Wf$.
\begin{lemma}
\label{Decomposition}
Let $\W$ be a disordered quantum walk and $\Wf$ a finite restriction of it. Denote by $\Resf$ the resolvent of $\Wf$, then for $z\in\T$ there exist normalized vectors $\Phi_-$, $\Phi_+\in \Plane$ such that the matrix elements $\Resf{\scriptstyle (2x-i,2y-j)}$ for $i,j\in\{0,1\}$ obey
\begin{equation}
\label{ResolventDecomposition}
|\Res^N{\scriptstyle (2x-i,2y-j)}|=\frac{1}{2}\cdot \left\{\begin{array}{ccl}
|\Scp{\Phi_+}{\Trans{y-1}\ldots\Trans{x} \Phi_-}|^{-1}&,& \text{if }x< y\\
|\Scp{\Phi_+}{\Trans{x-1}\ldots \Trans{y} \Phi_-}|^{-1}&,& \text{if }x> y
\end{array}\right.
\,.
\end{equation}
The vectors $\Phi_\pm$ depend in a non-trivial way on $N$, $x$, $y$ and $z$.
\end{lemma}
\begin{proof}
For the moment we suppress the dependence on $z$ and denote by $\phi^N_-$ and $\phi^N_+$ the left respectively right compatible sequences for $\Resf$. Our choice of boundary conditions implies
\[
|\phi^N_-{\scriptstyle (-2N-1)} |=|\phi^N_-{\scriptstyle(-2N) }|\quad \text{and}\quad |\phi^N_+{\scriptstyle (2N+1) }|=|\phi^N_+{\scriptstyle (2N+2)} |\, ,
\]
and hence, by Corollary \ref{CorPlane} we get $|\phi^N_\pm{\scriptstyle (2M-1)} |=|\phi^N_\pm{\scriptstyle (2M)} |$ for all $-N\leq M\leq N+1$. Now, in the case $x\leq y$ the statement follows from \eqref{ResolventFormula} if we set
\[
\Phi_-=\frac{1}{\sqrt{2}|\phi^N_-{\scriptstyle (2x)}|}\genfrac{(}{)}{0pt}{0}{\phi^N_-{\scriptstyle (2x-1)}}{\phi^N_-{\scriptstyle (2x)}} \quad\text{and}\quad
\Phi_+=\frac{1}{\sqrt{2}|\phi^N_+{\scriptstyle (2y)}|}\genfrac{(}{)}{0pt}{0}{-\overline{\phi^N_+{\scriptstyle (2y)}}}{\overline{\phi^N_+{\scriptstyle (2y-1)}}}\, ,
\]
and in the case $x\geq y$ we have to exchange $x$ and $y$ in this definition.
\end{proof}

\begin{remark}\label{Remdecomposition}
	We note in particular that given two normalized vectors $\phi_L$, $\phi_R \in \Plane$ and a set of transfer matrices, then we can always construct some restricted walk operator such that the right hand side of \eqref{ResolventDecomposition} constitutes the resolvent of that unitary operator. This follows by choosing the right boundary conditions in the construction of the finite restriction outlined in \eqref{eqn_fini_rest}, since the mapping of unitary coins to transfer matrices can be reversed.
\end{remark}

We also need the following estimate on the difference between products of transfer matrices of length $n$ associated to two different elements of the unit circle.

\begin{lemma}
	\label{transprodcont}
	Let $n\in\N$, $\theta \in \T$ and assume that the transfer matrices $T_1,\ldots,T_n$ satisfy $|a_i|>\kappa$. Then
	\begin{align*}
		\Prob{\, \left| \NPNonInvTrans{n}{\theta} - \NPNonInvTrans{n}{\theta^\prime} \right| > \eta \,} \leq \frac{1}{\eta} C(n) |\theta - \theta^\prime|
	\end{align*}
	for all normalized vectors $v\in\C^2$.
\end{lemma}
\begin{proof}
By Markov's inequality, the inverse triangle inequality and the definition of the operator norm it is sufficient to prove
\[
\Expect{ \Norm{T_{n}{\scriptstyle (\theta)} \cdot\ldots\cdot T_{1}{\scriptstyle (\theta)}-T_{n}{\scriptstyle (\theta^\prime)} \cdot\ldots\cdot T_{1}{\scriptstyle (\theta^\prime)}}}\leq C(n) |\theta - \theta^\prime| \,.
\]
Writing the operator difference as a telescope sum and using again the triangle inequality we obtain the bound
\begin{eqnarray*}
\Norm{T_{n}{\scriptstyle (\theta)} \cdot\ldots\cdot T_{1}{\scriptstyle (\theta)}-T_{n}{\scriptstyle (\theta^\prime)} \cdot\ldots\cdot T_{1}{\scriptstyle (\theta^\prime)}}  &\leq & \\
\Norm{\Trans[\theta]{n}-\Trans[\theta^\prime]{n}}\cdot \Norm{T_{n-1}{\scriptstyle (\theta)} \cdot\ldots\cdot T_{1}{\scriptstyle (\theta)}} & +& \\
\Norm{\Trans[\theta^\prime ]{n}}\cdot \Norm{\Trans[\theta ]{n-1}-\Trans[\theta^\prime]{n-1}}\cdot \Norm{T_{n-2}{\scriptstyle (\theta)} \cdot\ldots\cdot T_{1}{\scriptstyle (\theta)}} &+ &\\
 \ldots & +&\\
 \Norm{\Trans[\theta^\prime ]{n}\ldots \Trans[\theta^\prime]{2}}\cdot \Norm{\Trans[\theta]{1}-\Trans[\theta^\prime]{1}}&&
\end{eqnarray*}
For products of transfer matrices we have the operator norm bound
\[
\Norm{\Trans[\theta]{k}\ldots\Trans[\theta]{1}}\leq \left(\frac{2}{\kappa}\right)^k\,,\quad \forall \theta \in \C\,, |\theta|=1\, ,
\]
which follows from the norm inequality $\Norm{A\cdot B}\leq \Norm{A}\cdot \Norm{B}$ and the fact that the singular values $\lambda_\pm$ of $\Trans[\theta]{i}$ for $|\theta|=1$ satisfy
\[
\lambda_\pm = \frac{1\pm |c_i|}{a_i}\leq \frac{2}{\kappa}\, .
\]
The assertion of the lemma with $C(n)=n\cdot 2^{n-1}/\kappa^{n}$ follows now from the equation $\Norm{\Trans[\theta]{i}-\Trans[\theta^\prime]{i}}=|\theta-\theta^\prime|/|a_i|$.
\end{proof}

We also need the following easy corollary to lemma \ref{transprodcont}.
\begin{corollary}
	\label{cortransprodcont}
  	Let $n\in\N$, $\theta \in \T$ and assume that the transfer matrices $T_1,\ldots,T_n$ satisfy $|a_i|>\kappa$. Then
	\begin{align*}
		\Prob{\,\Abs{\,\Abs{\Scp{v_1}{\NPNonInvProd{n}{\theta}v_2}} -\Abs{\Scp{v_1}{ \NPNonInvProd{n}{\theta^\prime} \; v_2}}\,} > \eta \,} \leq \frac{1}{\eta} C(n) |\theta - \theta^\prime|
	\end{align*}
	for all normalized vectors $v_1,v_2\in\C^2$.
\end{corollary}
\begin{proof}
  After applying Markov's and the inverse triangle inequality we can just use the Cauchy-Schwarz inequality for the scalar product. Since the vector $v_1$ is normalized we arrive at the same situation as in the proof of lemma \ref{transprodcont} and the whole argument carries through also in this case with the same bounds as before.
\end{proof}

In theorem \ref{thmtransfermatrices} we assumed that the group $\Gmuz$ generated by the transfer matrices $T{\scriptstyle (z)}\in \mathrm{supp}(\mu)$ is non-compact and possesses no reducible subgroup of finite index, for some $z = \theta_0 \in \T$. The reason is that we wish to apply results of \cite{Furstenberg1963} to show that with high probability the scalar product
\[
\Abs{\Scp{v_1}{\NPNonInvProd{n}{z}v_2}}
\]
grows exponentially for all normalized vectors $v_1$ and $v_2$. More precisely, we prove the following corollary building on \cite{Furstenberg1963} in appendix \ref{F\"{u}rstenberg}. To simplify notation we set $x=1$ and $y-2=n$ for the remaining part of the section.
\begin{corollary}
\label{TransNormGrow}
Let $\mu$ be a measure on $\UG$ and suppose that the group $\Gmuz$ generated by the induced measure $\mu_z$ is non-compact and possesses no reducible subgroup of finite index for some $z \in \C$, then the product
\[
\NPNonInvTrans{n}{z}
\]
grows exponentially in $n$ almost surely. Furthermore, there exists $\gamma, \sigma>0$ such that for every $\varepsilon > 0$ there is a $N \in \N$ such that
\[
\Prob{\Abs{\Scp{v_1}{ \NPNonInvProd{n_0}{z}\; v_2}} > \ex{ (\gamma -\varepsilon )\, n}} > 1 - \ex{-\sigma \, n}
\]
holds for all $n\geq N$ and all normalized vectors $v_1$, $v_2 \in \C^2$.
\end{corollary}

This gives us a first handle on the exponential decay. Note however, that the above corollary depends on $z$, but we need uniform estimates for some open arc of the unit circle. But using lemma \ref{transprodcont}, we can derive at least an initial scale estimate, valid for some fixed $n_0$.

\begin{proposition}
	\label{initiallengthscale}
	Let $\theta_0 \in \T$, $n\in\N$ and suppose the group $G_{\mu_{\theta_0}}$ is non-compact and possesses no reducible subgroup of finite index. Then there exist $\sigma_0,\gamma_0\in \R_+$, a natural number $n_0\geq n$ and an open arc of the unit circle $I_\delta(\theta_0)$ centered around $\theta_0$ with arc length $2\delta > 0$, such that
	\begin{align*}
		\Prob{\, \text{for all } \theta \in I_\delta\,:\, \Abs{\Scp{v_1}{ \NPNonInvProd{n_0}{\theta}\; v_2}} > \ex{ \gamma_0 n_0} } > 1 - \ex{ - \sigma_0 n_0} \, .
	\end{align*}
	Furthermore, the constants $\sigma_0,\gamma_0\in \R_+$ are independent of the normalized vectors $v_1$, $v_2$.
\end{proposition}
\begin{proof}
According to corollary \ref{TransNormGrow}, there exist $\gamma,\sigma >0$ such that for $\varepsilon>0$ there is $N\in\N$ such that for  $m\geq N$
\[
\Prob{\Abs{\Scp{v_1}{ \NPNonInvProd{m}{\theta}\; v_2}} > \ex{ (\gamma -\varepsilon )\, m}} > 1 - \ex{-\sigma \, m} \,.
\]
If $n\geq N$ we set $n_0=n$, otherwise we set $n_0=N$. We can assume $\mu(\mathcal{U}_{ND})=1$, see \eqref{NonOffDiagonal} and lemma \ref{ZeroSet}. Hence, for $\varepsilon^\prime>0$ there exists $\kappa>0$ such that
\[
\Prob{\exists \,i\in \{1,\ldots ,n_0\}\,:\,|a_i|< \kappa }<\varepsilon^\prime \, .
\]
By corollary \ref{cortransprodcont} we have for arbitrary $\theta\in\T$
\[
\Prob{ \Abs{\Scp{v_1}{ \NPNonInvProd{n_0}{\theta}\; v_2}} > \ex{ (\gamma-\varepsilon) n_0} -\eta} > 1 - \ex{ - \sigma n_0}-\frac{C(n_0)}{\eta}|\theta-\theta_0|-\varepsilon^\prime \,,
\]
with a coefficient $C(n_0)$ possibly depending on $\varepsilon^\prime$. First, we choose $\varepsilon>0$ such that $\gamma-2\varepsilon >0$ and set $\gamma_0=\gamma-2\varepsilon$. Then we choose $\eta$ such that $\ex{(\gamma-\varepsilon)n_0}-\eta>\ex{\gamma_0 n_0}$. The constants $\sigma_0$ and $\varepsilon^\prime$ are chosen in a way that assures $\ex{-\sigma_0 n_0}>\ex{-\sigma n_0}+\varepsilon^\prime$. Finally, a bound $\delta>0$ for the distance $|\theta_0-\theta|$ can be determined from the constant $C(n_0)$ depending on $\kappa$.
\end{proof}

\subsection{Thouless Formula}
\label{sec_thouless}

A Thouless formula relates the density of states with the Lyapunov exponent. So we begin this section by defining the density of states for a disordered quantum walk in the usual way as a limit of measures corresponding to the  finite dimensional approximations of the walk operator as defined in section \ref{secfinrestrictions}

\begin{lemma}
For every bounded and continuous function $f:\T \rightarrow \C$ we have with probability one that
\begin{align}\label{eq_prop_dos}
 \lim_{N\rightarrow\infty} \frac{1}{4(N+1)}\tr(\chi_{N}f(\W)) = \frac{1}{2}\Expect{\Scp{\LocSe[0]{1}}{f(\W)\LocSe[0]{1}}+ \Scp{\LocSe[0]{2}}{f(\W) \LocSe[0]{2}}}\,;
\end{align}
where $\chi_{N}$ is defined in terms of projections onto the localized states $P_{\LocSe[l]{\alpha}}$ as
\begin{align*}
  \chi_{N} = \sum_{l=-N}^N\sum_{\alpha} P_{\LocSe[l]{\alpha}} + P_{\LocSe[-(N+1)]{2}}+P_{\LocSe[N+1]{1}}=:\sum_{l=-N}^N\sum_{\alpha} P_{\LocSe[l]{\alpha}} + P_{boundary}
\end{align*}
\end{lemma}
\begin{remark}
  The reason why we do not restrict \W\ only to the lattice sites $-N$ to $N$, but include the states $\LocSe[-(N+1)]{2}$ and $\LocSe[N+1]{1}$ in the definition of $\chi_N$ is that we want to compare it to the restricted operator $\hat\W(N)$, which is only unitary if we include these matrix elements (see eq. \eqref{eq_uniRest}).
\end{remark}
\begin{proof}
 For every fixed $f$ we can evaluate the trace on the left hand side of \eqref{eq_prop_dos} in the standard basis of localized states $\LocSe[i]{j}$ and get
  \begin{align*}
     \frac{1}{4(N+1)}\tr(\chi_{N}f(\W)) =&   \frac{1}{4(N+1)} (\sum_{l,j} \Scp{\LocSe[l]{j}}{f(\W)\LocSe[l]{j}} + \tr(P_{boundary} f(\W)))\; .
  \end{align*}
  Since $f$ is bounded by assumption, the second term on the right hand side vanishes when we make $N$ large.
  As in the self adjoint case one can now show that Birkhoff's theorem applies to the random variables $\{X_l:=\sum_{j} \Scp{\LocSe[l]{j}}{f(\W)\LocSe[l]{j}} \}$, so for fixed $f\in C(\T)$ we have a subset $\Omega_f\subset\Omega$ of measure one for which equation \eqref{eq_prop_dos} holds. If we in particular look at a dense countable subset $C_0\subset C(\T)$ we can consider the intersection of all the $\Omega_f$ with $f\in C_0$. Since this is a countable intersection of sets of full measure it has measure one as well and since it is dense in $C(\T)$, this finishes the proof.
\end{proof}

Equipped with this limit we can now define the density of states for a disordered quantum walk via the Riesz-Markov representation theorem:

\begin{definition}[Density of states]\label{def_DOS}
The density of states of a disordered quantum walk is the measure $\vartheta$ on the unit circle \T\ defined by
\begin{align*}
  \int_\T f(\theta)\; \vartheta(d\theta) = \frac{1}{2} \Expect{\Scp{\LocSe[0]{1}}{f(\W)\LocSe[0]{1}}+ \Scp{\LocSe[0]{2}}{f(\W) \LocSe[0]{2}}}
\end{align*}
\end{definition}

The unitary restriction $\W(N)$ also defines a measure $\vartheta_N$ on \T\ via the Riesz-Markov representation theorem
\begin{align}\label{eq_dos_fm}
\int_\T f(\theta)\; \vartheta_N(d\theta)=\frac{1}{4(N+1)}\tr(f(\W(N)))
\end{align}

Later in this sections we connect these measures to the Lyapunov exponent in order to proof the H\"{o}lder continuity of the integrated density of states as it is done in \cite{Joye2004}. Therefore we need to establish a relation between the unitary restriction $\W(N)$ and the non-unitary restriction $\chi_{N}\W\chi_N$. Such a relation is given by the following lemma, which is inspired by \cite{Joye2004}.

\begin{lemma}
  With the definitions from equation $\eqref{eq_uniRest}$ we have that for every bounded and continuous function $f:\T\rightarrow \C$
  \begin{align*}
    \lim_{N\rightarrow\infty} \frac{1}{4(N+1)}\left(\tr(f(\W(N)))-\tr(\chi_N f(\W)\chi_N\right)=0
  \end{align*}
  holds.
\end{lemma}
\begin{proof}
   Clearly it is enough to prove the above statement for functions with supremum norm equal to one. Since $f(z)$ is continuous, for every $\varepsilon$ we can find a trigonometric polynomial $p_\varepsilon(z)$ of finite degree $\tau(\varepsilon)$, such that $p_\varepsilon(z)$ approximates $f(z)$ uniformly on \T\ up to $\varepsilon$. Using the identity $(\W(N))^n = \chi_N(\hat\W)^n \chi_N$ for $n\in\N$ and adding and subtracting $p_\varepsilon(z)$ in the trace difference we get
  \begin{align}\label{eq_doseq}
	  \begin{split}
    |\tr(&f(\W(N))- \chi_N f(\W)\chi_N)|= \\
    &= |\tr(\chi_N(p_\varepsilon(\hat\W)-p_\varepsilon(\W)\chi_N)) +\tr((f-p_\varepsilon)(\W(N))-\chi_N(f-p_\varepsilon)(\W)\chi_N) |\\
    &\leq |\tr(\chi_N(p_\varepsilon(\hat\W)-p_\varepsilon(\W)\chi_N)| + 8(N+1)\varepsilon
  	\end{split}
  \end{align}
  Where in the second step we used that $p_\varepsilon(z)$ approximates $f(z)$ together with the restriction by $\chi_N$ and thereby reduced the problem to the difference of powers of \W\ and $\hat \W$. By induction we get the following identity for such differences
  \begin{align*}
    \chi_N(\W^k-\hat W_\omega^k)\chi_N=\sum_{i=0}^{k-1}  \chi_N \W^i(\W-\hat\W)\hat\W^{k-(i+1)}\chi_N
  \end{align*}
  which holds also for negative powers $k$ since we can then substitute the inverses by adjoints. Taking the trace we find for the summands on the right hand side
\begin{align*}
  | \tr(\W^i(\W-\hat\W)\hat\W^{k-(i+1)}\chi_N)|\leq \Norm[1]{(\W-\hat\W)\hat\W^{k-(i+1)}\chi_N}
   \leq \Norm[1]{\W-\hat\W}\;
\end{align*}
 In the last step we used that since $\hat\W \chi_N$ is bounded and that since $W$ and $\hat W$ differ at most in the two coins $U_{N+1}$ and $U_{-(N+1)}$ their difference is of finite rank we can apply
 \begin{align*}
   \Norm[1]{AB}\leq\Norm[1]{A}\Norm[Op]{B}
 \end{align*}
 for $A$ a compact and $B$ a bounded operator \cite{traceideals}. Substituting everything into equation \eqref{eq_doseq} we get therefore
 \begin{align*}
   \frac{1}{4(N+1)}   |\tr(f(\W(N))- \chi_N f(\W)\chi_N)| \leq  \frac{2\tau(\varepsilon)\Norm[1]{(\W-\hat\W)}}{4(N+1)} +2\varepsilon\;,
 \end{align*}
 where $\tau(\varepsilon)$ is the degree of $p_\varepsilon(z)$. Since $\W-\hat\W$ is of finite rank its trace norm is bounded uniformly in $N$, the assertion is proven.
\end{proof}

The following proposition describes conditions on the eigenvalues of the unitary restriction $\W(N)$.

\begin{proposition}
  For $z\in\C$ we define the vectors
\begin{align}\label{eq_phiRbot}
   \phi_R^\bot(z)= \VecD{-1}{\bar z e^{i\eta^{\bf R}}} && \phi_L(z)= \VecD{1}{z e^{-i\eta^{\bf L}}}
  \end{align}
  and the spectral polynomial
\begin{align}
  \label{SpectralPolynomialEquation}
    p_N(z)=z^{2N+1}\Scp{\phi_R^\bot(z)}{T_N\cdot\ldots\cdot T_{-N} \phi_L(z)}\,.
  \end{align}
Then $z\in \T$ is an eigenvalue of the unitary restriction $\W(N)$ of the Walk operator \W\ if and only if $p_N(z)=0$, hence all zeros of $p_N$ lie in $\T$.
\end{proposition}
\begin{proof}
  Because of the structure of the restriction $\W(N)$ the eigenvalue equation $(\W(N)-z)\phi=0$ implies, as in the infinite case, the following relations via transfer matrices between the components of the eigenvector $\phi$
  \begin{align}\label{eq_trans_ew}
    \VecD{\phi(2l+1)}{\phi(2l+2)}=\Trans{l}\VecD{\phi(2l-1)}{\phi(2l)}\; .
  \end{align}
  In addition we get the following relations from the boundary conditions imposed by the flip operations
  \begin{align*}
   z\phi(-2N-1)=e^{i\eta^{\bf L}}\phi(-2N)& \;&\text{and}&\; &z\phi(2(N+1))=e^{i\eta^{\bf R}}\phi(2N+1)\;,
  \end{align*}
  which can be seen from \eqref{eqn_fini_rest}. Therefore the first and the last two components of $\phi$ have to be proportional to the normalized vectors
  \begin{align*}
     \phi_L(z)= \frac{1}{\sqrt{2}}\VecD{1}{z e^{-i\eta^{\bf L}}}& \;&\text{and}&\; & \phi_R(z)= \frac{1}{\sqrt{2}}\VecD{z e^{-i\eta^{\bf R}}}{1}\;,
  \end{align*}
  respectively. Because of relation \eqref{eq_trans_ew} we can express $\phi_R(z)$ through $\phi_L(z)$ via a product of $2N+1$ transfer matrices. Noting that the vector  $\phi_R^\bot(z)$ as defined in equation \eqref{eq_phiRbot} is orthogonal to $\phi_R(z)$ the vector $\phi_L(z)$ has to fulfill the desired relation
  \begin{align*}
    \Scp{\phi_R^\bot(z)}{T_N\cdot\ldots\cdot T_{-N} \phi_L(z)}=0 \;,
  \end{align*}
  in order for $z$ to be an eigenvalue.
\end{proof}

The Thouless formula now follows from \eqref{SpectralPolynomialEquation}. We utilize the following linear factor decomposition of the spectral polynomial
\begin{equation}
\label{poly}
p_n(z)=z^{2N+1}\Scp{\phi^\bot_R(z)}{T_N\cdot\ldots\cdot T_{-N} \phi_L(z)}=C_N\cdot \prod_{l=1}^{4(N+1)} (z-e^{i\theta_l})\,,
\end{equation}
where the $\theta_l$ are the eigenvalues of the finite dimensional operator $\W(N)$ as defined in Section \ref{secfinrestrictions}. The coefficient $C_N$ is determined from the transfer matrices $T_{-N}$ up to $T_N$, more precisely, it is given by
\[
C_N=e^{-i(\eta^{\bf L}+\eta^{\bf R})}\left(\prod\limits_{l=-N}^N a_l\right)^{-1}
\]
which follows from combinatorial considerations. Now, for $|z| \neq 1$ we take the logarithm of the absolute value of (\ref{poly}) and divide by $4(N+1)$ to get
\[
\frac{2N + 1}{4(N+1)} \log|z| + \frac{\log(|\Scp{\phi^\bot_R(z)}{T_N\cdot\ldots\cdot T_{-N} \phi_L(z)}|)}{4(N+1)}=-\sum_{l=-N}^{N}\frac{\log(|a_l|)}{4(N+1)}+\sum_{l=1}^{4(N+1)}\frac{\log(|z-e^{i\theta_l}|)}{4(N+1)}\,.
\]
Note that according to the arguments in section \ref{InducedMeasure} we may assume that all $T_l$ are well-defined, \ie\ $a_l\neq0$ for all $l$. Hence, by the law of large numbers the first term converges to the expectation of $\log (|a|)$ and with $|z| \neq 1$ the second term can be written as
\[
\sum_{l=1}^{4(N+1)}\frac{\log(|z-e^{i\theta_l}|)}{4(N+1)}=\frac{1}{4(N+1)}\tr (\log (|z-\W(N)|))=\int \log (|z-e^{\ii \lambda}|)\, \vartheta_N(d\lambda)\, ,
\]
with the density of states $\vartheta_N$ as defined in \eqref{eq_dos_fm}. The right hand side of this equation converges to the corresponding expression with density of states $\vartheta$. In summary, we have the following expression
\[
\frac{2N + 1}{4(N+1)} \log|z| + \frac{\log|\Scp{\phi^\bot_R(z)}{T_N\cdot\ldots\cdot T_{-N} \phi_L(z)}|}{4(N+1)}\mathop{\longrightarrow}\limits_{N\rightarrow \infty}\frac{1}{2}\Expect{\log |a|}+ \int \log (|z-e^{\ii \lambda}|)\vartheta(d\lambda)\,.
\]
According to appendix \ref{F\"{u}rstenberg} the second term on the left hand side of this equation converges almost surely to the Lyapunov exponent $\gamma(z)$, which finally proves the Thouless formula
\begin{align}
\label{eqthouless}
\gamma(z) =  \int \log (|z-e^{\ii \lambda}|)\vartheta(d\lambda) + \frac{1}{2} \Expect{\log |a|} - \frac{1}{2} \log |z| \,,
\end{align}
for $z \in \C$, $|z| \neq 1$. As shown in appendix \ref{F\"{u}rstenberg}, the Lyapunov exponent equals the limit of a deceasing family of subharmonic functions, since it is defined as the logarithm of the norm of some holomorphic matrix valued function. The right hand side of the above equation constitutes also a subharmonic function, and it thus follows from arguments which can be found in the classical work of Craig and Simon \cite{craigsimon}, that the Thouless formula actually remains valid for $z \in \T$, \ie\ $|z| = 1$.

The arguments used by Craig and Simon \cite{craigsimon} also prove that the integrated density of states, as defined in the next section, is continuous on $\T$ as long as the Lyapunov exponent is non-negative on $\T$. In fact, the integrated density of states is H\"{o}lder continuous on $\T$, which is proven in the next section.

\subsection{H\"{o}lder continuity of the Integrated Density of States}
An important property we are going to use later on is the H\"{o}lder continuity of the integrated density of states $\Dens$ on $\T$, which we define for $z\in\T$ in the following way
\begin{equation}
\label{DensityOfStatesDefinition}
\Dens(z) :=\int_0^{\arg (z)}\vartheta (d\lambda)
\end{equation}
This helps us to prove an estimate for the probability of cases where there is a value of the spectral parameter $z\in \T$ such that the denominator of the resolvent formula \eqref{ResolventDecomposition} is small, called Wegner estimate, see section \ref{WegnerEstimate}. The integrated density of states $\Dens$ inherits the H\"{o}lder-continuity from the Lyapunov-exponent $\gamma (z)$, therefore we need the following lemma.
\begin{lemma}\label{HCont}
Let $\W$ be a disordered quantum walk satisfying the assumptions of theorem \ref{thmtransfermatrices}, then the corresponding Lyapunov exponent $\gamma (z)$ is H\"{o}lder continuous on $\T$, \ie\ there exists a finite constant $C$ and a strictly positive number $\alpha$ such that
\[
\left| \gamma (z) - \gamma (z')\right| \leq C\cdot \left| z-z'\right|^\alpha \quad , \, z,z'\in \T\,.
\] 
\end{lemma}
A proof of this statement for hamiltonian systems can be found in \cite{Carmona}. Lemma \ref{HCont} follows from a careful adaption of this proof to the setting considered in this paper. Although this is straightforward and contains no new ideas we give a self-contained proof of lemma \ref{HCont} in appendix \ref{F\"{u}rstenberg}.

The H\"{o}lder continuity of $\gamma (z)$ in turn implies the same property for the integrated density of states $\Dens$, this is due to the fact that $\gamma(z)$ and the density of states $\vartheta $ are related via the Thouless formula \eqref{eqthouless}.
\begin{proposition}
	\label{lemcontiIDS}
Let $\W$ be a disordered quantum walk satisfying the assumptions of theorem \ref{thmtransfermatrices}, then the corresponding integrated density of states $\Dens$ is H\"{o}lder continuous on $\T$, \ie\ there exists a finite constant $K$ and a strictly positive number $\beta$ such that
\[
|\Dens (z) - \Dens (z')|\leq K|z-z'|^\beta \quad , \, z,z'\in \T\,.
\]
\end{proposition}
\begin{proof}

For $a\in \T$ and $\phi \in [0,\pi]$ we define the function $\hat{\Dens}_{a,\phi} (z):= \Dens (z)\chi_{\{\theta \in \C\,:\, |\arg (\theta\cdot\bar a)|\leq \phi\}}(z)$ and apply the Cauchy transform for values of $z\in \mathbb{D}$ satisfying $|\arg (z\cdot \bar a)|\leq \phi/2$. Note that according to \cite{craigsimon} the integrated density of states $\Dens$ is continuous almost everywhere on $\T$, so for $|z|<1$ we can apply integration by parts to obtain
\begin{align}
\label{CauchyDensStates}
\begin{split}
K\hat {\Dens}_{a,\phi} (z)& = \int\limits_\T \frac{\hat {\Dens}_{a,\phi} (e^{i \lambda})}{1-e^{-i \lambda } z}m(d\lambda )\\
&= \Dens (a e^{i \phi}) \frac{\log (1- a e^{i\phi} \bar z)}{2\pi i}-\Dens (a e^{-i\phi}) \frac{\log (1- a e^{-i\phi} \bar z)}{2\pi i} + \\
&+\int\limits_\T \log (1-e^{i\lambda} \bar z)\vartheta (d\lambda )-\int\limits_{|\arg (\theta\cdot\bar a)|> \phi}\log (1-e^{i\lambda} \bar z)\vartheta (d\lambda)
\,.
\end{split}
\end{align}
In this formula we choose the logarithm as $\log (re^{i\theta})=\log r +i\theta$ with $\theta\in [-\pi,\pi)$ if $|\arg a|\leq\phi$ and with $\theta\in[0,2\pi)$ if $|\arg a|>\phi$. Since $\Dens$ is almost everywhere continuous on $\T$ its Fourier series converges point wise almost everywhere to the exact value. The Cauchy transform projects the Fourier series of $\Dens$ to the positive powers, see \eg\ \cite{bookcauchytransform}, and since $\Dens$ is a real-valued function this implies the relation $\Dens (z) = K\Dens (z) + \overline{K\Dens (z)}-K\Dens (0) $. Hence, the integrated density of states is H\"{o}lder-continuous if the real part of each of the four terms in \eqref{CauchyDensStates} is H\"{o}lder-continuous in the limit $|z|\rightarrow 1$.  This is obvious for the first two terms and the last term. The third term may be rewritten as
\[
\int\limits_\T\log (1-e^{i\lambda} \bar z)\vartheta (d\lambda) = \int\limits_\T\log (|z-e^{i\lambda} |)\vartheta (d\lambda) + i\int\limits_\T \arg (1-e^{i\lambda}\bar z) \vartheta (d\lambda) -i \lambda_0\,.
\]
The constant $\lambda_0$ depends on the branch of the logarithm we have chosen and is real, hence, the real part of this expression is just the first integral, which is H\"{o}lder-continuous by the Thouless formula and Lemma \ref{HCont}.
\end{proof}

\subsection{Wegner Estimate}
\label{WegnerEstimate}

We proceed by proving a Wegner-type bound for one-dimensional disordered quantum walks. We closely follow the arguments in \cite{Carmona87} and \cite{damanik02}. In fact, all arguments more or less take over and there are only minor modifications to be made for the unitary case considered here. Let us state the result first and then adapt the proof. Recall the definition of $I_\varepsilon(\theta_0)$ as the open arc of the unit circle centered around $\theta_0$ with elements whose phases differ from the phase of $\theta_0$ only up to $\varepsilon$.

\begin{proposition}[Wegner bound]\label{wegnerbound}
	Suppose that the single site coin distribution $\mu$ fulfills the condition stated in \ref{thmtransfermatrices}. Then there exist for every $0 < \beta < 1$ and every $\sigma > 0$ a natural number $n_0 = n_0(\beta,\sigma)$ and $\alpha = \alpha(\beta,\sigma) > 0$ such that
	\begin{align*}
		\Prob{\; \exists\, \theta\, \in \, I_{e^{- \sigma (2N+1)^\beta}}(\theta_0) \,:\, |\Scp{\phi^{\bf R}}{T_{N}{\scriptstyle (\theta)}\cdot\ldots\cdot T_{-N}{\scriptstyle (\theta)}\phi^{\bf L}}| < e^{- \sigma (2N+1)^\beta } } \;\leq\; e^{-\alpha N^\beta} 
	\end{align*}
	holds for all $\theta_0 \in \T$ and $N\in\N$ with $N > n_0$, and normalized vectors $\phi^{\bf L,R} \in \Plane$.
\end{proposition}

The results of the last section are crucial for the proof of our Wegner-type bound. We need the following lemma, which is basically a consequence of the H\"{o}lder continuity of the integrated density of states. It bounds the probability to find an eigenvalue of a localized eigenfunction in small sections of the unit circle, implying that the eigenvalues cannot clump together arbitrarily. Let $\W(0,N)$ refer to the unitary restriction of $\W$ constructed in the last section, acting on the $2N+3$ lattice sites centered around the origin (which we indicate by 0). Recall that its image however is of dimension $4(N+1)$.

\begin{lemma}
	\label{lemwegner}
	Suppose that the integrated density of states of the disordered walk operator $\W$ is H\"{o}lder continuous. Then there exist $\rho > 0$ and $C > 0$ such that for each $\theta_0 \in \T$ and for all $0 < \varepsilon < 1$ we have
	\begin{align*}
		&\mathbbm{P}\left(\; \exists\, \theta\, \in \, I_\varepsilon(\theta_0) \text{ and } \varphi \in \C^{4(N+1)} \, \norm{\varphi} = 1 \, ,  \text{ such that }  \right. \\ 
		&\left. (\W(0,N) - \theta)\varphi = 0 \text{ and } (|\varphi(-2N-1)|^2 + |\varphi(2N+2)|^2)^{\frac{1}{2}} \,<\, \varepsilon \; \right) \;\leq\; 4(N+1) C \varepsilon^\rho \;.
	\end{align*}
\end{lemma}

\begin{proof}
	Let $E(0,N,\theta_0,\varepsilon)$ denote the event whose probability, denoted by $p$, we would like to bound. Recall that it is defined with respect to the walk operator $\W(0,N)$ acting on the $2N+3$ lattice sites centered around the origin. Now let $E(k,N,\theta_0,\varepsilon)$, $k \in \N$, be the corresponding event for the unitary restriction $\W(k \cdot 2(N+1),N)$ of $\W$ on the $2N+3$ lattice sites centered around $k \cdot 2(N+1)$, constructed in the same way like $\W(0,N)$. Since the single site coin distribution is i.i.d on the different lattice sites, we get due to the ergodic properties of the system that
	\begin{align*}
		p \;=\; \lim_{L \to \infty}\,\frac{1}{2L+1} \; \#\{\,k\,\in\,(-L,\dots,L)\,:\, E(k,N,\theta_0,\varepsilon)  \text{ occurs }\,\} \;.
	\end{align*}
	Now choose $L \in \N$ and let $k_1,\dots,k_j \in \{-L,\dots,+L\}$ be such that the event $E(k_l,N,\theta_0,\varepsilon)$ occurs, and let $\theta_{k_l} \in \T$ and $\varphi_{k_l}$ be the corresponding eigenvalue resp. eigenfunction of $\W(k_l \cdot 2(N+1),N)$. Let us extend these functions to the whole of $\C^{(2L+1) \cdot 4(N+1)}$ by setting them equal to zero outside their domain of definition, and denote these extensions by $\tilde{\varphi}_{k_l}$. By construction, these functions form an orthonormal system. Consider now the unitary restriction $\W(0,L \cdot 2(N+1)+N)$, having image of dimension $(2L+1) \cdot 4(N+1)$. By the definition of the events $E(k_l,N,\theta_0,\varepsilon)$ we have that
	\begin{align*}
		\norm{(\W(0,L\cdot 2(N+1)) - \theta_{k_l}) \tilde{\varphi}_{k_l}} \;\leq\;\varepsilon \qquad \forall \, k_l \;.
	\end{align*}
	According to \cite{simontempleineq}, which is easily adapted to the unitary case, we thus have that the number of eigenvalues of $\W(0,L \cdot 2(N+1))$ in the open arc $I_\varepsilon(\theta_0)$ is bounded from below by $j$. But this number can be bounded using the integrated density of states $\mathcal{N}$. We conclude that
	\begin{align*}
		p \; &=\; \lim_{L \to \infty}\,\frac{1}{2L+1} \; \#\{\,k\,\in\,(-L,\dots,L)\,:\, E(k,N,\theta_0,\varepsilon)  \text{ occurs }\,\} \\
		&\leq \;  \lim_{L \to \infty}\,\frac{1}{2L+1} \; \#\{ \text{ eigenvalues of } \W(0,L\cdot 2(N+1)+N) \text{ in }  I_\varepsilon(\theta_0) \;\} \\
		&=\; 4(N+1) \lim_{L \to \infty}\,\frac{1}{(4(N+1))(2L+1)} \; \#\{ \text{ eigenvalues of } \W(0,L\cdot 2(N+1)+N) \text{ in }  I_\varepsilon(\theta_0) \;\} \\
		&\leq\; 4(N+1) \mathcal{N}(I_\varepsilon(\theta_0)) \;\leq\; C 4(N+1) \varepsilon^\rho \;.
	\end{align*}
\end{proof}

We proceed with the proof of the Wegner estimate. Again, we closely follow \cite{Carmona87} and \cite{damanik02} and just briefly sketch the arguments.

\begin{proof}[Proof of proposition \ref{wegnerbound}]
	To begin with, note that it follows from corollary \ref{TransNormGrow} using routine arguments, that there exist $\alpha > 0$, $\delta > 0$, $n_0 \in \N$ such that we have for all $\theta \in \T$, $n>n_0$, normalized vectors $v$,
	\begin{align}\label{eqwegninvexp}
	\Expect{\norm{T_{n}{\scriptstyle (\theta)}\cdot\ldots\cdot T_{1}{\scriptstyle (\theta)}v}^{-\delta}}\; \leq\; e^{-\alpha_1 n} \;.
	\end{align}
For a positive number $\tau$ let $n_N = [\tau N^\beta] + 1$, where $[.]$ indicates the integer part. Furthermore, let $\kappa$ be equal to $\tau\alpha_1/2\delta$ and define $\varepsilon = \exp(-\sigma (2N+1)^\beta)$. Following \cite{Carmona87} and \cite{damanik02}, we see that after introducing the events
	\begin{align*}
		A_\kappa(\theta,N) = \left\{ \Norm{T_{-N+n_N}{\scriptstyle (\theta)}\cdot\ldots\cdot T_{-N}{\scriptstyle (\theta)}\phi^{\bf L}} > e^{\kappa N^\beta} \right\} \\
		B_\kappa(\theta,N) = \left\{ \Norm{T_{N-n_N}^{-1}{\scriptstyle (\theta)}\cdot\ldots\cdot T_{N}^{-1}{\scriptstyle (\theta)}\phi^{\bf R}} > e^{\kappa N^\beta} \right\} \,,
	\end{align*}
	we can upper bound the probability
	\begin{align*}
		&\Prob{\; \exists\, \theta\, \in \, I_\varepsilon(\theta_0) \,:\, |\Scp{\phi^{\bf R}}{T_{N}{\scriptstyle (\theta)}\cdot\ldots\cdot T_{-N}{\scriptstyle (\theta)}\phi^{\bf L}}| < \varepsilon } \;\leq\; (i) \,+\, (ii) \,+\, (iii) \,+\, (iv) \;
	\end{align*}
	by four terms $(i)$-$(iv)$, which are defined as follows,
	\begin{itemize}
		\item[] \begin{align*}
				(i) \;=\; \Prob{\; \exists\, \theta\, \in \, I_\varepsilon(\theta_0) \,:\, |\Scp{\phi^{\bf R}}{T_{N}{\scriptstyle (\theta)}\cdot\ldots\cdot T_{-N}{\scriptstyle (\theta)}\phi^{\bf L}}| < \varepsilon \cap \bigcap_{\theta \in  I_{3\varepsilon}(\theta_0)} (A_{\frac{\kappa}{2}}(\theta,N) \cap B_{\frac{\kappa}{2}}(\theta,N)) } \, 
			\end{align*}
		\item[] \begin{align*}
				(ii) \;=\; \Prob{A_\kappa(\theta_0,N) \cap B_\kappa(\theta_0,N) \cap \bigcup_{\theta \in   I_{3\varepsilon}(\theta_0)} \left(A_{\frac{\kappa}{2}}(\theta,N)\right)^c } \, 
			\end{align*}
		\item[] \begin{align*}
				(iii) \;=\; \Prob{A_\kappa(\theta_0,N) \cap B_\kappa(\theta_0,N) \cap \bigcup_{\theta \in   I_{3\varepsilon}(\theta_0)} \left(B_{\frac{\kappa}{2}}(\theta,N)\right)^c } \,
			\end{align*}
		\item[] \begin{align*}
				(iv) \;=\; \Prob{\left(A_\kappa(\theta_0,N)\right)^c} \,+\, \Prob{\left(B_\kappa(\theta_0,N)\right)^c}
			\end{align*}
	\end{itemize}
	Using Markov's inequality and equation \eqref{eqwegninvexp}, we directly get
	\[
	(iv) \;\leq\; 2 e^{-\frac{1}{2}\tau(N)^\beta} \;.
	\]
	Now assume that the event whose probability is $(i)$ occurs. The transfer matrices $T_{i}{\scriptstyle (\theta)}$ as well as the vectors $\phi^{\bf L,R}$ correspond to a unitary restriction $\W(N)$ of some walk operator. Recall that for any normal operator $A$, we have
	\[ \norm{(A-z)^{-1}} \;\leq\; \frac{1}{\inf_{\lambda \in \sigma(A)} |\lambda - z|}, \]
	which together with lemma \ref{Decomposition}  and remark \ref{Remdecomposition} implies that there exists $\theta \in I_{3\varepsilon}(\theta_0) \cap \sigma(\W(N))$. Let $\varphi \in \C^{4(N+1)}$ be the corresponding normalized eigenfunction of $\W(N)$. Since we assumed that both events $A_\frac{\kappa}{2}(\theta,N)$ and $B_\frac{\kappa}{2}(\theta,N)$ occur, we have that 
	\begin{align*}
		&\varphi(-2N-1) \,\leq\, \frac{1}{\sqrt{2}} e^{-\frac{1}{2}\kappa N^\beta} \,\\
		&\varphi(2N+2)) \,\leq\, \frac{1}{\sqrt{2}} e^{-\frac{1}{2}\kappa N^\beta} \;.
	\end{align*}
	Thus, we have that 
	\begin{align*}
		(i) \;\leq\; &\mathbbm{P}\left(\; \exists\, \theta\, \in \, I_{3\varepsilon}(\theta_0) \text{ and } \varphi \in \C^{4(N+1)} \, \norm{\varphi} = 1 \, ,  \text{ such that }  \right. \\ 
		&\left. (\W(0,N) - \theta)\varphi = 0 \text{ and } (|\varphi(-2N-1)|^2 + |\varphi(2N+2))|^2)^{\frac{1}{2}} \,<\,  e^{-\frac{1}{2}\kappa N^\beta}  \; \right) \\
		\;\leq\; &4(N+1) C (\max(3\varepsilon, e^{-\frac{1}{2}\kappa N^\beta}))^\rho \;.
	\end{align*}
	To bound $(ii)$ and $(iii)$, we note that for $\theta\in\T$
	\begin{align}\label{eqdifftransdifftheta}
		T_{i}{\scriptstyle (\theta)} - T_{i}{\scriptstyle (\theta^\prime)} \;=\; \frac{1}{a_i}  \begin{pmatrix}\det(U_i) (\bar{\theta} - \bar{\theta}^\prime) & 0 \\0 & (\theta - \theta^\prime)\end{pmatrix}  \;.
	\end{align}
	If the events $A_\kappa(\theta_0,N) \cap \left(A_{\frac{\kappa}{2}}(\theta,N)\right)^c$ take place we have, if $N$ is large enough,
	\begin{align*}
		\left|\Norm{T_{-N+n_N}{\scriptstyle (\theta)}\cdot\ldots\cdot T_{-N}{\scriptstyle (\theta)}\phi^{\bf L}} 
		\;-\; \Norm{T_{-N+n_N}{\scriptstyle (\theta^\prime)}\cdot\ldots\cdot T_{-N}{\scriptstyle (\theta^\prime)}\phi^{\bf L}}\right|^\eta \;\geq\;\frac{1}{2^\eta} e^{\eta \kappa N^\beta}\\
	\end{align*}
	for all $\eta\leq1$. By applying again Markov's inequality and proceeding as in the proof of lemma \ref{transprodcont}, we get
	\begin{align*}
		(ii) \;&\leq\; 2^\eta e^{-\eta \kappa N^\beta}\; \left(\sum_{k=1}^{n_N+1} \, \Expect{\norm{T_1{\scriptstyle (\theta )}}}^{k-1} \Expect{\norm{T_1{\scriptstyle (\theta)}-T_1{\scriptstyle (\theta^\prime)}}} \Expect{\norm{T_1{\scriptstyle (\theta^\prime)}}}^{n_N+1-k} \right)^\eta\\
		\;&\leq\;2^\eta e^{-\eta \kappa N^\beta}\; \max\left(\,\Expect{\norm{T_1{\scriptstyle (\theta)}}^\eta}^{n_N+1},\Expect{\norm{T_1{\scriptstyle (\theta^\prime)}}^\eta}^{n_N+1}\,\right) (n_N+1)  (3\varepsilon)^\eta \, ,
	\end{align*}
where we used $(|a|+|b|)^\eta \leq |a|^\eta+|b|^\eta$ and \eqref{eqdifftransdifftheta} together with $\Expect{\frac{1}{|a_1|^\eta}} \leq  \Expect{\norm{T_1{\scriptstyle (\theta^\prime)}}^\eta}$. Now, by assumption, there exists $\eta>0$ such that $\Expect{\norm{T_1{\scriptstyle (\theta^\prime)}}^\eta}$ is finite. Hence, by choosing $\tau$ small enough and applying the same procedure to $(iii)$, we can find some $\alpha_2 > 0$ such that
	\begin{align*}
		(ii) \;+\; (iii) \; \leq\; 2 \,e^{-\alpha_2 N^\beta}\;,
	\end{align*}
	and the result follows.
\end{proof}

We note shortly that the Wegner estimate can be extended to the following ``uniform'' version with respect to the quasi-energy $\theta \in \T$. Instead of bounding the probability to find an eigenvalue in some small interval, we would like to bound the probability that two independent disordered walk operators have one equal eigenvalue. Thus, in view of lemma \ref{Decomposition}, consider the quantity
\begin{align*}
	\mathbbm{P}_{\omega,\tilde{\omega}}\left(\; \exists\, \theta\, \in \,\T \,:\, |\Scp{\phi^{\bf R}}{T_{N}{\scriptstyle (\theta)}\cdot\ldots\cdot T_{-N}{\scriptstyle (\theta)}\phi^{\bf L}}| < e^{- \sigma (2N+1)^\beta } \right. \\
	\left. \text{ and } |\Scp{\tilde{\phi}^{\bf R}}{\tilde{T}_{N}{\scriptstyle (\theta)}\cdot\ldots\cdot \tilde{T}_{-N}{\scriptstyle (\theta)}\tilde{\phi}^{\bf L}}| < e^{- \sigma (2N+1)^\beta } \right) \;,
\end{align*}
where the transfer matrices $T_i$ and $\tilde{T}_i$ and the boundary conditions $\phi^{\bf L,R}$ and $\tilde{\phi}^{\bf L,R}$ are chosen independently from another, and hence can be thought of as corresponding to two independent walk operators $W_{\omega_1}$ and $W_{\omega_2}$, both acting on a $4(N+1)$-dimensional Hilbert space. But if the event 
\begin{align*}
	|\Scp{\phi^{\bf R}}{T_{N}{\scriptstyle (\theta)}\cdot\ldots\cdot T_{-N}{\scriptstyle (\theta)}\phi^{\bf L}}| < e^{- \sigma (2N+1)^\beta }
\end{align*}
occurs for some $\theta \in \T$, then we know that there exists an eigenvalue of $W_{\omega_1}$ in the interval of length $e^{- \sigma (2N+1)^\beta }$ around $\theta$. Using the fact that $W_{\omega_1}$ has $4(N+1)$ eigenvalues, we can apply the union bound and end up with 
\begin{align}\label{equniformwegner}
	\begin{split}
	&\mathbbm{P}_{\omega,\tilde{\omega}}\left(\; \exists\, \theta\, \in \,\T \,:\, |\Scp{\phi^{\bf R}}{T_{N}{\scriptstyle (\theta)}\cdot\ldots\cdot T_{-N}{\scriptstyle (\theta)}\phi^{\bf L}}| < e^{- \sigma (2N+1)^\beta } \right. \\
	&\left. \text{ and } |\Scp{\tilde{\phi}^{\bf R}}{\tilde{T}_{N}{\scriptstyle (\theta)}\cdot\ldots\cdot \tilde{T}_{-N}{\scriptstyle (\theta)}\tilde{\phi}^{\bf L}}| < e^{- \sigma (2N+1)^\beta } \right) \\
	&\leq\; 4 (N+1) \, \Prob{\; \exists\, \theta\, \in \, I_{e^{- \sigma (2N+1)^\beta}}(\theta_0) \,:\, |\Scp{\tilde{\phi}^{\bf R}}{T_{N}{\scriptstyle (\theta)}\cdot\ldots\cdot T_{-N}{\scriptstyle (\theta)}\tilde{\phi}^{\bf L}}| < e^{- \sigma (2N+1)^\beta } } \\
	&\leq\; 4(N+1) \, e^{-\alpha N^\beta} \\
	&\leq\; 3 (2N+1) \, e^{-\alpha^\prime (2N+1)^\beta} 
	\end{split}
\end{align}
for some $\alpha^\prime > 0$.

\subsection{Multiscale Analysis}

We now carry out a multiscale-analysis similar to the ones in \cite{Germinet:2001p4369,pre05533269} in order to extend the initial length scale estimate to products of arbitrary length. The Wegner estimate is used to control the bad instances.

\begin{lemma}[MSA Induction step]
	\label{msalemma}
		Let $X: \T \times \Omega \to \R_+$ be family of positive, real valued random variables, indexed by an element of the unit circle. Assume that we can verify the following two estimates.
	\begin{itemize}
		\item[\bf 1] There exist for any $n_0 \in \N$ some positive numbers $\gamma_0$, $\sigma$ and a subset $I(n_0)$ of the unit circle such that
			\begin{align*}
				\Prob{ \; \text{for all } \theta \in I(n_0) \;:\; X(\theta,\omega) > e^{\gamma_0 n_0} } \; > \; 1 - e^{-\sigma n_0} \;.
			\end{align*}
		\item[\bf 2] Given a natural number $m$ and the product
			\begin{align*}
				X^{m}(\theta,\underline{\omega})=\prod_{i=1}^{m} X(\theta,\omega_i)
  			\end{align*}
			where the factors are understood as being independent and identically distributed (i.i.d), we have for all $0<\delta<1$ the estimate
			\begin{align}
				\label{eqreqwegner}
				\Prob{ \;\exists \; \theta \in I(n_0) \;:\; X^{m}(\theta,\underline{\omega}) < e^{-\frac{\gamma_0}{2}m^\delta}  \text{ and } X^{m}(\theta,\underline{\tilde{\omega}}) < e^{-\frac{\gamma_0}{2}m^\delta} } < 3 m e^{-\sigma m^\delta}\;,
			\end{align}
			for two i.i.d. products $X^{m}(\theta,\underline{\omega})$ and $X^{m}(\theta,\underline{\tilde{\omega}})$.
	\end{itemize}
	Then we have for all natural numbers $k$, $1 < \alpha < 2$, $\xi+1<\alpha$, length scales $n_k=n_0^{\alpha^k}$
	\begin{align}
		\label{eqmsa}
		\Prob{\;\exists \;\theta \in I(n_0) \;:\; X^{n_{k}}(\theta,\underline{\omega})< e^{\gamma_{k} n_{k}} \text{ and }  X^{n_{k}}(\theta,\underline{\tilde{\omega}})< e^{\gamma_{k} n_{k}}}\leq e^{-\sigma n_{k}^{\xi}}
	\end{align}
	 holds with $\gamma_0\geq\gamma_{k}\geq \frac{\gamma_0}{2}$.
\end{lemma}

\begin{proof}
	We prove the lemma by induction. So assuming that \eqref{eqmsa} is fulfilled for some $k \in \N$, we use assumption {\bf 2} to prove that is also holds for $k+1$, with a slightly smaller $\gamma_{k+1}$, which can nevertheless be lower bounded by $\gamma_0 / 2$. We find that this requires that, depending on our choice of $\alpha$ and $\xi$, the initial scale $n_0$ has to be large enough. The result then follows by applying the first induction step to a suitable initial scale $n_0$, which can be constructed using assumption {\bf 1}.

  We divide the proof of the induction step into two parts, the deterministic and the probabilistic estimate. In the deterministic part we make assumptions about the value of the product of the random variables $X^{n_k}$ that helps us to estimate a new growth rate $\gamma_{k+1}$.

  In the probabilistic part we bound the probability that these assumptions are violated or in other words we bound the probability that the new growing rate is smaller then $\gamma_{k+1}$ as estimated in the deterministic part.

 So let us begin by stating the two assumptions we use in our deterministic estimate.
  \begin{itemize}
    \item[\bf A1] Only $n_k^\beta$ of the $n_k^{\alpha-1}$ factors in the product of the $X^{n_k}(\theta,\omega)$ grow with a rate smaller then $\gamma_k$. Where $0<\beta<1$ is chosen such that $\alpha > 1+\beta$
    \item[\bf A2] None of the $n^\beta$ bad instances from assumption {\bf A1} is smaller than $e^{-\frac{\gamma_0}{2} n_k^\delta}$.
  \end{itemize}
 Now let us find a lower bound $\gamma_{k+1}$ on the growth rate of $X^{n_{k+1}}$. Taking the logarithm we get
 \begin{align}\label{eq_msa_partition}
	 \ln (X^{n_{k+1}}(\theta,\underline{\omega}))= \sum_{i=1}^{n_k^{\alpha-1}} \ln(X^{n_k}(\theta,\omega_i)) \geq (n_k^{\alpha-1}-n_k^\beta)\gamma_k n_k - \frac{\gamma_0}{2} n_k^\delta n_k^\beta \;
 \end{align}
 where we used the assumptions {\bf A1} and {\bf A2} in the second step of the inequality. In order to lower bound the new growth rate we define $\gamma_{k+1}$ by the right hand side of \eqref{eq_msa_partition} divided by $n_k^{\alpha}$. This leads to the recursion relation
 \begin{align}\label{eq_msa_new_grow_gamma}
   \gamma_{k+1} = (1-n_k^{1+\beta-\alpha})\gamma_k - \frac{\gamma_0}{2} n_k^{\delta+\beta-\alpha} \geq \gamma_k -\frac{3}{2}  \gamma_0 n_k^{1+\beta-\alpha} \;,
 \end{align}
 because $\delta <1$ and $\gamma_k\leq\gamma_0$ by assumption. It remains to show that still $\gamma_0\geq\gamma_{k+1}\geq\frac{\gamma_0}{2}$ holds. Let us start with the lower bound. Considering \eqref{eq_msa_new_grow_gamma} we find
 \begin{align}\label{eq_msa_gammak}
   \gamma_{k+1}-\gamma_0 = \sum_{l=0}^k (\gamma_{l+1}-\gamma_l) \geq - \frac{3}{2}\gamma_0\sum_{l=0}^k n_l^{1+\beta-\alpha}\;.
 \end{align}
 Using $n_k=n_0^{\alpha^k}$, and setting $q=n_0^{1+\beta-\alpha}$ we have to show that
 \begin{align}
   q+\sum_{l=1}^{k} q^{\alpha^l} \leq \frac{1}{3}\,
 \end{align}
to guarantee $\gamma_{k+1}\geq \frac{\gamma_0}{2}$. Since the summands are positive and $q<1$ holds, because $1+\beta<\alpha$, we can bound the finite sum by the infinite series and use $q^{\alpha^k}\leq q^{k(\alpha-1)+1}$. Inserting this into \eqref{eq_msa_gammak}, evaluating the geometric series and exploiting $q\leq q^{\alpha-1}$ we get the condition
 \begin{align}
   n_0^{(1+\beta-\alpha)(\alpha-1)}\leq \frac{1}{4}
 \end{align}
 which can be achieved for $n_0$ large enough since $1+\beta<\alpha$. By assumption this proves $\gamma_{k+1}\geq\frac{\gamma_0}{2}$. So we are left to prove that $\gamma_k$ is a decaying sequence. From
  \begin{align}
    \gamma_{k+1}-\gamma_{k} = -(n_k^{1+\beta-\alpha}\gamma_k+\frac{\gamma_0}{2}n_k^{\delta+\beta-\alpha})< 0
  \end{align}
  we conclude $\gamma_k>\gamma_{k+1}$ and therefore $\gamma_0\geq\gamma_{k+1}$ by assumption. This
  finishes the deterministic estimate.

 Let us now proceed with the probabilistic part, \ie\ we have to bound the probability that either assumption {\bf A1} or assumption {\bf A2} is violated for both products $X^{n_{k+1}}(\theta,\omega)$ and $X^{n_{k+1}}(\theta,\tilde{\omega})$. Let us denote the event that assumption {\bf Ai} is violated for $X^{n_{k+1}}(\theta,\omega)$ by $\complement \text{Ai}[X^{n_{k+1}}(\theta,\omega)]$. We can estimate
 \begin{align*}
	 &\Prob{\;\exists \theta \in I(n_0)\;:\; \text{Assumptions {\bf A1} or {\bf A2} are violated for } X^{n_{k+1}}(\theta,\omega) \text{ and } X^{n_{k+1}}(\theta,\tilde{\omega}) }
	 \\ &=\; 2 \, \Prob{ \;\exists\; \theta \in I(n_0)\;: \;\complement \text{A1}[X^{n_{k+1}}(\theta,\omega)] \text{ and } \complement \text{A2}[X^{n_{k+1}}(\theta,\tilde{\omega})] }
	 \\ &+\; \Prob{ \;\exists\; \theta \in I(n_0)\;:  \;\complement \text{A1}[X^{n_{k+1}}(\theta,\omega)] \text{ and } \complement \text{A1}[X^{n_{k+1}}(\theta,\tilde{\omega})] }
	 \\ &+\; \Prob{ \;\exists\; \theta \in I(n_0)\;:  \;\complement \text{A2}[X^{n_{k+1}}(\theta,\omega)] \text{ and } \complement \text{A2}[X^{n_{k+1}}(\theta,\tilde{\omega})] }
	 \\ &\leq\; 3 \, \Prob{ \;\exists\; \theta \in I(n_0)\;:  \;\complement \text{A1}[X^{n_{k+1}}(\theta,\omega)] }
	 \\ &+\; \Prob{ \;\exists\; \theta \in I(n_0)\;:  \;\complement \text{A2}[X^{n_{k+1}}(\theta,\omega)] \text{ and } \complement \text{A2}[X^{n_{k+1}}(\theta,\tilde{\omega})] }
 \end{align*}
 Now assume {\bf A1} is violated, \ie\ more than $n_k^\beta$ factors of the product $X^n(\theta,\omega)$ grow with a rate smaller then $\gamma_k$, or put differently, more than $\frac{n_k^\beta+1}{2}$ pairs of products  grow with a rate smaller then $\gamma_k$. Hence we can use the induction hypotheses \eqref{eqreqwegner} to bound
 \begin{align*}
	 &3 \, \Prob{ \;\exists\; \theta \in I(n_0)\;:  \;\complement \text{A1}[X^{n_{k+1}}(\theta,\omega)] }
	 \\ &\leq\;  3 \, \Prob{ \;\exists\; \theta \in I(n_0)\;:  \;\complement \text{A1}[X^{n_{k}}(\theta,\omega)] \text{ and } \complement \text{A1}[X^{n_{k}}(\theta,\tilde{\omega})]  }^\frac{n_k^\beta+1}{2} \;\leq\; 3 \, e^{- \sigma n_k^\xi \frac{n_k^\beta + 1}{2}} ;.
 \end{align*}
 From the requirement \eqref{eqreqwegner} we get that
 \begin{align*}
	 \Prob{ \;\exists\; \theta \in I(n_0)\;:  \;\complement \text{A2}[X^{n_{k+1}}(\theta,\omega)] \text{ and } \complement \text{A2}[X^{n_{k+1}}(\theta,\tilde{\omega})] } \;\leq\; 3 \, n_k \, e^{- \sigma n_k^\delta} \,.
 \end{align*}
 Thus we can combine these two estimates in order to get
 \begin{align*}
	&\Prob{\;\exists  \theta \in I(n_0)\;:\;\text{Assumptions {\bf A1} or {\bf A2} are violated for } X^{n_{k+1}}(\theta,\omega) \text{ and } X^{n_{k+1}}(\theta,\tilde{\omega}) } \\
	&\leq\; 3 \, e^{- \sigma n_k^\xi \frac{n_k^\beta + 1}{2}} \;+\; 3 \, n_k \, e^{- \sigma n_k^\delta}  \;\leq\; e^{-\sigma n_k^{\alpha \cdot \xi}}
 \end{align*}
 provided that $\xi + \beta > \alpha \cdot \xi$ and $\delta > \alpha \cdot \xi$ and $n_0$ big enough.

 Together with the requirement $\beta < \alpha - 1$ originating from the deterministic estimate, we see that the induction step can be carried out if $\xi < \beta < \alpha - 1$ and $\xi < \delta / \alpha$. But $\xi+1 < \alpha$ holds by assumption, and thus we can find $\beta$ and $\delta$ accordingly.
\end{proof}

\subsection{Dynamical Localization}\label{ProffOfExponentialDecay}

\begin{lemma}
	\label{expdecayspectralmeasure}
Assuming the conditions of theorem \ref{thmtransfermatrices} are fulfilled, then there exist positive constants $C_1$, $C_2$, $0<\xi<1$ such that we have

\begin{align*}
	\Expect{|\rhof|(\T)} \,\leq\, C_1 e^{-C_2 |x-y|^{\xi}} \,,
\end{align*}
where $\rhof$ is the spectral measure of the walk operator $\W$ corresponding to sites $x$, $y \in \Z$.
\end{lemma}

\begin{remark}
	The right hand side of this inequality clearly constitutes a function 
	\[\f:n \mapsto  C e^{-C_2 n^\xi} \]
	which decays faster than any polynomial in $n$, and together with the arguments given at the beginning of section \ref{secfinrestrictions} the preceding lemma thus also proves theorem \ref{thmtransfermatrices}.
\end{remark}

\begin{proof}
We show that if the conditions of theorem \ref{thmtransfermatrices} are fulfilled, then there exists for every $\theta \in \T$ an open arc $I_\delta$, $\delta > 0$, around $\theta$ and appropriate constants $C_1$, $C_2$, $0<\xi<1$ such that 
\begin{align*}
	\Expect{|\rhof|(I_\delta)} \,\leq\, C_1 e^{-C_2 |x-y|^{\xi}} \,
\end{align*}
for large enough ring sizes $L$. The result then follows by compactness of $\T$.

For the rest of this section, we assume without loss of generality that $x>y$. Using lemma \ref{Decomposition} and the theory of Cauchy transforms (see section \ref{secfinrestrictions}), we can apply the upper bound
\begin{align*} 
	|\rhof|(I_\delta) \,&\leq\, \lim_{\kappa \rightarrow \infty} \pi \kappa \cdot \lambda  \left| \left\{ \, \theta \in I_\delta \,:\,2 \,|\Scp{\Phi_+}{\Trans{x-1}\ldots\Trans{y} \Phi_-}|  \,<\,\frac{1}{\kappa} \, \right\} \right| \;.
\end{align*}
Hence, we have to lower bound the growth rate of $|\Scp{\Phi_+}{\Trans{x-1}\ldots\Trans{y} \Phi_-}|$. This is done by applying the MSA lemma.

We first choose an $1<\alpha<2$, $0<\xi<1$, and an initial length scale $n_0 \in \N$ respecting the requirements of lemma \ref{msalemma} such that there exists $k\in \N$ with the property that the $k$-induction step applied to $n_0$ provides us with a length scale $\ell = n_0^{\alpha^k}$ fulfilling  $2 \ell + 2 \leq |x-y| \leq \ell^\alpha$, where $\ell^\alpha$ is identified with is integer component. Note that this is always possible by suitably choosing $n_0$, $\alpha$ and $\xi$, if $|x-y|$ is big enough (the other cases are of course not important if we look at the decay rate, and are covered by a constant). Then, we can factor out from $|\Scp{\Phi_+}{\Trans{x-1}\ldots\Trans{y} \Phi_-}|$ a box of length $\ell$, either starting at $x$ and stretching out to the right or starting at $y$ and stretching out to the left, \ie\
\begin{align*}
	|\Scp{\Phi_+}{\Trans{x-1}\ldots\Trans{y} \Phi_-}| = |\Scp{\Phi_+}{\Trans{x-1}\ldots\Trans{x-1-\ell} \tilde{\Phi}_-}||\Scp{\tilde{\Phi}_+}{\Trans{x-2-\ell}\ldots\Trans{y} \Phi_-}| \;,
\end{align*}
or
\begin{align*}
	|\Scp{\Phi_+}{\Trans{x-1}\ldots\Trans{y} \Phi_-}| = |\Scp{\Phi_+}{\Trans{x-1}\ldots\Trans{y+\ell+1} \tilde{\Phi}_-}||\Scp{\tilde{\Phi}_+}{\Trans{y+\ell}\ldots\Trans{y} \Phi_-}| \;.
\end{align*}
Both of these boxes can again be factored into i.i.d. boxes of length $n_0$, \ie\
\begin{align*}
	|\Scp{\Phi_+}{\Trans{x-1}\ldots\Trans{x-1-\ell} \tilde{\Phi}_-}| = \prod_{i=1}^\ell |\Scp{\tilde{\Phi}^i_+}{\Trans{n_0}\ldots\Trans{1} \tilde{\Phi}^i_-}|
\end{align*}
with some collection of normalized vectors $\{\tilde{\Phi}^i_\pm\}$ in the plane $\Plane$, and the same can be achieved for $|\Scp{\Phi_+}{\Trans{x-1}\ldots\Trans{y+\ell+1} \tilde{\Phi}_-}|$. Now by assumption, the conditions of corollary \ref{TransNormGrow} hold for the group of transfer matrices associated to any $\theta \in \T$. Hence we may apply proposition \ref{initiallengthscale} to construct an open arc $I_\delta$, $\delta > 0$ and get an initial length scale estimate valid for the length $n_0$ with constants $\gamma_0$ (for the scalar product) and $\sigma_I$ (for the probability). Thus, assumption {\bf 1} of lemma \ref{msalemma} is verified for the random variables
\begin{align*}
	|\Scp{\Phi^i_+}{\Trans{n_0}\ldots\Trans{1} \Phi^i_-}|\;.
\end{align*}
The verification of assumption {\bf 2} of lemma \ref{msalemma} is done using the Wegner estimate and was carried out in the discussion preceding equation \eqref{equniformwegner}, with decay constant $\sigma_W$ (bound on the probability). Note that we can choose $\sigma = \min(\sigma_W,\sigma_I)$, and both the Wegner estimate and the initial scale estimate are still valid for $\sigma$ as the decay constant in the bound of the probability. Hence, both assumptions needed for carrying out the MSA are verified. Next, define the event $A(I_\delta,\ell) \subset \Omega$ as the set of those $\omega$ for which either the box 
\[|\Scp{\Phi_+}{\Trans{x-1}\ldots\Trans{x-1-\ell} \tilde{\Phi}_-}|\]
or the box 
\[|\Scp{\tilde{\Phi_+}}{\Trans{y+\ell}\ldots\Trans{y} \Phi_-}|\]
is bigger or equal to $e^{\gamma \ell}$, $\gamma = \frac{1}{2} \gamma_0$ for all $\theta \in I_\delta$. Here, $\gamma_0$ is the deterministic decay constant we get from the initial scale estimate. As assumptions {\bf 1} and {\bf 2} of lemma \ref{msalemma} are verified, we have that
\begin{align}
	\label{secdynloceq1}
	\Prob{A(I_\delta,\ell)} \geq 1 - e^{-\sigma \ell^\xi} \;.
\end{align}
Assuming that $A(I_\delta,\ell)$ occurs, and reshuffling, as well as renaming indices, we end up with 
\begin{align}
	\label{secdynloceq2}
	|\rhof|(I_\delta) \,&\leq\,  \; \lim_{\kappa \rightarrow \infty} \pi \kappa \cdot \lambda  \left| \left\{ \, \theta \in I_\delta \,:\,e^{\gamma \ell}\,2\,|\Scp{\hat{\Phi}_+}{\Trans{|x-y|-1-\ell}\ldots\Trans{1} \hat{\Phi}_-}|  \,<\,\frac{1}{\kappa} \, \right\} \right| \\
	&= e^{-\gamma \ell} \; \lim_{\kappa \rightarrow \infty} \pi \kappa \cdot \lambda  \left| \left\{ \, \theta \in I_\delta \,:\,2\,|\Scp{\hat{\Phi}_+}{\Trans{|x-y|-1-\ell}\ldots\Trans{1} \hat{\Phi}_-}|  \,<\,\frac{1}{\kappa} \, \right\} \right| 
\end{align}
with $\hat{\Phi}_\pm$ being vectors in the plane $\Plane$. But $\left(2\,\Scp{\hat{\Phi}_+}{\Trans{|x-y|-1-\ell}\ldots\Trans{1} \hat{\Phi}_-}\right)^{-1}$ is, by lemma \ref{Decomposition} and remark \ref{Remdecomposition}, the resolvent of some unitary walk operator restricted to a box of length $|x-y| - \ell$. It then follows from equation \eqref{eqspecmeascauchy} that
\begin{align}
	\label{secdynloceq3}
	\lim_{\kappa \rightarrow \infty} \pi \kappa \cdot \lambda  \left| \left\{ \, \theta \in I_\delta \,:\,2\,|\Scp{\hat{\Phi}_+}{\Trans{|x-y|-1-\ell}\ldots\Trans{1} \hat{\Phi}_-}|  \,<\,\frac{1}{\kappa} \, \right\} \right| \;\leq\; |\rho_{\omega,N}^{|x-y|-1-\ell,1}|(I_\delta)\;\leq\; 1 \;.
\end{align}
Now, if we combine equations \eqref{secdynloceq1}, \eqref{secdynloceq2} and \eqref{secdynloceq3} we can estimate
\begin{align*}
	\Expect{|\rhof|(I_\delta)} \,&=\, \int_{A(I_\delta,\ell)} \Prob{d\omega} \, |\rhof|(I_\delta) \, + \,  \int_{\complement A(I_\delta,\ell)} \Prob{d\omega} \, |\rhof|(I_\delta) \\
	&\leq\, e^{-\gamma \ell} + e^{-\sigma \ell^\xi}
\end{align*}
 under the assumption on $|x-y|$. But since we have chosen $\ell$ and $\alpha$ such that $|x-y| \leq \ell^\alpha$, it follows that
\begin{align}
	\label{secdyneqfertig}
	\begin{split}
	\Expect{|\rhof|(I_\delta)} \,&\leq\,C \left(e^{-\gamma (\ell^\alpha)^\frac{1}{\alpha}} + e^{-\sigma (\ell^\alpha)^\frac{\xi}{\alpha}} \right) \\ 
	&\leq\,  C \left( e^{-\gamma |x-y|^\frac{1}{\alpha}} + e^{-\sigma |x-y|^\frac{\xi}{\alpha}} \right) \\
	&\leq\,  C \; e^{- \min \{\gamma,\sigma\}|x-y|^\frac{\xi}{\alpha}}\,,
	\end{split}
\end{align}
where the constant $C>0$ is chosen to cope with the cases where $|x-y|$ is not ``big enough''. Given the requirements of theorem \ref{thmtransfermatrices}, the derivations in the preceding sections assure us that we can always find $\gamma$, $\sigma > 0$ and $\xi$, $\alpha$ such that equation \eqref{secdyneqfertig} is true with $0<\frac{\xi}{\alpha}<1$, and the proof is complete.
\end{proof}

What remains to be proven is the decay of the eigenfunctions. We first prove the statement for the eigenfunctions of the finite restrictions $\Wf$. This can be derived from our bounds on the time evolution of initially localized particles. We use a variant of Wiener's theorem for measures on the unit circle, proven in the appendix, see theorem \ref{wienermeasureunitcircle}. Using in addition that the matrix elements of some eigenprojector $P^\omega_{\{\theta\}}$ of $\Wf$ at the eigenvalue $\theta$ are given by the value of the spectral measure at the point $\xi$,
\begin{align*}
	\Scp{\LocS[y]{e_i}}{P^\omega_{\{\theta\}}\LocS{e_j}} \, = \, \rhof(\{\theta\}) \,.
\end{align*}
we arrive at  
\begin{align*}
	\sum_{\theta \in \sigma(\Wf)} |\Scp{\LocS[y]{e_i}}{P^\omega_{\{\theta\}}\LocS{e_j}}|^2 \, &= \, \lim_{T \to \infty} \frac{1}{T+1} \sum_{t=0}^T \left|\int_\T \rhof(\{\theta\}) \theta^t \right|^2 \\
	 &\leq\, \lim_{T \to \infty} \frac{1}{T+1} \sum_{t=0}^T |\Scp{\LocS[y]{e_i}}{\Wf^t \LocS{e_j}}|^2 \\
	 &\leq\,  \sup_{t \in \N} |\Scp{\LocS[y]{e_i}}{\Wf^t \LocS{e_j}}|
\end{align*}
where we used the Cauchy-Schwarz inequality in the last step and $\sigma(\Wf)$ denotes the spectrum of $\Wf$. Now, we have by lemma \ref{expdecayspectralmeasure} that we can bound the decay of the expectation over all configurations of the spectral measure by $\f$, which implies
\begin{align*}
	\Expect{ \sum_{x,y \in \Z} \sum_{\theta \in \sigma(\Wf)} \f(|x-y|)^{-1} \,|\Scp{\LocS[y]{e_i}}{P^\omega_{\{\theta\}}\LocS{e_j}}|^2} \, &<\, \infty \;,
\end{align*}
since $\f$ decays faster than any polynomial. This in turn implies for every $\theta \in \sigma(\Wf)$ and any $\omega \in \Omega$ the existence of a random variable with finite expectation $A_\theta(\omega)$, such that
\begin{align*}
	 |\Scp{\LocS[y]{e_i}}{P^\omega_{\{\theta\}}\LocS{e_j}}|^2\,\leq\, A_\theta(\omega) \f(|x-y|) \,,
\end{align*}
ensuring the decay of eigenfunctions of $\Wf$ with probability one. To apply the argument for the eigenfunctions of $\W$, defined on the infinite lattice, we have to show that the spectral measures $\rhof$ converge to $\rho^{x,y}_\omega$ of $\W$, if $N\to \infty$. But this follows from equation \eqref{eqfinitelimit} and the fact that the polynomials are dense in the space of continuous functions on the unit circle.

\appendix
\section{Products of random matrices}
\label{F\"{u}rstenberg}
In this section we gather some parts of random matrix theory. Our aim is to analyze the behavior of random products of matrices, in particular we are interested in the growth of quantities like
\[
\Norm{g_n\cdot \ldots \cdot g_1\cdot v}\quad \mathrm{and} \quad |\Scp{w}{g_n\cdot \ldots \cdot g_1\cdot v}|
\]
where $v$ and $w$ are normalized vectors and the $g_i$ are complex valued matrices drawn independent and identically distributed according to some measure $\mu$ on $\mathrm{GL}(\K,k)$ where $\K=\R$ or $\C$. In fact, the measure $\mu$ depends on a parameter $z\in\T$, which we take into account by writing $\mu_z$, hence, the matrices $g_i$ also depend on $z$ but to simplify notation we omit this dependence.

If $\K=\R$ all results derived in this appendix are well-known, see \eg\ \cite{Bougerol,Carmona}, however, for $\K=\C$ the literature is less exhaustive. So, in this appendix we adapt these results to the case $\K=\C$. Many statements can be translated directly from $\R$ to $\C$ by observing $\C\cong \R^2$. This approach is similar to the analysis carried out in \cite{Bourget2003}. More precisely, we identify $\C$ with $\R^2$ via the isomorphism
\begin{align}\label{eq_CRIso}
  \Lambda : \C \rightarrow \R^2\,,\,\Lambda(x+iy)=\genfrac(){0pt}{}{ x}{ y}\,,\quad x,y\in \R\,.
\end{align}

This mapping naturally extends to $\C^k$ and induces a homomorphism $\Psi$ mapping linear operators on $\C^k$ to linear operators on $\R^{2k}$ via the relation
\[
\Psi :\mathrm{Mat}(\C,k)\rightarrow \mathrm{Mat}(\R,2k)\, ,\quad\Psi(M)=\Lambda \circ M \circ \Lambda^{-1}\,.
\]
It is easy to see that both $\Lambda$ and $\Psi$ are $\R$-linear and since $\Lambda$ is an isomorphism $\Psi$ is injective, however, it is not surjective. Moreover, $\Psi$ is a homomorphism with respect to matrix multiplication, hence a subgroup $G$ of $\mathrm{GL}(\C,k)$ is mapped isomorphically to a subgroup $\Psi(G)$ of $\mathrm{GL}(\R,2k)$. The mappings $\Lambda$ and $\Psi$ have the following useful properties:
\begin{proposition}
The mappings $\Lambda$ and $\Psi$ preserve the euclidean norm, \ie\
\[{\Norm{M\cdot v}}_{\C}={\Norm{\Psi (M)\cdot \Lambda (v)}}_{\R} \, ,\quad \forall M\in \mathrm{Mat}(\C, k)\, ,v\in\C^k\, ,
\]
where the norm is given by the euclidean norm in $\C^k$ respectively $\R^{2k}$. Moreover, we have that the map $\Psi$ preserves the modulus of the determinant, \ie\
\[
|\Det{M}|=|\Det{\Psi (M)}|\, ,\quad \forall M\in \mathrm{Mat}(\C,k)\, .
\]
\end{proposition}
\begin{proof}
The first statement of the lemma follows easily from the fact that
\[
{\Norm{\Psi (M)\cdot \Lambda (v)}}_{\R}={\Norm{\Lambda (M\cdot v)}}_{\R}
\]
for all $M\in \mathrm{Mat}(\C,k)$ and $v\in \C^k$.

Since the modulus of the determinant of a matrix $M$ is given by the product of its singular values, it it sufficient to prove that $\Psi$ preserves the singular values of $M$. This follows from the equations
\[
\Psi (M\cdot M^*)=\Psi (M) \cdot \Psi (M)^*\,,\quad \forall M\in \mathrm{Mat}(\C ,k)\,
\]
and
\[
\Lambda (M\cdot v)=\Psi (M)\cdot \Lambda (v)\, ,\quad \forall M\in \mathrm{Mat}(\C,k)\,,\, v\in \C^k
\]
together with the invertibility of $\Lambda$.
\end{proof}
The transfer matrices $T$ we consider satisfy $|\Det{T }|=1$, therefore we are interested in $\Psi (T)$ where $T$ is an element of the group of matrices with determinant of modulus one, which is denoted by
\[
\mathrm{SL}_\mathbbm{T}(k) =\mathbbm{T}\times \mathrm{SL}(\C,k)=\{M\in \mathrm{Mat}(\C,k)\,:\,|\Det{M}|=1\}\,.
\]
\begin{remark}
\label{Unimodular}
Since the map $\Psi$ is continuous we have that $\Det{.}\circ \Psi$ is a continuous map from $\mathrm{SL}_\mathbbm{T}(k)$ to $\{-1,1\}$. The group $\mathrm{SL}(\C,k)$ as well as $\mathrm{SL}_\mathbbm{T}(k)$ is simply connected. Therefore, the function $\Det{.}\circ \Psi$ is constant on $\mathrm{SL}_\mathbbm{T}(k)$ and since $\Det{\Psi(\Id_{\C^k})}=\Det{\Id_{\R^{2k}}}=1$ it is equal to $1$.
\end{remark}

Our aim is to prove that, under certain circumstances, the Lyapunov exponent, as defined subsequently, is positive and H\"{o}lder-continuous on $\T$. To this end it is helpful to consider the group generated by the support of the measure $\mu_z$, which is defined in the following way:
\begin{definition}
Let $\{\mu_z\}_{z\in\T}$ be a family of measures on $\mathrm{GL}(\K,k)$, the general linear group on the $\K$-vector space of dimension $k$.
\begin{itemize}
\item[i)] The Lyapunov exponent associated with $\mu_z$ is the limit
\[
\gamma (z) :=\lim\limits_{n\rightarrow \infty}\frac{1}{n}\Expect{\Norm{g_n\cdot\ldots\cdot g_1}}\,,
\]
where the expectation value is with respect to the $n$-fold convolution $\mu_z^n$ of the measure $\mu_z$.
\item[ii)] The support of the measure $\mu_z$ is defined as the set
\[
\mathrm{supp}(\mu_z):=\{M\in \mathrm{GL}(\K,k)\,:\,\mu_z(\mathcal{B}_\varepsilon(M))>0\,\forall \varepsilon>0\}\,,
\]
with $\mathcal{B}_\varepsilon(M)$ denoting the sphere of operator norm radius $\varepsilon$ around $M\in \mathrm{GL}(\K,k)$.
\item[iii)] The smallest closed subgroup of $\mathrm{GL}(\K,k)$ containing the elements of $\mathrm{supp}(\mu_z)$ is denoted by $\langle\mu_z\rangle$.
\end{itemize}
\end{definition}
To such a measure $\mu_z$ we may associate a map $R_z$ acting on functions $f$ from the projective space $P\K^k \cong \{v\in \K^k\,:\, \Norm{v}=1\}$ to the real numbers in the following way
\begin{equation}
\label{MuOp}
(R_z f)(\bar x) := \int f(g\bar x)\mu_z (dg)\quad \text{for bounded measurable} \quad f: P\K^k\rightarrow \R\,.
\end{equation}
From now on we adopt the notation $\bar x $ for elements in $P\K^k$, whereas $x$ denotes the elements of $\K^k$. In the following, we restrict the definition of $R_z$ to functions $f$ from a set $\La_\alpha$ defined with the help of a distance $\delta$ on $P\K^k$ and $\alpha>0$ in the following sense
\[
f\in\La_\alpha \quad \Longleftrightarrow \quad m_\alpha (f):=\sup\limits_{\bar x\neq\bar y}\frac{|f(\bar x)-f(\bar y)|}{\delta (\bar x,\bar y)^\alpha}<\infty\,.
\]
If we define a norm on $\La_\alpha$ via ${\Norm{f}}_\alpha = {\Norm{f}}_\infty + m_\alpha (f)$ it becomes a Banach space (note that $f$ is assumed to be bounded), and denoting the set of bounded operators on $\La_\alpha$ by $\B(\La_\alpha)$ we have $R_z \in \B(\La_\alpha)$.

This leads us to the notion of an invariant measure $\nu_z$ corresponding to $\mu_z$.
\begin{definition}
\label{ImportantDef}
Let $\{\mu_z\}_{z\in\T}$ be a family of measures on $\mathrm{GL}(\K,k)$. An invariant measure $\nu_z$ for $\mu_z$ is a measure fulfilling
\[
\int f(g \bar x)\mu_z (dg)\nu_z (d\bar x)=\int f(\bar x)\nu_z (d \bar x)
\]
for any bounded and measurable function $f:P\K^k\rightarrow \R$. The operator $N_z\in \B(\La_\alpha)$ defined via
\[
(N_z (f)) (\bar x):=\int f(\bar y)\nu_z (d\bar y)
\]
is called the projection onto constant functions.
\end{definition}
\begin{remark}\label{Remark:Invariant}
According to Prokhorov's Theorem the space of probability measures $\mathcal{M}(P\R^k)$ constitutes a sequentially compact set, i.e. each sequence contains a weakly-converging subsequence. This fact can be used to give an expression for an invariant measure for a measure $\mu$. Let $\delta$ be an arbitrary probability measure on $\R^k$, then
\[
\nu_n=\frac{1}{n}\sum_{k=0}^{n-1} \mu^k\delta
\]
defines a sequence of probability measures on $\R^k$ and it is easy to check that the limit of any weakly-convergent subsequence is an invariant measure for $\mu$.
\end{remark}
While invariant measures $\nu_z$ are hard to construct explicitly, their existence and certain properties can be established by analyzing $\langle \mu_z \rangle$. In fact, many properties of $\gamma (z)$ can be established through analysis of the corresponding invariant measure $\nu_z$. We start with the connection between the real invariant measure and its complex analogue.
\begin{remark}
Let $\mu$ be a measure on $\mathrm{GL}(\C,k)$ and denote the corresponding measure on $\mathrm{GL}(\R,2k)$ obtained via the isomorphism $\Lambda$ by $\mu_\R$. An invariant measure $\nu_\R$ for $\mu_\R$ induces an invariant measure $\nu$ for $\mu$ via
\[
\nu(X)=\nu_\R (\bigcup\limits_{\theta\in\T}\Lambda(\theta X))\, ,\quad \forall \,X\subset P\C^k.
\]
Hence, if there is an invariant measure for the real embedding of $\mu$ then there is an invariant measure for $\mu$ itself.
\end{remark}
In particular, for $\K=\R$ there exists a unique continuous invariant measure $\nu_z$ if $\langle \mu_z\rangle$ is strongly irreducible and non-compact \cite{Bougerol,Carmona}.
\begin{definition}
Let $G$ be a subgroup of $\mathrm{GL}(\mathbbm{K},k)$ with $k\in\N$. We call $G$ strongly irreducible over $\mathbbm{K}$ if there is no finite sequence of nontrivial subspaces $S_1,\ldots,S_n\subset \mathbbm{K}^k$ such that $G\,(S_1\cup \ldots\cup S_n) \subset S_1\cup \ldots\cup S_n $.
\end{definition}
\begin{remark}
A subgroup $G\subset \mathrm{GL}(\mathbbm{K},k)$ is strongly irreducible over $\K$ if and only if all subgroups of $G$ with finite index are irreducible over $\mathbbm{K}$.
\end{remark}
The following two propositions translate non-compactness and strong irreducibility from $\R$ to $\C$.
\begin{proposition}\label{prop_Psi_comp}
A subgroup $G$ of $\mathrm{GL}(\C,k)$ is compact if and only if $\Psi(G)$ is compact.
\end{proposition}
\begin{proof}
The map $\Psi$ is continuous, hence $\Psi(G)$ is compact if $G$ is compact. On the other hand, we can invert $\Psi$ continuously on the set $\Psi(\mathrm{GL}(\C,k))$. As a consequence $G$ is compact if $\Psi(G)$ is compact.
\end{proof}
We need to identify circumstances under which the subgroup $\Psi(G)$ is strongly irreducible. For $G$ this means that all subgroups of finite index are irreducible over $\R$, \ie\ the only $\R$-linear invariant subspaces of $\C^k$ are trivial. The following proposition simplifies the analysis in the sense, that it is sufficient to check strong irreducibility for $G$ over $\C$ in certain cases. We prove later (see the proof of lemma \ref{ComplexF\"{u}rstenbergLemma}) that these are exactly the cases of interest to us.
\begin{proposition}
\label{PhasesIncluded}
Let $G$ be a subgroup of $\mathrm{GL}(\C,k)$ and denote the direct product of $G$ and $\T$ by $G_\T=\T\times G$. The group $\Psi(G_\T)$ is strongly irreducible over $\R$ if and only if $G_\T$ is strongly irreducible over $\C$.
\end{proposition}
\begin{proof}
Clearly, if $\Psi (G_\T)$ is strongly irreducible over $\R$, then $G_\T$ is strongly irreducible over $\C$. Now, suppose there is a finite sequence of nontrivial subspaces $S_1,\ldots,S_n \subset \R^{2k}$ such that $\Psi (G_\T)\,S_1\cup\ldots \cup S_n \subset S_1\cup\ldots\cup S_n$. The sets $P_i=\Lambda^{-1}(S_i)\subset \C^{k}$ satisfy $G_\T\,P_1\cup \ldots \cup P_n \subset P_1\cup\ldots \cup P_n$ and each $P_i$ is an $\R$-linear subspace of $\C^k$. But since $G_\T$ contains all operators $\ex{i\,\phi}\mathbbm{1},\, \phi \in \R$ each $P_i$ is also a $\C$-linear subspace of $\C^{k}$, hence $G_\T$ is not strongly irreducible over $\C$.
\end{proof}
The next lemma summarizes the results we wish to translate from the case $\K=\R$ to $\K=\C$. For a proof of the first part we refer to \cite{Furstenberg1963} or \cite{Carmona} proposition IV.4.6, the second part can also be found in \cite{Carmona} propositions I.V.4.4 and IV.4.5, and the third part is a variation of theorem V.$6.2.$ in \cite{Bougerol}. First, we define the term $\zeta$-integrability of 
a measure $\mu$.
\begin{definition}\label{defnintegrable}
	Let $\mu$ be a measure on  $\mathrm{GL}(\mathbbm{K},k)$ with $k\in\N$. We call $\mu$ $\zeta$-integrable, if there exists $\zeta>0$ such that
\begin{align}\label{eq_largeDev_finiteCond}
	\int \Norm{g}^\zeta \mu(dg)< \infty \;.
\end{align}
\end{definition}
\begin{remark}
\label{NormExpectation}
For unitary matrices $U$ and $z\in\T$ the norm of the transfer matrix
\[
\tau_z(U)=\frac{1}{a}\left(
\begin{array}{cc}
\frac{\det{U}}{z} & c\\
-b & z
\end{array}
\right)\, ,\quad \text{with}\quad U=
\left(
\begin{array}{cc}
a & b\\
c & d
\end{array}
\right)\in \UG\, ,
\]
can by bounded by $\Norm{\tau_z(U)}\leq 2/|a|$. Therefore, $\Expect{\Norm{T}^\zeta}$ is finite for some $\zeta>0$ if $\Expect{|a|^{-{\zeta}}}$ is finite.
\end{remark}

\begin{lemma}
\label{F\"{u}rstenbergLemma}
Let $\{\mu_z\}_{z\in\T}$ be a family of $\zeta$-integrable measures on $\mathrm{SL}(\R,k)$. If the groups $\langle\mu_z\rangle\subset \mathrm{SL}(\R,k)$ are non-compact and strongly irreducible over $\R$, then

\begin{itemize}
\item[i)]The Lyapunov exponent satisfies $\gamma(z)>0$ for all $z\in\T$. For all vectors $v\neq 0$ in $\R^k$
\[
\lim_{n\rightarrow\infty}\frac{1}{n}\log{\Norm{g_n\cdot \ldots \cdot g_1\cdot v}}= \gamma (z)
\]
with probability one.
\item[ii)] Defining the function
\[
\Phi_z (\bar x):=\int \log\frac{ \Norm{ g x}}{\Norm{x}}\mu_z (dg)
\]
we have the relation $\gamma (z) = N_z (\Phi_z)$ and for fixed $z\in\T$ the function $\Phi_z$ is continuous.
\item[iii)] There exists $\sigma>0$ such that for any $v\neq 0$ in $\R^k$ and $\varepsilon>0$ there is $N\in\N$ with
\[
\Prob{\Norm{g_n\ldots g_1v}>\ex{n\,(\gamma -\varepsilon)}}>1-\ex{-\sigma\,n}
\]
for all $n\geq N$.
\end{itemize}
\end{lemma}
Given a measure $\mu$ on a subgroup $G\subset \mathrm{GL}(\C,k)$ we can extend this measure by an arbitrary measure $\nu$ on $\T$, \ie\ we consider the product measure $\mu_\nu=\mu\times \nu$ on $\T\times G$. For arbitrary $g\in G$ and $\ex{i\phi} \in \T$ it is clear that $\Norm{\ex{i\phi}\times g }= \Norm{g}$. Hence, we can simulate the random variable $\Norm{g_n\ldots g_1v}$, where the $g_i$ are drawn according to $\mu$, by drawing the $g_i$ according to $\mu_\nu$. But this means we may equivalently consider the group $\langle\mu_\nu\rangle $ instead of $\langle\mu\rangle$ when analyzing the norm growth of the matrix products. If we choose the uniform distribution on $\T$ as $\nu$ we have $\langle\mu_\nu\rangle=\T\langle\mu\rangle$ since $\mathrm{supp}(\nu)=\T$ in this case.

The next lemma is the desired generalization of \ref{F\"{u}rstenbergLemma} for complex valued matrices.
\begin{lemma}
\label{ComplexF\"{u}rstenbergLemma}
Let $\{\mu_z\}_{z\in\T}$ be a family of $\zeta$-integrable measures on $\mathrm{SL_\T}(k)$. If the group $\langle\mu_z\rangle\subset \mathrm{SL_\T}(k)$ are non-compact and strongly irreducible over $\C$, then
\begin{itemize}
\item[i)]The Lyapunov exponent satisfies $\gamma(z)>0$ for all $z\in\T$. For all vectors $v\neq 0$ in $\C^k$
\[
\lim_{n\rightarrow\infty}\frac{1}{n}\log{\Norm{g_n\cdot \ldots \cdot g_1\cdot v}}= \gamma (z)
\]
with probability one.
\item[ii)] Defining the function
\[
\Phi_z (\bar x):=\int \log\frac{ \Norm{ g x}}{\Norm{x}}\mu_z (dg)
\]
we have the relation $\gamma (z) = N_z (\Phi_z)$ and for fixed $z\in\T$ the function $\Phi_z$ is continuous.
\item[iii)] There exists $\sigma>0$ such that for any $v\neq 0$ in $\C^k$ and $\varepsilon>0$ there is $N\in\N$ with
\[
\Prob{\Norm{g_n\ldots g_1v}>\ex{n\,(\gamma -\varepsilon)}}>1-\ex{-\sigma\,n}
\]
for all $n\geq N$.
\end{itemize}
\end{lemma}
\begin{proof}
For the most part, the lemma is just a translation of lemma \ref{F\"{u}rstenbergLemma} from $\C^{k}$ to $\R^{2k}$ using the propositions of this section. It remains to prove that without loss of generality we may consider the group $\T \langle\mu\rangle$ instead of $\langle\mu\rangle$ such that we can apply corollary \ref{PhasesIncluded}. But, as already mentioned, the norm of an element $g_\T\times g_G\in \T\times G$ only depends on the factor $g_G$. Moreover, $g_\T$ and $g_G$ are independently drawn according to the product measure $\mu \times \nu$ where $\nu$ is the uniform distribution on $\T$. Therefore, the cumulative distribution function $F_X(x)=\Prob{X\leq x}$ for the two random variables $X=\Norm{g_n\ldots g_1 v}$ with $g_i$ either given by $g_{G}$ or $g_{\T }\times g_{G}$ coincide and the lemma is proven.
\end{proof}
Equipped with these basic results we now consider the space $\C^2$ and prove certain regularity properties of the invariant measure $\nu_z$ and later on the H\"{o}lder-continuity of $\gamma(z)$. Still, these results are known for $\K=\R$, but their translation to the complex domain is less obvious. The first lemma we now prove gives the invariant measure $\nu_z$ an interpretation as the probability measure of a random variable on $P\C^2$. Let us denote the n-fold left product $g_1 \cdot \dots \cdot g_n$ with $g_i \in \mathrm{supp}(\mu_z)$ by $m_z^{\bf n}$.  
\begin{lemma}\label{convergenceinvaiantmeasure}
	 Let $\mu_z$ be a $\zeta$-integrable measure on $\mathrm{GL}(\C,2)$, with $\Abs{\det(g)}=1$ for $g\in\mathrm{supp}(\mu_z)$ and suppose that $\langle\mu_z\rangle$ is non-compact and strongly irreducible. It follows that $m_z^{\bf n} \cdot \norm{m_z^{\bf n}}^{-1}$ converges almost surely to a rank one matrix and its direction is a random variable which is distributed according to the invariant measure $\nu_z$ corresponding to $\mu_z$.
\end{lemma}

\begin{proof}
	It is clear from lemma \ref{ComplexF\"{u}rstenbergLemma} (i) that if the Lyapunov exponent is positive then the biggest singular value of the product $m_z^{\bf n}$ has to grow exponentially with probability one. Since $|\det(m_z^{\bf n})| = 1$, the second singular value is equal to the inverse of the first one, and hence decreases exponentially. Hence  $m_z^{\bf n} \cdot \norm{m_z^{\bf n}}^{-1}$ converges almost surely to a rank one matrix. Let $m_z(\omega)$ denote any such limit point. Since $\langle\mu_z\rangle$ is strongly irreducible, the invariant measure $\nu_z$ is continuous and hence the equation
	\begin{align*}
		\int f(\bar{x}) \nu_z(m_z(\omega)d\bar x) = \int f\left(\frac{m_z(\omega) \bar x}{\norm{m_z(\omega) \bar x}}\right) \nu_z(d\bar x)
	\end{align*}
	defines a measure $m_z(\omega) \nu$ on $P\C^2$. If $m_z^{\bf n} \cdot \norm{m_z^{\bf n}}^{-1} \nu$ is defined accordingly, but with $m_z(\omega)$ replaced with $m_z^{\bf n} \cdot \norm{m_z^{\bf n}}^{-1}$, then the measures $m_z^{\bf n} \cdot \norm{m_z^{\bf n}}^{-1} \nu$ converge weakly to $m_z(\omega) \nu$. Moreover, by lemma II.2.1 in \cite{Bougerol}, the measures $m_z^{\bf n} \cdot \norm{m_z^{\bf n}}^{-1} \nu$ converge to some probability measure on $P\C^2$, which according to the above considerations is equal to a Dirac point measure $\delta_{D(\omega)}$. Hence we have  $m_z(\omega) \nu = \delta_{D(\omega)}$. This proves the assertion.
\end{proof}

\begin{lemma}[Iwasawa Decomposition]\label{Iwasawadecomp}
	Let $M \in \mathrm{GL}(\mathbbm{\C},k)$ with $k\in\N$. Then there exist a unitary matrix $U(M)$ and a lower triangular matrix $s(M)$, with positive diagonal entries such that
	\[ M = s(M) \cdot U(M) \]
	Moreover, this decomposition is unique and we have that
	\[ \frac{M M^* e_1}{\norm{M^* e_1}^2} = \frac{s(M) e_1}{\Scp{s(M)e_1}{e_1}} \;, \]
	where $e_1, \dots, e_n$ denotes the standard basis of $\C^n$.
\end{lemma}

\begin{proof}
	The existence, as well as the uniqueness of the decomposition, follows by applying the Gram-Schmidt procedure to the vectors $M^* e_1, \dots, M^* e_n$. The second statement is easily seen after noting that $s(M)^* e_1 = \Scp{s(M)e_1}{e_1} e_1$.
\end{proof}

\begin{lemma}\label{finiteexpectsM}
	Let $\mu_z$ be a $\zeta$-integrable measure on $\mathrm{GL}(\C,2)$, with $\Abs{\det(g)}=1$ for $g\in\mathrm{supp}(\mu_z)$ and suppose that $\langle\mu_z\rangle$ is non-compact and strongly irreducible. Then there exists an $\alpha > 0$ such that
\begin{align*}
	\sup_V \Expect{\left[\frac{\norm{s(V^* m_z^{\bf n} V)e_1}}{\Scp{s(V^* m_z^{\bf n} V) e_1}{e_1}}\right]^\alpha} \; < \; \infty \;,
\end{align*}
where $V$ runs over the unitary subgroup of $\mathrm{GL}(\C,2)$ and $s(V^* m_z^{\bf n} V)$ denotes the lower triangular matrix in the Iwasawa decomposition of $V^* m_z^{\bf n} V$.
\end{lemma}

\begin{proof}
  To shorten notation within this proof let us define for a matrix $m$
  \begin{align}\label{eq_def_tau}
    \tau(m):= \frac{s(m)}{\Scp{s(m)e_1}{e_1}}\; \text{ and }\; \eta(m) = \tau(m) e_1 - e_1
  \end{align}
  Therefore our goal is to bound the expectation value of $\Norm{\tau(V^* m_z^{\bf n} V) e_1}^\alpha$ from above for some $\alpha>0$.

  Let us now study the Iwasawa decomposition of a product of two matrices $m_1$ and $m_2$. Applying the decomposition either to the full product or subsequently we get
  \begin{align*}
    m_1 m_2 &= s(m_1 m_2) U(m_1 m_2)\\
    m_1 m_2 &= s(m_1) U(m_1) m_2 = s(m_1)s(U(m_1)m_2) U(U(m_1)m_2)
  \end{align*}
  which implies by the uniqueness of the decomposition
  \begin{align}
    s(m_1 m_2) = s(m_1)s(U(m_1)m_2)
  \end{align}
  Inserting this identity into our definition of $\eta$ we find
  \begin{align*}
     \eta(V^* m_z^{\bf n} V) = \eta(V^* m_1 V) +\sum_{k=1}^{n-1} \tau(V^*m_z^{\bf k} V) \eta(U(V^*m_z^{\bf k} V) V^* Y_{k+1} V)
  \end{align*}
  for a left product $m_z^{\bf n}$ of matrices $m_i$. Returning to the quantity of interest we find
  \begin{align*}
    &\Norm{\tau(V^* m_z^{\bf n} V) e_1} = \\
    &= \Norm{\eta(V^* m_z^{\bf n} V)+ e_1}\leq \Norm{\tau(V^* m_1 V e_1)}+\sum_{k=1}^{n-1} \Norm{\tau(V^*m_z^{\bf k} V) \eta(U(V^*m_z^{\bf k} V) V^* m_{k+1} V)}
  \end{align*}
  The first term on the right hand side is bounded from above by one so we are left to study the remaining summands. First note that $\eta(m)$ is by construction orthogonal to $e_1$, so in order to upper bound it, we have only to consider vectors along the $e_2$ direction. Rearranging the the expression a little bit and abbreviating $U(V^*m_z^{\bf k}V)m_{k+1}$ by $x$ we get
  \begin{align}\label{eq_secfact}
    \frac{\Norm{s(V^*m_z^{\bf k}V)\eta(x)}}{\Norm{\eta(x)}} \frac{\Norm{\eta(x)}}{\Scp{s(V^*m_z^{\bf k} V)e_1}{e_1}}\leq \Norm{\eta(x)} \frac{\Norm{s(V^*m_z^{\bf k}V)e_2}}{\Abs{\Scp{s(V^*m_z^{\bf k} V)e_1}{e_1}}}
  \end{align}
  The expectation of the first factor can be bounded by properties of the Iwasawa Decomposition (see lemma  \ref{Iwasawadecomp}) and the determinant condition by
  \begin{align}
   \Expect{ \Norm{\eta(x)}^\alpha}\leq  \Expect{ \frac{\Norm{x x^* e_1}}{\Norm{x^*e_1}^2}^\alpha}\leq \Expect{ \Norm{x}\Norm{x^{*-1}}^\alpha}\leq\Expect{\Norm{m_{k+1}}^{2\alpha}}
  \end{align}
  Due to the $\zeta$-integrability the expectation value on the right hand side is finite for some $\alpha>0$ and independent of $k$.
  We now consider the second factor of equation \eqref{eq_secfact}. Since $\Norm{s(x)e_2}=\Scp{s(x)e_2}{e_2}$ and the lower triangular form of $s(x)$ for every $x$ we find
  \begin{align}
     \Expect{\frac{\Norm{s(V^*m_z^{\bf k}V)e_2}}{\Abs{\Scp{s(V^*m_z^{\bf k} V)e_1}{e_1}}}}^\alpha\leq \Expect{\frac{\det s(V^*m_z^{\bf k}V)}{\Norm{V^*m_z^{\bf k} V e_1}^2}}^\alpha\leq\Expect{\Norm{V^*m_z^{\bf k} V e_1}^{-2\alpha}}
  \end{align}
  Since $\gamma(z)$ is positive there are constants $C>0$ and $0<\beta<1$ such that this expectation value can be bounded by $C\beta^k$.
  Putting everything together we get
  \begin{align}
    \sup_V \Expect{\left[\frac{\norm{s(V^* m_z^{\bf n} V)e_1}}{\Scp{s(V^* m_z^{\bf n} V) e_1}{e_1}}\right]^\alpha}\leq  \Expect{\Norm{\eta(V^* m_1 V)}^\alpha} + \sum_{j=1}^\infty \Expect{\Norm{m_z^{\bf k} V e_1}^{-2\alpha}}\Expect{\Norm{m_{1}}^{2\alpha}}
  \end{align}
  which is finite for some $0<\alpha<1$, because the infinite series is bounded from above by a converging geometric series.

\end{proof}

This leads us to the following corollary, which provides us with an estimate on the regularity properties of the invariant measure.

\begin{corollary}\label{cor_invMeas}
	Let $\mu_z$ be a $\zeta$-integrable measure on $\mathrm{GL}(\C,2)$, with $\Abs{\det(g)}=1$ for $g\in\mathrm{supp}(\mu_z)$ and suppose that $\langle\mu_z\rangle$ is non-compact and strongly irreducible. Denote its invariant measure by $\nu_z$. Then there exists an $\alpha > 0$ such that
	\begin{align*}
		\sup_{y \in P\C^2} \int \left(\frac{\norm{x}}{|\Scp{x}{y}|}\right)^\alpha \nu_z(dx) \; < \; \infty \;.
	\end{align*}
\end{corollary}

\begin{proof}
	Let $a_{1,2}(n)$ be the singular values of the $n$-fold product $m_z^{\bf n}$ and note that because we have $\Abs{\det(g)}=1$ for $g\in\mathrm{supp}(\mu_z)$, it follows that $a_2(n) = a_1(n)^{-1}$. Since $\langle\mu_z\rangle$ is non-compact and strongly irreducible, we have that $a_1(n) \to \infty$ for $n \to \infty$, and moreover, the direction of the rank one matrix $\lim\limits_{n \to \infty} m_z^{\bf n} \norm{m_z^{\bf n}}^{-1}$ is distributed according to $\nu_z$. Let
	\[ m_z^{\bf n} = U_z^n A_z^n K_z^n \]
	be the polar decomposition of $m_z^{\bf n}$, with $U_z^n$, $K_z^n$ being unitary matrices. It follows from the above arguments that the distribution of the limit $\lim\limits_{n \to \infty} U_z^n e_1$ is again given by $\nu_z$. Now, since for any unit vector $u \in \C^2$,
	\[ \Scp{m_z^{\bf n}  (m_z^{\bf n})^* e_1}{u} = a_1(n)^2 \Scp{U_z^n e_1}{e_1}\Scp{e_1}{U_z^n u}+ O(a_1^{-1}(n))+O(1) \]
	we have by lemma \ref{Iwasawadecomp} that
	\[ \lim_{n \to \infty} \frac{s(m_z^{\bf n})e_1}{\Scp{s(m_z^{\bf n})e_1}{e_1}} = \frac{D}{\Scp{D}{e_1}} \]
	where $D$ is a random variable with distribution $\nu_z$. Next, let $y \in P\C^2$ be some unit vector and let $V$ be defined by $y = V^* e_1$. If we replace each matrix $g$ in the $n$-fold product $m_z^{\bf n}$ by $V g V^*$, then all arguments remain valid up to replacing $m_z^{\bf n}$ by $V m_z^{\bf n} V^*$ and $D$ by $VD$. It follows that
	\[ \lim_{n \to \infty} \frac{s(V m_z^{\bf n} V) e_1}{\Scp{s(V m_z^{\bf n} V) e_1}{e_1}} = \frac{V D}{\Scp{D}{y}} \;. \]
	Since the distribution of $D$ is $\nu_z$, we have that
	\begin{align*}
		\int \left(\frac{\norm{x}}{|\Scp{x}{y}|}\right)^\alpha \nu_z(dx) \;\leq\; \lim_{n \to \infty} \Expect{\left[\frac{\norm{s(V^* m_z^{\bf n} V)}}{\Scp{s(V^* m_z^{\bf n} V) e_1}{e_1}}\right]^\alpha} \;,
	\end{align*}
	which is finite for some $\alpha > 0$ by lemma \ref{finiteexpectsM}.
\end{proof}

Finally, we give a proof of corollary \ref{TransNormGrow}. To be more precise, we prove the following more general lemma, which implies \ref{TransNormGrow} when taking into account the properties of the transfer matrices.

\begin{corollary}\label{cor_largeDev_scp_comp}
Let $\mu$ be a $\zeta$-integrable measure on $\mathrm{GL}(\C,2)$, with $\Abs{\det(g)}=1$ for $g\in\mathrm{supp}(\mu)$ and suppose that $\langle\mu_z\rangle$ is non-compact, strongly irreducible and contractive. Then there exists $\gamma, \sigma, \varepsilon_0 >0$ such that for every $\varepsilon_0>\varepsilon > 0$ there is a $N \in \N$ such that
\[
\Prob{\Abs{\Scp{v_1}{ \NPNonInvProd{n_0}{z}\; v_2}} > \ex{ (\gamma -\varepsilon )\, n}} > 1 - \ex{-\sigma \, n}
\]
holds for all $n\geq N$ and all normalized vectors $v_1$, $v_2 \in \C^2$.
\end{corollary}

For the proof of this corollary we follow along the lines of \cite{damanik02}. The actual proof for corollary \ref{cor_largeDev_scp_comp} relies on the following result.
\begin{lemma}\label{lem_largedef_notscp}
Let $\mu$ be a $\zeta$-integrable measure on $\mathrm{GL}(\C,2)$, with $\Abs{\det(g)}=1$ for $g\in\mathrm{supp}(\mu)$ and suppose that $\langle\mu_z\rangle$ is non-compact, strongly irreducible and contractive. Then
  there is a $\varepsilon_0>0$ such that for all $0<\varepsilon<\varepsilon_0$ there is a $\delta>0$ and $N\in\N$ such that for all $n\geq N$ and $y\in P\C^2$
  \begin{align}
    \sup_{x\neq 0} \Prob{\frac{\Abs{\Scp{g_z^{\bf n}x}{y}}}{\Norm{g_z^{\bf n}x}}<e^{-\varepsilon n}}<e^{-\delta n}
  \end{align}
\end{lemma}

\begin{proof}
  As already mentioned the proof is an adaption of \cite{damanik02,Bougerol} proposition VI.2.2. for the case of $P\R^2$ to our setting. In order to bound the probability
  \begin{align}\label{eq_prob_largedev}
    \Prob{\frac{\Scp{g_z^{\bf n}x}{y}}{\Norm{g_z^{\bf n}x}}<e^{-\varepsilon n}}= \int  \chi \left(\frac{\Scp{g x}{y}}{\Norm{g x}}<e^{-\varepsilon n}\right) \mu_z^{n}(dg)
  \end{align}
  we are looking for a function that upper bounds the characteristic function $\chi$. A possible choice is given by
  \begin{align}
    \Gamma_n(x) = f_n(\Abs{\Scp{\frac{x}{\Norm{x}}}{y}}),\; \; \; f_n(z) = \begin{cases} 1 & 0\leq z\leq e^{-\varepsilon n}
     \\2-ze^{\varepsilon n}  & e^{-\varepsilon n} \leq z\leq 2e^{-\varepsilon n} \\
      0 & 2e^{-\varepsilon n}\leq z\leq 1 \end{cases}
  \end{align}
Computing the desired probability we get from equation \eqref{eq_prob_largedev}
\begin{align}\label{eq_prob_gamma_dec}
  \Prob{\frac{\Scp{g_z^{\bf n}x}{y}}{\Norm{g_z^{\bf n}x}}<e^{-\varepsilon n}}\leq \Expect{\Gamma_n(g_z^{\bf n}x)}\leq \Abs{\int\Gamma_n(g x)\mu_z^n(dg) - \int\Gamma_n(x) \nu_z(dx)}+\Abs{\int\Gamma_n(x) \nu_z(dx)}
\end{align}
using the definition of the expectation value and the triangle inequality. The first term can be bounded by  a theorem about the convergence of the invariant measure (see V.2.5 in \cite{Bougerol}), which states that with the given assumptions there is a $\alpha_0>0$ such that for all $0<\alpha<\alpha_0$ there are finite constants $0<C_\alpha$ and $0<\rho_\alpha<1$ such that
\begin{align}
  \Abs{T^n(\Gamma_n(g_z^{\bf n}x))-\nu_z(\Gamma_n)}\leq\Norm[\alpha]{T^n(\Gamma_n(g_z^{\bf n}x))-\nu_z(\Gamma_n)}\leq C_\alpha \rho_\alpha^n \Norm[\alpha]{\Gamma_z}
\end{align}
holds. Noting that $x\in P\C^2$ can be written as $x= (r , \sqrt{1-r^2}e^{i \phi})$ with $0<r<1$ , using the definition of $\Gamma$, the mean value theorem for $f$ and the unitary invariance of the scalar product and $\delta$ we find
\begin{align}\label{eq_alphabound}
  \Abs{\Gamma(x_1)-\Gamma(x_2)} \leq  \Abs{\Abs{\Scp{ x_1}{y}}-\Abs{\Scp{ x_2}{y}}}e^{\varepsilon n}\leq \Abs{r_1 - r_2} e^{\varepsilon n}
  \leq \delta(x_1,x_2) e^{\varepsilon n}\; .
\end{align}
In the last step we used that with the normal form of $x\in P\C^2$ from above we have
\begin{align}
  \delta(x_1,x_2)^2 \geq (r_1\sqrt{1-r_2^2}-r_2\sqrt{1-r_1^2})^2\geq (r_1-r_2)^2
\end{align}
Equation \eqref{eq_alphabound} implies now the following bound on the $\alpha$-norm of $\Gamma_z$ for $0<\alpha <1$
\begin{align}\label{eq_bound_gamma_I}
  \Norm[\alpha]{\Gamma_z} = \Norm[\infty]{\Gamma_z} + m_\alpha(\Gamma_z)\leq \sqrt{2}e^{\varepsilon n} +1\;\; .
\end{align}
A bound on the second term in equation \eqref{eq_prob_gamma_dec} can be achieved by regularity properties of the invariant measure $\nu_z$ (see corollary \ref{cor_invMeas}). Setting $\mathcal{B}=\{x\in P\C^2;\, \Abs{\Scp{x}{y} } \leq2e^{-\varepsilon n}\}$ we get
\begin{eqnarray}\label{eq_bound_gamma_II}
  \int \Gamma_z(x) \nu_z(dx) &\leq & \nu_z(\{x\in P\C^2;\, \Abs{\Scp{x}{y} }\\
  &\leq &2e^{-\varepsilon n}\})\nonumber\\
  &\leq &\int_\mathcal{B}\frac{\Abs{\Scp{x}{y}}^\beta}{\Abs{\Scp{x}{y}}^\beta}\nu_z(dx) \nonumber \\
  &\leq & 2^\beta e^{-\varepsilon\beta n}\int \frac{\Norm{x}^\beta}{\Abs{\Scp{x}{y}}^\beta}\nu_z(dx) \nonumber
\end{eqnarray}
Using corollary \ref{cor_invMeas} the integral on the right hand side of this equation is finite for $\beta$ small enough.
Combining \eqref{eq_prob_gamma_dec} with the two bound \eqref{eq_bound_gamma_I} and \eqref{eq_bound_gamma_II} we find that for some $\alpha_0>0$, $n_0$ and $0<\alpha<\alpha_0$ there are $C_\alpha,K>0$ and $0<\rho_\alpha<1$ such that
\begin{align}
   \Prob{\frac{\Scp{g_z^{\bf n}x}{y}}{\Norm{g_z^{\bf n}x}}<e^{-\varepsilon n}}\leq 2^\beta K e^{-\varepsilon\beta n} +  C_\alpha \rho_\alpha^n(\sqrt{2}e^{\varepsilon n} +1)
\end{align}
For $0<\varepsilon < \log \rho=\varepsilon_0$ we can find $n_0$ and $\delta>0$ such that the right hand side is bounded by $e^{-\delta n}$ for all $n\leq n_0$.
\end{proof}

By using proposition \ref{ComplexF\"{u}rstenbergLemma}.iii and corollary \ref{lem_largedef_notscp} there also exists a $\varepsilon_0>0$ such that
  \begin{align*}
   \Abs{\Scp{g_z^{\bf n}x}{y}}\geq \Norm{g_z^{\bf n}x} e^{-\varepsilon n}\geq e^{(\gamma(z)-2\varepsilon)n}
  \end{align*}
  for all $0<\varepsilon<\varepsilon_0$ with probability larger then $1-e^{\delta n}-e^{\sigma n}$ for all $n$ larger then some $n_0$, we have also proven corollary \ref{cor_largeDev_scp_comp}

Our second aim was to prove that the Lyapunov exponent $\gamma (z)$ is H\"{o}lder continuous on $\T$ for the cases considered in this paper. We transfer the $z$-dependency from our measure $\mu_z$ to the matrices $g$, more precisely, we assume there is a measure $\mu$ and $z$-dependent matrices $g_z$ such that we may write
\[
\int f(g)\mu_z(dg) = \int f(g_z) \mu (dg_z)=:\Expect{f(g_z)}
\]
for all bounded measurable $f$. We still impose the following integrability assumption on the measure $\mu$
\begin{equation}
\label{IntAssum}
\exists \zeta >0 \,\forall z\in \T :\quad \Expect{{\Norm{g_z}}^\zeta} <\infty \,,
\end{equation}
which also implies that $\Expect{\log\Norm{g_z}} $ is finite. Matrix products of length $n$, where each factor is given by a matrix $g_z$ are denoted by $g_z^{\bf n}$. Another assumption concerns the dependency of the matrices $g_z$ on the spectral parameter $z\in\T$, we assume that
\begin{equation}
\label{ContMat}
 \Expect{\Norm{g_{z'}  g_z^{-1} }}\leq C|z-z'|\, ,\quad z,z'\in\T\, ,
\end{equation}
which is certainly satisfied by the matrices $\tau_z (U)$ in remark \ref{NormExpectation}.

The Lyapunov exponent admits the following representation
\[
\gamma (z) = \nu_z(\Phi_z)\,,
\]
with the function
\[
\Phi_z(\bar x)=\Expect{\log \frac{\Norm{g_z x}}{\Norm{x}}}\, ,
\]
and the unique invariant continuous measure $\nu_z$ corresponding to $\mu_z$, see lemma \ref{ComplexF\"{u}rstenbergLemma}. The following proposition concerns some continuity results about $\Phi_z$ and $\gamma (z)$ and is a summary of propositions V.4.8 and V.4.9 in \cite{Carmona}.
\begin{proposition}
\label{ImportantProp}
Given that $\langle \mu_z \rangle $ is strongly irreducible and non-compact together with assumption \eqref{ContMat} implies that
\begin{itemize}
\item[i)] the function $\Phi_z$ is continuous on $P\C^2 \times \T$ and H\"{o}lder continuous on $\T$
\item[ii)] the Lyapunov-exponent $\gamma (z)$ is continuous in $z$ and the convergence
\[
\gamma(z) =\lim\limits_{n \rightarrow \infty}\frac{1}{n}\Expect{\log \frac{\Norm{g_z^{\bf n}x}}{\Norm{x}}}
\]
is uniform in $z$ and $x$.
\end{itemize}
\end{proposition}
\begin{proof}
By lemma \ref{ComplexF\"{u}rstenbergLemma} $\Phi_z$ is continuous on $P\C^2$ for fixed $z$, hence, it remains to prove the H\"{o}lder continuity of $\Phi_z(\bar x)$ in $z$ with constants independent of $\bar x$. So, $i)$ follows from
\[
\left| \Phi_z(\bar x) -\Phi_{z'}(\bar x)\right|\leq \Expect{\log \frac{\Norm{g_{z'}\bar x}}{\Norm{g_z\bar x}}}\leq \Expect{\log \Norm{g_{z'} g_z^{-1}}}\leq C|z-z'|\, ,
\]
where the last inequality is just assumption \eqref{ContMat}.

In order to prove $ii)$ note that under our assumptions on $\mu_z$ it is guaranteed that  the invariant measure $\nu_z$ corresponding to $\mu_z$ is unique, cf. \cite{Carmona,Bougerol}. This already implies that the $\nu_z$ is weakly-continuous in $z$. If a sequence $z_n\in\T$ is converging to $\lambda$, then
\begin{eqnarray*}
\lim\limits_{n\rightarrow \infty} \gamma (z_n)&=&\lim\limits_{n\rightarrow \infty} \nu_{z_n}(\Phi(z_n))\\
&=&\lim\limits_{n\rightarrow \infty} \nu_{z_n}(\Phi_z) + \nu_{z_n}(\Phi_{z_n}-\Phi_z)\\
&=&\nu_z(\Phi_z)\\
&=&\gamma(z)
\end{eqnarray*}
Hence, $\gamma(z)$ is continuous on $\T$.

If we define probability measures $\nu_{n,z,\bar x}=\frac{1}{n}\sum_{k=0}^n \mu^k \delta_{\bar x}$ and choose sequences $z_n\in \T$ converging to $z$ and $\bar x_n \in P\C^2$ converging to $\bar x$ we see that the probability measures $\nu_{n,z_n,\bar x_n}$ converge weakly to the invariant measure $\nu_z$. Clearly, we have
\[
h_n(z,\bar x):=\frac{1}{n}\Expect{\log g_z^{\bf n} \bar x}=\nu_{n,z,\bar x} (\Phi_z)\ ,
\]
and moreover $h_n(z,\bar x)$ is continuous on $\T\times P\C^2$ by arguments similar to those used to prove $i)$. Since $\T\times P\C^2$ is compact it is sufficient to prove that $h_n(z_n,\bar x_n)$ converges to $\gamma(z)$. This follows from
\begin{eqnarray*}
\lim\limits_{n\rightarrow\infty}h_n(z_n,\bar x_n)&=&\lim\limits_{n\rightarrow\infty}\nu_{n,z_n,\bar x_n}(\Phi_{z_n})\\
&=&\lim\limits_{n\rightarrow\infty}(\nu_{n,z_n,\bar x_n}(\Phi_{z})+\nu_{n,z_n,\bar x_n}(\Phi_{z_n}-\Phi_z))\\
&=&\nu_z(\Phi_z)\\
&=&\gamma(z)
\end{eqnarray*}
\end{proof}
The H\"{o}lder continuity of the Lyapunov exponent follows from the simple decomposition
\begin{equation}
\label{DecoHC}
\gamma (z) - \gamma (z') = (\nu_z - \nu_{z'})(\Phi_{z'})+\nu_z(\Phi_z-\Phi_{z'})
\end{equation}
if we prove that the statement is true for both parts of the sum. This is easy for the second part, since it was shown in proposition \ref{ImportantProp} that the map $\Phi_z$ is H\"{o}lder continuous in $z$ independent of the argument in $P\C^2$. So, we are left to prove the H\"{o}lder continuity of the first term in (\ref{DecoHC}).

For that, let us first define the following distance on $P\C^2$,
\[
\delta (\bar x,\bar y) := \frac{\left|\Det{x ,y }\right|}{\Norm{x}\Norm{y}}\,,
\]
where the determinant of two vectors $x,y\in \C^2$ is understood as the determinant of the two-dimensional matrix with rows $x$ and $y$. The map $R_z$ defined in \eqref{MuOp} maps $\La_\alpha$ to itself, and as a mapping defined on a Banach space it inherits some nice properties. We consider $R_z$ as an element of $\B(\La_\alpha)$, the Banach algebra of bounded operators on $\La_\alpha$ with respect to the operator norm, which we denote by ${\Norm{M}}_\infty$ for $M\in \B(\La_\alpha)$. The following lemma provides all necessary tools  for the proof of the H\"{o}lder continuity of $\gamma (z)$.
\begin{lemma}
\label{BougerolLemma}
\begin{itemize}
For $z\in\T$ define the operator $Q_z=R_z-N_z$, then
\item[i)]For $\theta \in \C$ with $|\theta| > \max(1,\lim\sup\limits_{n\in \N} {\Norm{Q_z^n}}_\infty)$ the resolvent of $R_z $ exists and is given by
\[
G_{z,\theta}:=(R_z - \theta\Id )^{-1} = \frac{N_z}{\theta -1} + \sum_{n=0}^\infty \frac{Q_z^n}{\theta^{n+1}}\,.
\]
\item[ii)] There exist positive constants $C_\alpha$ and $\rho_\alpha<1$ such that ${\Norm{Q_z^n}}_\infty\leq C_\alpha \rho_\alpha^n$. Moreover, for $\theta\in \C$ with $|\theta|>1$ we have the uniform resolvent bound
\[
{\Norm{G_{z,\theta}}}_\infty\leq \frac{1}{|\theta -1|}+\frac{C_\alpha}{1-\rho_\alpha}\,.
\]
\item[iii)] For $\epsilon>0$ denote by $\B_{\epsilon}$ the circle of radius $1+\epsilon$ around $z=0$. Then we have the Cauchy identity
\[
 N_z +\Id =\frac{1}{2\pi \ii} \int\limits_{\B_\epsilon}G_{z,\theta} d\theta \,.
\]
\end{itemize}
\end{lemma}
We postpone the proof of this lemma to the end of this section. First, we prove the main result of this section, the H\"{o}lder continuity of $\gamma (z)$.
\begin{lemma}
The Lyapunov-exponent $\gamma(z)$ is H\"{o}lder-continuous on $\T$, \ie\ there exist positive constants $\alpha >0$ and $\Gamma_\alpha$ such that
\[
\left| \gamma (z)-\gamma (z')\right |\leq \Gamma_\alpha \left|z-z' \right|^\alpha\,.
\]
\end{lemma}
\begin{proof}
According to the discussion above it is sufficient to prove that the operator $N_z$ is H\"{o}lder continuous in $z$. We already know $\Phi_z \in \La_\alpha$ for arbitrary $\alpha>0$, hence we can apply lemma \ref{BougerolLemma} to see that
\[
N_z-N_{z'} = \frac{1}{2\pi\ii}\int\limits_{\B_\epsilon} (G_{z,\theta}-G_{z',\theta})d\theta\, ,
\]
hence, we may estimate 
\[
\]
\begin{equation}
\label{InvMeasEst}
|N_z (f)- N_{z'}(f)|={\Norm{N_z (f)- N_{z'}(f)}}_\alpha\leq \frac{1}{2\pi}\int\limits_{\B_\epsilon}{\Norm{G_{z',\theta}(R_{z'}-R_z)G_{z,\theta}(f) }}_\alpha d\theta\,.
\end{equation}
The crucial point is to establish a H\"{o}lder-estimate of the norm
\begin{equation}
\label{TransOpNorm}
{\Norm{(R_{z'}-R_z)(h)}}_\alpha =  {\Norm{(R_{z'}-R_z)(h)}}_\infty + m_\alpha ((R_{z'}-R_z)(h))\, ,
\end{equation}
which we do for both terms separately. For $\beta>0$ and $h\in \La_\beta$ the first term obeys the estimate
\begin{eqnarray}
\label{HolderFirstPart}
{\Norm{(R_{z'}-R_z)(h)}}_\infty &\leq&  m_\beta (h)\sup\limits_{\bar x} \Expect{\delta (g_{z'}\bar x,g_z \bar x)^\beta}\nonumber \\
&\leq &m_\beta (h)\Expect{{\Norm{g_{z'}}}^\beta{\Norm{g_z}}^\beta}\cdot |\det (g_{z'}\bar x,g_z \bar x)|^\beta \nonumber\\
&\leq & m_\beta (h)\Expect{{\Norm{g_{z'}}}^\beta{\Norm{g_z}}^\beta}\cdot |z'-z|^\beta\, ,
\end{eqnarray}
where we used $\Norm{g_z\bar x}\geq {\Norm{g_z}}^{-1}$ for $\bar x \in P\C^2$ and $|\det{g_z}|=1$ at the first inequality and \eqref{ContMat} in the last step. The integrability assumption \eqref{IntAssum} assures that there exists $\beta>0$ such that the expectation value in \eqref{HolderFirstPart} is finite, \ie\ there exists a constant $C_\beta <\infty$ such that ${\Norm{(R_{z'}-R_z)(h)}}_\infty \leq C_\beta {\Norm{h}}_\beta|z'-z|^\beta$ for positive but sufficiently small $\beta$. With such a $\beta$ at hand we choose $\alpha =\beta /2$.

Spelled out, the second term in \eqref{TransOpNorm} reads
\begin{equation}
m_\alpha ((R_{z'}-R_z)(h)) =\sup\limits_{\bar x\neq \bar y}\frac{|(R_{z'}-R_z)(h)(\bar x)-(R_{z'}-R_z)(h)(\bar y)|}{\delta (\bar x, \bar y)^\frac{\beta}{2}}
\end{equation}
\begin{equation}
|(R_{z'}-R_z)(h)(\bar x)-(R_{z'}-R_z)(h)(\bar y)|\leq A_\alpha  {\Norm{h}}_{2\alpha}\min (|z'-z|^{2\alpha } , \delta \left(\bar x, \bar y)^{2\alpha }\right)
\end{equation}
and together with the estimate $\min (C^2/y, y)\leq C$ this implies
\[
m_\alpha ((R_{z'}-R_z)(h))\leq A_\alpha {\Norm{h}}_{2\alpha}|z'-z|^{\alpha}\,.
\]
The assertion follows from lemma \ref{BougerolLemma} ii) and $A_\alpha \sup\limits_{z\in\T}{\Norm{\Phi_z}}_{2\alpha}<\infty$.
\end{proof}
\begin{proof}[Proof of lemma \ref{BougerolLemma}.]
The operators $N_z$ and $Q_z$ are both bounded and satisfy $Q_zN_z=N_zQ_z=0$ and $N_z^2=N_z$. Hence, it follows from the triangle inequality for the operator norm and the completeness of $\B(\La_\alpha )$ that the series
\[
\sum_{n=0}^\infty \frac{(N_z+Q_z)^n}{\theta^{n+1}}=\frac{N_z}{\theta -1} +\sum_{n=0}^\infty \frac{Q_z^n}{\theta^{n+1}}
\]
converges to a bounded operator for $|\theta|>\max(1,\lim\sup\limits_{n\in \N} {\Norm{Q_z^n}}_\infty )$, which proves i).

Property ii) follows from a general theorem concerning cocyles, see proposition IV.3.15 in \cite{Carmona}, if we prove the existence of an integer $N$ such that
\begin{equation}
\label{CocycleEq}
\sup\limits_{z\in\T, \bar x\neq \bar y }\int\log \frac{\delta (g_z\bar x, g_z \bar y)}{\delta (\bar x,\bar y)}\mu^N(dg_z)<0\,.
\end{equation}
Again, denoting $n$-fold products of matrices $g_z$ by $g_z^{\bf n}$, we get the following relation for the distance $\delta $
\[
\log \frac{\delta (g_z^{\bf n} \bar x, g_z^{\bf n} \bar y) }{\delta (\bar x, \bar y)} = -\log \Norm{g_z^{\bf n}\bar x}-\log \Norm{g_z^{\bf n}\bar y}\,
\]
which follows from $|\det (g_z)|=1$ and $\Norm{\bar x}=\Norm{\bar y}=1$. By \ref{ImportantProp} the two terms on the right hand side converge uniformly with respect to $z\in \T$ and $\bar x,\bar y\in P\C^2$ to $\gamma (z)$, and moreover $\gamma (z)$ is also continuous by \ref{ImportantProp}. Hence, there exists an integer $N$ such that almost surely
\[
\sup\limits_{z\in\T, \bar x\neq \bar y }\log \frac{\delta (g_z\bar x, g_z \bar y)}{\delta (\bar x,\bar y)}<0\, ,
\]
which proves \eqref{CocycleEq} and hence ii) is verified.

Now, the Cauchy identity iii) follows immediately from i) and ii).
\end{proof}

\section{Measures with nonempty interior}
This section deals with properties of transfer matrices induced by probability measures $\mu$ on $\UG$. As explained in section \ref{InducedMeasure} a measure $\mu$ induces a family of measures $\mu_z$ on $\mathcal{U}_{ND}$. Our aim is to analyze the support of the measures $\mu_z$ for measures $\mu$ with nonempty interior and prove that the group \Gmuz\ generated by $\mu_z$ is non-compact and admits no reducible subgroup of finite index. The requirements of lemma \ref{ComplexF\"{u}rstenbergLemma} namely that there is a $\zeta>0$ with $\Expect{|a|^{-\zeta}}<\infty$, are fulfilled in this case, see remarks \ref{NormExpectation} and \ref{AbsCont}. The next proposition proves that \Gmuz\ to be non-compact if the support of $\mu$ has nonempty interior.
\begin{proposition}
\label{NonComp}
Let $\mu$ be a measure on $\UG $ such that $\mathrm{supp} (\mu)$ has nonempty interior in the standard topology of $\UG $, \ie\ there exists $U\in \mathrm{supp} (\mu)$ and an open set $O_U \subset \mathrm{supp} (\mu)$ with $U\in O_U$. Then \Gmuz\ is non-compact for all $z\in \C\backslash \{0\}$.
\end{proposition}
\begin{proof}
According to lemma \ref{ImageTau} there exist $r\in \R_+$ and $\alpha ,\beta , \gamma \in [0,2\pi)$ such that the transfer matrix corresponding to $U\in O_U$ and $z\in \C\backslash \{0\}$ can be written as
\[
T(r,\alpha,\beta,\gamma) =
\left( \begin{array}{cc}
\sqrt{1+r^2}\ex{i\alpha}|z|^{-1} & r\ex{i\beta}\\
r\ex{i\gamma} & \sqrt{1+r^2}\ex{i(\beta +\gamma -\alpha)}|z|
\end{array}
\right)\,.
\]
Since $\tau_z$ and its inverse are continuous maps, we have that $\tau_z(O_U)$ is an open subset of $\mathrm{supp} (\mu_z)$. Hence, there exists $r'\neq r$ such that $T(r',\alpha,\beta,\gamma)\in \mathrm{supp}(\mu)$. A straightforward calculation shows that
\begin{equation}
\label{TtimesTInv}
T(r,\alpha,\beta,\gamma )\cdot T(r',\alpha,\beta,\gamma )^{-1}=
\left(
\begin{array}{cc}
f(r,r') &\ex{i(\alpha-\beta)}\frac{g(r,r')}{|z|}\\
\ex{-i(\alpha-\beta)}g(r,r')|z|&f(r,r')
\end{array}
\right)
\end{equation}
with
\begin{align*}
f(r,r')=\sqrt{(1+r^2)(1+r'^2)}-rr'\\
g(r,r')=r\sqrt{1+r'^2}-r'\sqrt{1+r^2}\,.
\end{align*}
Hence, the eigenvalues of this matrix are given by
\[
\lambda_\pm =f(r,r')\pm g(r,r')\,.
\]
Since the determinant of (\ref{TtimesTInv}) is of modulus one these eigenvalues have to obey $|\lambda_+ \lambda_-|=1$. Therefore, either $|\lambda_\pm|=1$ or there is an eigenvalue of (\ref{TtimesTInv}) with modulus strictly larger than 1, in which case \Gmuz\ is clearly non-compact.

Hence, it remains to prove that $|\lambda_\pm|\neq 1$. This follows from the equation
\[
||\lambda_+ |-|\lambda_- ||= 2\cdot \mathrm{min}\{ |f(r,r')|,|g(r,r')|\} \, ,
\]
and the observation that $f(r,r')=0$ is not possible and $g(r,r')=0$ leads to $r=r'$.
\end{proof}
Next we show that non emptiness of $\mu_z$  implies strong irreducibility of \Gmuz. 
\begin{proposition}
\label{StrongIrr}
Let $\mu$ be a measure on the group of unitaries $\UG $ such that the interior of $\mathrm{supp}(\mu)$ is nonempty. Then $\forall z\,\in \C\backslash \{0\}$ we have that all subgroups of finite index in \Gmuz are irreducible.
\end{proposition}
\begin{proof}
We choose an open set $O\in \mathrm{supp}(\mu )$ and consider the corresponding set $\tau_z (O) \subset \mathrm{GL}(\C,2)$. Since $\tau_z$ and its inverse are continuous $\tau_z (O)$ is also open in the standard topology of $\mathrm{GL}(\C,2)$. For $M\in\tau_z(O)$ the set $\tau_z(O)\cdot M^{-1}$ is an open neighborhood of $\Id$, hence there exists $\varepsilon >0$ such that
\[
\{\Id + \delta U \,:\, 0\leq \delta \leq \varepsilon \, ,\, U\in\UG \}\subset \tau_z(O)\cdot M^{-1}\,.
\]
Therefore, for arbitrary normalized vector $v$ we have
\[
\mathcal{B}_\varepsilon(v)\subset \tau_z(O)\cdot M^{-1}\cdot v \,,
\]
which implies that the set $\tau_z(O)\cdot w$ has nonempty interior for all normalized $w\in\C^2$. This already excludes strong irreducibility, since the interior of a finite union of subspaces is empty.

\end{proof}
\begin{remark}
\label{AbsCont}
For measures $\mu$ having a component which is absolutely continuous with respect to the Haar measure $H$ the support $\mathrm{supp}(\mu)$ has nonempty interior. Moreover, for all $\zeta>0$ we have
\[
\mathbbm{E}_\mu \left(|a|^{-\zeta}\right)=\mathbbm{E}_{\mu_c} \left(|a|^{-\zeta}\right)+\mathbbm{E}_{\mu_s} \left(|a|^{-\zeta}\right)\, ,
\]
where $\mu_c$ resp. $\mu_s$ denote the absolutely continuous resp. singular part of the measure $\mu$. The first term can be estimated by the expectation value with respect to the Haar measure
\[
\mathbbm{E}_{\mu_c} \left(|a|^{-\zeta}\right)\leq {\Norm{\mu_c}}_\infty\mathbbm{E}_H \left(|a|^{-\zeta}\right)={\Norm{\mu_c}}_\infty\int_0^\pi \frac{1}{|\sin{\theta}|^\zeta}d\theta\, ,
\]
which is finite for all $0\leq\zeta<1$. The second term is finite if the singular measure $\mu_s$ consists of finitely many points\footnote{Note, that we can assume that each of these points fulfills $|a|>0$, see section \ref{InducedMeasure}.}.
\end{remark}


\label{AppFurst}
\section{Discrete measures}
\label{Appdiscretemeasure}
In this section we give a proof of corollary \ref{finitecoinset} that establishes the dynamical localization of a disordered quantum walk, where the coin operations are drawn from the set
\[
\left\{
H=\frac{1}{\sqrt{2}}\left(
\begin{array}{cc}
1 & 1 \\
1 & -1
\end{array}
\right)\, ,\,
X=\left(
\begin{array}{cc}
a & b \\
-\bar{b} & \bar{a}
\end{array}
\right)\,
\right\}\, ,
\]
with $a,b \in\C$ fixed, but arbitrary, such that $0<|a|^2+|b|^2=1$ and $|a|<|\mathfrak{I}b|$. The actual probability distribution according to which the coins are drawn is irrelevant for our purposes. Our aim is to prove that the group generated by the corresponding transfer matrices satisfies the requirements of theorem \ref{thmtransfermatrices} for all $\theta_0\in \T$, \ie\ we have to prove that $\Gmu[\mu_{\theta_0}]$ is $\zeta$-integrable, non-compact and possesses no reducible subgroup of finite index. The first requirement holds easily for $a\neq 0$.
We proceed by proving that \Gmuz\ is non-compact for all $z\in\C\backslash \{0\}$.
\begin{proposition}
\label{NonCompact}
Let $\mu$ be a measure in $\UG$ such that $\mathrm{supp}(\mu)=\{H,X\}$, then \Gmuz\ is non-compact for all $z\in\C\backslash\{0\}$.
\end{proposition}
\begin{proof}
The transfer matrices corresponding to $H$ and $X$ are given by
\[
T_H=\tau_z (H)=\left(
\begin{array}{cc}
-\frac{\sqrt{2}}{z} & 1 \\
-1 & \sqrt{2}z
\end{array}
\right)\quad \mathrm{and} \quad
T_X=\tau_z( X )=\frac{1}{a}\left(
\begin{array}{cc}
\frac{1}{z} & -\bar{b} \\
-b & z
\end{array}
\right) \,.
\]
Similarly to the proof of proposition \ref{NonComp} we consider the product
\[
T_H\cdot T_X^{-1}=
\frac{1}{a}\left(
\begin{array}{cc}
\bar{b}-\sqrt{2} & (1-\sqrt{2}\bar{b})z^{-1}\\
(\sqrt{2}b-1)z & \sqrt{2}-b
\end{array}
\right)
\]
and determine the modulus of its eigenvalues, which is given by
\[
|\lambda_\pm|=\frac{|i\mathfrak{I}b\pm\sqrt{|a|^2-|\mathfrak{I}b|^2}|}{|a|}\,.
\]
Under the assumption $|a|<|\mathfrak{I}b|$ there exists an eigenvalue $\lambda_\pm$ with modulus strictly larger than one.
\end{proof}
In order to prove that a disordered quantum walk according to corollary \ref{finitecoinset} exhibits dynamical localization we have to prove that \Gmuz\ is strongly irreducible for $z\in\C\backslash \{0\}$. This can be shown by using proposition II$.4.3.$ in \cite{Bougerol}. The original proof in \cite{Bougerol} is given for $\mathrm{GL}(\R,2)$, therefore and for the convenience of the reader we repeat this proof for $\mathrm{GL}(\C,2)$.
\begin{proposition}
Let $\mu$ be a measure on $\mathrm{GL}(\C,2)$ such that \Gmu\ is non-compact and $|\det{M}|=1$ for all $X\in \mathrm{supp}(\mu)$. Then \Gmu\ is strongly irreducible if for any $\bar v$ in the projective space $P\C^2$ the set $\{X\,\bar v\,:\,X\in \Gmu\}$ contains more than two elements.
\end{proposition}
\begin{proof}
Suppose there are elements $\bar v_1\,\ldots \bar v_n \in P\C^2$ such that $\Gmu\,\bar v_1\cup \ldots \cup \bar v_n=\bar v_1\cup\ldots \cup \bar v_n$. Each $g\in \Gmu$ induces a permutation $\pi_g$ on the set $\{\bar v_1,\ldots \bar v_n\}$ and the map $g\rightarrow \pi_g$ is a group homomorphism. Denoting the kernel of this homomorphism by $K$ we can construct the quotient group $G/K$ which is isomorphic to the symmetric group $S_n$, hence $G/K$ is finite. Now, every $g\in K$ acts on the corresponding vectors $v_i \in \C^2$ as $Mv_i=\lambda_iv_i$. Assuming $n>2$ we can choose three elements $\bar v_1,\,\bar v_2$ and $\bar v_3$ and find coefficients $\alpha,\beta \in \C$ with $\alpha\neq 0 \neq \beta$ such that $v_3 =\alpha v_1 +\beta v_2$. Applying $g\in K$ yields the equation
\[
\lambda_3\alpha v_1 +\lambda_3\beta v_2 = \lambda_1\alpha v_1 +\lambda_1\beta v_2\,,
\]
which implies $\lambda_1=\lambda_2=\lambda_3=\lambda$. Hence, $g=\lambda\mathbbm{1}$ for all $g\in K$ and since $|\det{g}|=1$ this implies $\lambda=\ex{i\phi}$ with $\phi\in \R$. This is a contradiction to the non-compactness of \Gmu\, which excludes that both $K$ and $G/K$ are compact.
\end{proof}
It remains to prove that the transfer matrices corresponding to $H$ and $X$ generate more than two elements in $P\C^2$ when applied to an arbitrary element $\bar v\in P\C^2$. We prove a slightly weaker, but sufficient result. In order to apply lemma \ref{initiallengthscale} we might show that \Gmuz\ is non-compact and strongly irreducible for all but finitely many $z\in \T$, hence it suffices to prove strong irreducibility for almost all $z\in \C$.
\begin{proposition}
Let $\bar v\in P\C^2$, $z\in \C^2\backslash \{0\}$ and \Gmuz the group generated by $T_H$ and $T_X$. Then $\Gmuz\bar v$ contains more than two elements for all but at most two $z\in \C\backslash\{0\}$.
\end{proposition}
\begin{proof}
The elements $\bar v\in P\C^2$ are given by the vectors
\[
\bar v_{r,\phi}=\genfrac(){0pt}{}{ r}{ \sqrt{1-r^2}\ex{i\phi}}\,,\quad r,\,\phi \in \R\,.
\]
In order to check whether there are three different elements in $\Gmuz\bar v_{r,\phi}$ we first observe that this is the case if $v_{r,\phi}$ is not one of the two eigenvectors of $T_H\cdot T_X^{-1}$ because then $\bar v_{r,\phi},\,T_H\cdot T_X^{-1}\bar v_{r,\phi}$ and $(T_H\cdot T_X^{-1})^2 \,\bar v_{r,\phi}$ are linearly independent. The same argument applies for the eigenvectors of $T_H^{-1}\cdot T_X$, hence we have to analyze the cases when some eigenvectors of $T_H\cdot T_X^{-1}$ and $T_H^{-1}\cdot T_X$ coincide. The condition that one eigenvector of $T_H\cdot T_X^{-1}$ coincides with one eigenvector of $T_H^{-1}\cdot T_X$ gives us a relation depending on $b$ and $z$. There are four such relations corresponding to the four possible choices of signs in
\[
\frac{\mathfrak{R}b-\sqrt{2}\pm i\gamma}{\sqrt{2}b-1}=z^2\frac{\mathfrak{R}b+\sqrt{2}\pm i\gamma}{\sqrt{2}b+1}\, ,
\]
with $\gamma=\sqrt{|\mathfrak{I}b|^2-|a|^2}\in \R$. Note that $b=\pm 1/\sqrt{2}$ is excluded by the condition $|a|<\mathfrak{I}b$. There are at most two solutions for $z$ to this polynomial equation, hence there are at most two $z$ such that \Gmuz\ is not strongly irreducible.
\end{proof}

\section{Absence of continuous spectrum}

For the sake of completeness we include in this chapter the proof of a unitary and discrete time version
of the RAGE theorem \cite{ragethm}. For this purpose we follow along the lines of
a proof for the Hamiltonian case, where the time evolution is induced
by a self adjoint operator \cite{teschl2009,Kirsch:2007bf}. The goal of this
section is therefore to connect the spectral properties of a realization \W\ of a family the disordered walk operators
$\{\W\}$ with the dynamical behavior of vectors of the Hilbert space.

In order to do so, we are interested in the properties of the measure $\hat \rho^{x,y}_\omega (t)$ of the time evolution of a
realization \W\ of a disordered quantum walk quantum walk
\begin{align*}
\hat \rho^{x,y} (t) = \Scp{\LocS[y]{\phi}}{W_\omega^t \LocS{\psi}}=\int_\UC \theta^t \rho^{x,y}_\omega(d\theta)\, .
\end{align*}
We can at first prove the following version of Wiener's theorem:

\begin{theorem}
	\label{wienermeasureunitcircle}
Let $\mu$ be a complex Borel measure  on the unit circle \UC\ and define
\begin{align*}
\hat \mu(t) = \int_\UC \theta^t \mu(d\theta)
\end{align*}
then the time average of $\hat\mu(t)$ has the limit
\begin{align*}
\lim_{T\rightarrow\infty} \frac{1}{T+1} \sum_{t=0}^T \Abs{\hat\mu(t)}^2 = \sum_{\lambda \in \UC} \Abs{\mu(\{\lambda\})}^2
\end{align*}
\end{theorem}

\begin{proof}
Starting from the definition of $\hat \mu(t)$, using linearity of the integral and the that $\theta\in\UC$ we get
\begin{align*}
&\lim_{T\rightarrow\infty} \frac{1}{T+1} \sum_{t=0}^T \Abs{\hat\mu(t)}^2 \, = \\
=\, &\lim_{T\rightarrow\infty} \int_{-\pi}^\pi\int_{-\pi}^\pi\left(\frac{1}{T+1}\sum_{t=0}^T e^{\ii(x-y)t} \right)\mu(d(e^{x}))\mu^\star(d(e^y))\, .
\end{align*}
Since the geometric series within the parenthesis is bounded by $T$ and converges point wise to the indicator function $\chi_{\{0\}}(x)$ we can interchange the limit with the integration by dominated convergence  and arrive at
\begin{align*}
  &\int_{-\pi}^{\pi}\int_{-\pi}^{\pi} \chi_{\{0\}}(x-y)\mu(d(e^x))\mu^\star(d(e^y)) \,= \\
  &=\,\int_{-\pi}^{\pi} \mu(\{y\})\mu^\star(dy) = \sum_{x\in[-\pi,\pi)} \Abs{\mu(\{e^x\})}^2 \,.
\end{align*}
\end{proof}

For a given unitary $W$ we can decompose the Hilbert space into three orthogonal subspaces $\H_{ac}$, $\H_{sc}$ and $\H_{pp}$ each containing the
vectors $\phi$ for which the spectral measure
$\rho_{\phi}$ is absolutely continuous, singular continuous or pure point, respectively and these subspaces are
 left invariant by $W$.
In addition define $\H_c = \H_{ac}\oplus\H_{sc}$.
A connection between this decomposition and the time evolution of the quantum walk is given by the following theorem:

\begin{theorem}\label{thm_prerage}
Let $W$ be a unitary operator and $G$ a compact operator then for $\psi_c\in\H_{c}$
\begin{align*}
&\lim_{T\rightarrow\infty} \frac{1}{T+1}\sum_{t=0}^T \Norm{G W^t \psi_{c}}^2=0
\end{align*}
holds.
\end{theorem}

\begin{proof}
Choose a vector $\psi_c\in\H_{c}$. Since $W$ leaves $\H_{c}$ invariant, by $\Scp{\phi}{W\psi}=\Scp{P_{\H_{c}} \phi}{W\psi}$ we see that if $\rho_{\psi}$ is continuous so is $\rho_{\phi,\psi}$. For the compact operator $G$ there is a sequence of finite rank operators $G_n=\sum_{k=0}^n \alpha_k\Scp{\phi_k}{.}\eta$ converging to it. By the triangle inequality we only have to check the single rank one summands so we end up with
\begin{align*}
&\lim_{T\rightarrow\infty}\frac{1}{T+1}\sum_{t=0}^T \betrag{\Scp{\phi}{W^t\psi_{c}}}^2=\lim_{T\rightarrow\infty}\frac{1}{T+1}\sum_{t=0}^T \betrag{\hat\rho_{\phi,\psi}(t)}^2=0 \,,
\end{align*}
which follows from Wiener's theorem.
\end{proof}

Equipped with this result, we can prove the discrete time version of the Rage theorem:
\begin{theorem}[RAGE]
	\label{ragethmuni}
Let W be a unitary operator and $G_n$ a sequence of compact operators converging strongly to the identity. Then we have for $\H_c$ and $H_{pp}$
\begin{align*}
\H_c &= \{\psi\in\H\, ;\, \lim_{n\rightarrow\infty}\lim_{T\rightarrow\infty} \frac{1}{T+1}\sum_{t=0}^T \Norm{G_n W^t\psi}^2=0\}\\
\H_{pp} &= \{\psi\in\H\, ;\, \lim_{n\rightarrow\infty}\sup_{t\geq 0} \Norm{(\Id-G_n) W^t\psi}^2=0\}\; .
\end{align*}
\end{theorem}
\begin{proof}
  We start with the continuous case, which by theorem \ref{thm_prerage} holds for every $\psi\in\H_x$.  Decomposing now an arbitrary $\phi\in\H$ into $\psi_c\in\H_c$ and $\psi_p\in\H_{pp}$ we can infer that we have to find a lower bound on
  \begin{align*}
   \lim_{n\rightarrow\infty}\lim_{T\rightarrow\infty}\sum_{t=0}^T\Norm{G_nW^t\Psi_{p}}^2 + \Norm{G_nW^t\Psi_{c}}^2
   -2\betrag{\Scp{G_nW^t\Psi_{p}}{G_nW^t\Psi_{c}}}
  \end{align*}
  Since the norm of the  $G_n$ is uniformly bounded, because they converge to the identity strongly, and by theorem \ref{thm_prerage} the last two summands tend to zero.
  So we have to show, that $\Norm{G_nW^t\psi_p}$ is bounded away from zero for $n$ large enough. Instead we prove that
  \begin{align*}
   \lim_{n\rightarrow\infty}\sup_{t\geq 0}\Norm{(\Id-G_n)W^t\psi_p}=0
  \end{align*}
Being an element of $\H_{pp}$ $\psi_p$ can be decomposed into eigenvectors $\psi_k$ of the unitary operator $W$.
Inserting this decomposition for $\psi_p$ we can upper bound the norm by
\begin{align*}
 \lim_{n\rightarrow\infty}\sup_{t\geq 0}\sum_{k=1}^N\betrag{\alpha_ke^{-\ii\lambda_k}}\Norm{(\Id-G_n)\psi_k}+ \Norm{\Id-G_n}\sum_{k=N+1}^\infty\betrag{\alpha_ke^{-\ii\lambda_k}}\Norm{\psi_k}
\end{align*}
The first sum goes to zero by strong convergence of the $G_n$ and the second goes to zero if we make $N$ large enough and using the fact that a strong convergent sequence of operators is bounded. At the same time this proves the second claim of the rage theorem for $\psi\in\H_{pp}$.

If we now again decompose an arbitrary vector $\psi\in\H$ in its components in $\H_c$ and $\H_{pp}$ as in the first case we are left to prove that $\sup_{t\geq0}\Norm{(\Id-G_n)W^t\psi_c}$ stays strictly larger than zero for all $n$. Assuming the contrary we find
\begin{align*}
  \Norm{W^t\psi_c}&\leq\ \lim_{T\rightarrow\infty} \frac{1}{T+1}\sum_{t=0}^T \Norm{(\Id-G_n)W^t\psi_c} + \Norm{G_nW^t\psi_c}\\
  &\leq \sup_t \Norm{(\Id-G_n)W^t\psi_c}\xrightarrow{n\rightarrow\infty} 0\; ,
\end{align*}
by strong convergence of the $\{G_n\}$. This contradiction then concludes the proof of the rage theorem.
\end{proof}

Finally we can prove the absence of continuous spectrum of a walk operator.

\begin{lemma}
	\label{lempurepoint}
Let \W\ be a realization of a family of  walk operators $\{\W\}$ that exhibits dynamical localization, \eg\ the assumptions of theorem \ref{thmtransfermatrices} can be verified for the complete unit circle \T.
Then with probability one \W\ does not have continuous spectrum.
\end{lemma}
\begin{remark}
  Lemma \ref{lempurepoint} also holds for the open arc of the unit circle $I_\delta$ under the assumptions of theorem \ref{thmtransfermatrices}. The proof is essentially the same as the one given for lemma \ref{lempurepoint} bellow, with $\chi(I_\delta)$ inserted before all $\W^t$.
\end{remark}

\begin{proof}
In order to prove the statement we show, that for a given realization $W_{\omega}$ of the family of
 walk operators $\{\W\}$, the subspace $\H_c(\omega)$ contains with probability one only the zero vector.
 Therefore we look at the projection of an arbitrary $\psi\in\H$ onto $\H_c(\omega)$. Choosing the projectors
 $G_n=\sum_{\betrag{l}\leq n}\sum_{i} P_{\LocSe[l]{i}}$, where $P_{\LocSe[l]{i}}$ denotes the projector onto the state
 $\LocSe[l]{i}$ and by $\Norm{\psi}=\Norm{\W^t\psi}$ we get
\begin{align}
\label{eq_pcnorm_dec}
\begin{split}
&\Norm{P_{\H_c(\omega)}\psi}^2 = \\
&= \frac{1}{T+1} \sum_{t=0}^T (\Norm{(\Id-G_n)\W^t P_{\H_c(\omega)}\psi}^2 + \Norm{G_n \W^t P_{\H_c(\omega)}\psi}^2)\\
&=\frac{1}{T+1} \sum_{t=0}^T \Norm{(\Id-G_n)\W^t\psi}^2 - \Norm{(\Id-G_n)\W^t P_{\H_{pp}(\omega)}\psi}^2 \\
&+ \Norm{G_n \W^t P_{\H_c(\omega)}\psi}^2
\end{split}
\end{align}
Note that the $G_n$ are a sequence of operators of finite range that converge strongly to the identity. Taking the limit with respect to first $T$ and then $n$ of eq.
\eqref{eq_pcnorm_dec} we see that the sums over the second and the third term on the right hand side tend to zero by theorem \ref{ragethmuni}.

To make the connection to theorem \ref{thmtransfermatrices} we take the expectancy with respect to $\omega$ on both sides
and get with Fatou's lemma
\begin{align}
\label{eq_pcnorm}\begin{split}
\Expect{\Norm{P_{\H_c(\omega)}\psi}^2} &= \Expect{\lim_{n\rightarrow\infty}\lim_{T\rightarrow\infty}
 \frac{1}{T+1} \sum_{t=0}^T \Norm{(\Id-G_n)\W^t\psi}^2}\\
&\leq \liminf_{n\rightarrow\infty}\Expect{\lim_{T\rightarrow\infty}\frac{1}{T+1}\sum_{t=0}^T\Norm{(\Id-G_n)\W^t\psi}^2}\;.
\end{split}
\end{align}
In order to proof that $\H_{c}(\omega)$ only contains the zero vector almost surely, we have to check that the right hand side
vanishes for every $\psi\in\H$.
We do this by checking it for the total set $\{\LocSe{i},x\in\Z, i=1,2\}$ of our one dimensional lattice.

Inserting the definition of the $G_n$ into equation \eqref{eq_pcnorm} we find for $\psi =\LocSe{i}$ that
\begin{align*}
 & \liminf_{n\rightarrow\infty}\Expect{\lim_{T\rightarrow\infty}\frac{1}{T+1}\sum_{t=0}^T\Norm{(\Id-G_n)\W^t\LocSe{i}}^2}\\
 &\leq \liminf_{n\rightarrow\infty}\Expect{\sum_{\betrag{l}=n+1}^\infty\sum_{j=1}^2 \sup_{t\geq 0}\betrag{\Scp{\LocSe[l]{j}}{W^t \LocSe{i}}}^2}\\
  &\leq \liminf_{n\rightarrow\infty}\sum_{\betrag{l}=n+1}^\infty\sum_{j=1}^2\Expect{ \sup_{t\geq 0}\betrag{\Scp{\LocSe[l]{j}}{W^t \LocSe{i}}}^2}\\
  &\leq \liminf_{n\rightarrow\infty} \sum_{l=n+1}^\infty\sum_{j=1}^2 C_1(e^{-C_2 \betrag{l+x}^\xi} + e^{-C_2 \betrag{l-x}^\xi} )
\end{align*}
holds because of the dynamical localization of the disordered quantum walk, \ie\ lemma \ref{expdecayspectralmeasure}.
For $n$ large enough $l$ is larger than $\betrag{x}$, so we can forget about the absolute values in the exponents and the sum takes the simple form of a harmonic series and we find for the projection of $\LocSe{i}$ onto $\H_c(\omega)$
\begin{align*}
\Expect{\Norm{P_{\H_c(\omega)}\LocSe{i}}^2} \leq \liminf_{n\rightarrow\infty} 2C_1(e^{C_2x^\xi}+e^{-C_2x^\xi}) \sum_{l=n+1}^\infty e^{-C_2l^\xi}
\end{align*}
The fact that $e^{-C_2l^\xi}<1$ ensures the convergence of the harmonic series and by throwing away the first
$n$ terms and making $n$ large, the expression tends to zero as required. Positivity of  $\Norm{P_{\H_c(\omega)}\LocSe{i}}$ then already implies $P_{\H_c(\omega)}\LocSe{i}=0$ for all $x$ and $i$ almost surely and totality of the set of all $\LocSe{i}$ then completes the proof.
\end{proof}

\label{apprage}

\section*{Acknowledgement}

We would like to thank Reinhard Werner and Tobias Osborne for many valuable discussions and advice and David Gross for many interesting discussions about random matrices. VBS would like to thank Hajo Leschke for introducing him to the subject of random operators as well as the Erwin-Schr\"{o}dinger Institute in Vienna and the Mittag-Leffler institute in Stockholm, where part of the work was done. We gratefully acknowledge support by the DFG-Forschergruppe 635 and by the European projects SCALA and CoQuit.

\bibliographystyle{alpha}
\bibliography{literatur}

\newcommand{\etalchar}[1]{$^{#1}$}
\begin{thebibliography}{AVWW10}

\bibitem[ABN{\etalchar{+}}01]{Ambainis2001}
Andris Ambainis, Eric Bach, Ashwin Nayak, Ashvin Vishwanath, and John Watrous.
\newblock One-dimensional quantum walks.
\newblock In {\em {ACM} Symposium on Theory of Computing}, pages 37--49, 2001.

\bibitem[Amb03]{ambainis-2003-1}
Andris Ambainis.
\newblock Quantum walks and their algorithmic applications.
\newblock {\em Int. J. Quant. Inf.}, 1:507, 2003.

\bibitem[And58]{Anderson}
Philip~W. Anderson.
\newblock Absence of diffusion in certain random lattices.
\newblock {\em Phys. Rev.}, 109(5):1492--1505, 1958.

\bibitem[AVWW10]{Vogts2010}
Andre Ahlbrecht, Holger Vogts, Albert~H. Werner, and Reinhard~F. Werner.
\newblock Asymptotic evolution of quantum walks with random coin.
\newblock arXiv:1009.2019, 2010.

\bibitem[BHJ03]{Bourget2003}
Olivier Bourget, James~S. Howland, and Alain Joye.
\newblock {Spectral Analysis of Unitary Band Matrices}.
\newblock {\em Comm. Math. Phys.}, 234(2):19--227, 2003.

\bibitem[BL85]{Bougerol}
Philippe Bougerol and Jean Lacroix.
\newblock {\em Products of Random Matrices with Applications to {S}chr\"odinger
  Operators}.
\newblock Progress in Probability and Statistics. Birkh\"auser, 1985.

\bibitem[CKM87]{Carmona87}
Ren{{\'e}} Carmona, Abel Klein, and Fabio Martinelli.
\newblock Anderson localization for {B}ernoulli and other singular potentials.
\newblock {\em Comm. Math. Phys.}, 108(1):41--66, 1987.

\bibitem[CL90]{Carmona}
Ren{{\'e}} Carmona and Jean Lacroix.
\newblock {\em Spectral Theory of Random Schr\"{o}dinger Operators}.
\newblock Probability and Its Applications. Birkh\"auser, 1990.

\bibitem[CMR06]{bookcauchytransform}
Joseph~A. Cima, Alec~L. Matheson, and William~T. Ross.
\newblock {\em The {C}auchy {T}ransform}, volume 125 of {\em Mathematical
  Surveys and Monographs}.
\newblock American Mathematical Society, 2006.

\bibitem[CS83]{craigsimon}
Walter Craig and Barry Simon.
\newblock Subharmonicity of the {L}yaponov index.
\newblock {\em Duke Math. J.}, 50:551--560, 1983.

\bibitem[DSS02]{damanik02}
David Damanik, Robert Sims, and G{{\"u}}nter Stolz.
\newblock Localization for one-dimensional, continuum, {B}ernoulli-{A}nderson
  models.
\newblock {\em Duke Math. J.}, 114(1):59--100, 2002.

\bibitem[F{\"u}r63]{Furstenberg1963}
Harry F{\"u}rstenberg.
\newblock Noncommuting random products.
\newblock {\em Trans. Amer. Math. Soc.}, 108:377--428, 1963.

\bibitem[GK01]{Germinet:2001p4369}
Francois Germinet and Abel Klein.
\newblock Bootstrap multiscale analysis and localization in random media.
\newblock {\em Comm. Math. Phys.}, 222(2):415--448, 2001.

\bibitem[Gol85]{ragethm}
Jerome~A. Goldstein.
\newblock {Bound states and scattered states for contraction semigroups.}
\newblock {\em Acta Appl. Math.}, 4:93--98, 1985.

\bibitem[HJS09]{Hamza2009}
Eman Hamza, Alain Joye, and G{\"u}nter Stolz.
\newblock {Dynamical Localization for Unitary Anderson Models}.
\newblock {\em Math. Phys. Anal. Geom.}, 12(4):381--444, 2009.

\bibitem[JM10]{Joye2010}
Alain Joye and Marco Merkli.
\newblock {Dynamical Localization of Quantum Walks in Random Environments}.
\newblock arXiv:1004.4130, 2010.

\bibitem[Joy04]{Joye2004}
Alain Joye.
\newblock {Density of States and Thouless Formula for Random Unitary Band
  Matrices}.
\newblock {\em Annales Henri Poincare}, 5:347--379, 2004.

\bibitem[Kem03]{Kempe2003}
Julia Kempe.
\newblock Quantum random walks - an introductory overview.
\newblock {\em Contem. Phys.}, 44:307, 2003.

\bibitem[Kem05]{Kempe2005}
Julia Kempe.
\newblock Quantum random walks hit exponentially faster.
\newblock {\em Probab. Theory Rel.}, 133(2):215--235, 2005.

\bibitem[KFC{\etalchar{+}}09]{karski-2009-325}
Michal Karski, Leonid Forster, Jai-Min Choi, Andreas Steffen, Wolfgang Alt,
  Dieter Meschede, and Artur Widera.
\newblock Quantum walk in position space with single optically trapped atoms.
\newblock {\em Science}, 325:174, 2009.

\bibitem[Kir07]{Kirsch:2007bf}
Werner Kirsch.
\newblock An invitation to random {S}chr\"odinger operators.
\newblock {\em Panoramas et Syntheses}, 25:1--119, 2007.

\bibitem[Kle08]{pre05533269}
Abel Klein.
\newblock {Multiscale analysis and localization of random operators.}
\newblock {\em Panoramas et Synth\`eses}, 25:221--259, 2008.

\bibitem[Kon09a]{Konno2009a}
Norio Konno.
\newblock Localization of an inhomogeneous discrete-time quantum walk on the
  line.
\newblock arXiv:0908.2213, 2009.

\bibitem[Kon09b]{Konno2009}
Norio Konno.
\newblock One-dimensional discrete-time quantum walks on random environments.
\newblock {\em Quant. Inf. Proc.}, 8(5):387--399, 2009.

\bibitem[LS09]{Linden2009}
Noah Linden and James Sharam.
\newblock Inhomogeneous quantum walks.
\newblock {\em Phys. Rev. A}, 80(5):052327, 2009.

\bibitem[MR95]{Motwani1995}
Rajeev Motwani and Prabhakar Raghavan.
\newblock {\em {Randomized Algorithms}}.
\newblock Cambridge University Press, 1995.

\bibitem[Pol96]{Poltoratski:1996p4363}
Alexei~G. Poltoratski.
\newblock On the distributions of boundary values of {C}auchy integrals.
\newblock {\em Proc. Amer. Math. Soc.}, 124:2455, 1996.

\bibitem[Sim05]{traceideals}
Barry Simon.
\newblock {\em Trace ideals and their applications}, volume 120 of {\em
  Mathematical Surveys and Monographs}.
\newblock American Mathematical Society, Providence, RI, 2005.

\bibitem[SK10]{shikano}
Yutaka Shikano and Hosho Katsura.
\newblock Localization and fractality in inhomogeneous quantum walks with
  self-duality.
\newblock {\em Phys. Rev. E}, 82(3):031122, 2010.

\bibitem[SMS{\etalchar{+}}09]{Schmitz2009}
Hector Schmitz, Robert Matjeschk, Christian Schneider, Jan Glueckert, Martin
  Enderlein, Thomas Huber, and Tobias Schaetz.
\newblock Quantum walk of a trapped ion in phase space.
\newblock {\em Phys. Rev. Lett.}, 103(9):090504, 2009.

\bibitem[ST85]{simontempleineq}
Barry Simon and Michael Taylor.
\newblock Harmonic analysis on {\rm sl}(2,{\bf r}) and smoothness of the
  density of states in the one-dimensional {A}nderson model.
\newblock {\em Comm. Math. Phys.}, 101(1), 1985.

\bibitem[Tes09]{teschl2009}
Gerald Teschl.
\newblock {\em {Mathematical Methods in Quantum Mechanics: With Applications to
  Schr\"odinger Operators}}.
\newblock Amer. Math. Soc., 2009.

\end{thebibliography}

\end{document}